\PassOptionsToPackage{unicode}{hyperref}
\PassOptionsToPackage{hyphens}{url}
\PassOptionsToPackage{dvipsnames,svgnames,x11names}{xcolor}
\documentclass[12pt]{article}

\usepackage{amsmath,amssymb}
\usepackage{iftex}
\ifPDFTeX
  \usepackage[T1]{fontenc}
  \usepackage[utf8]{inputenc}
  \usepackage{textcomp} 
\else 
  \usepackage{unicode-math}
  \defaultfontfeatures{Scale=MatchLowercase}
  \defaultfontfeatures[\rmfamily]{Ligatures=TeX,Scale=1}
\fi
\usepackage{lmodern}
\ifPDFTeX\else  
\fi
\IfFileExists{upquote.sty}{\usepackage{upquote}}{}
\IfFileExists{microtype.sty}{
  \usepackage[]{microtype}
  \UseMicrotypeSet[protrusion]{basicmath} 
}{}
\makeatletter
\@ifundefined{KOMAClassName}{
  \IfFileExists{parskip.sty}{%
    \usepackage{parskip}
  }{
    \setlength{\parindent}{0pt}
    \setlength{\parskip}{6pt plus 2pt minus 1pt}}
}{
  \KOMAoptions{parskip=half}}
\makeatother
\usepackage{xcolor}
\makeatletter
\ifx\paragraph\undefined\else
  \let\oldparagraph\paragraph
  \renewcommand{\paragraph}{
    \@ifstar
      \xxxParagraphStar
      \xxxParagraphNoStar
  }
  \newcommand{\xxxParagraphStar}[1]{\oldparagraph*{#1}\mbox{}}
  \newcommand{\xxxParagraphNoStar}[1]{\oldparagraph{#1}\mbox{}}
\fi
\ifx\subparagraph\undefined\else
  \let\oldsubparagraph\subparagraph
  \renewcommand{\subparagraph}{
    \@ifstar
      \xxxSubParagraphStar
      \xxxSubParagraphNoStar
  }
  \newcommand{\xxxSubParagraphStar}[1]{\oldsubparagraph*{#1}\mbox{}}
  \newcommand{\xxxSubParagraphNoStar}[1]{\oldsubparagraph{#1}\mbox{}}
\fi
\makeatother

\usepackage{longtable,booktabs,array}
\usepackage{calc} 
\usepackage{etoolbox}
\makeatletter
\patchcmd\longtable{\par}{\if@noskipsec\mbox{}\fi\par}{}{}
\makeatother
\IfFileExists{footnotehyper.sty}{\usepackage{footnotehyper}}{\usepackage{footnote}}
\makesavenoteenv{longtable}
\usepackage{graphicx}
\makeatletter
\def\maxwidth{\ifdim\Gin@nat@width>\linewidth\linewidth\else\Gin@nat@width\fi}
\def\maxheight{\ifdim\Gin@nat@height>\textheight\textheight\else\Gin@nat@height\fi}
\makeatother
\setkeys{Gin}{width=\maxwidth,height=\maxheight,keepaspectratio}
\makeatletter
\def\fps@figure{htbp}
\makeatother

\makeatletter
\@ifpackageloaded{caption}{}{\usepackage{caption}}
\AtBeginDocument{%
\ifdefined\contentsname
  \renewcommand*\contentsname{Table of contents}
\else
  \newcommand\contentsname{Table of contents}
\fi
\ifdefined\listfigurename
  \renewcommand*\listfigurename{List of Figures}
\else
  \newcommand\listfigurename{List of Figures}
\fi
\ifdefined\listtablename
  \renewcommand*\listtablename{List of Tables}
\else
  \newcommand\listtablename{List of Tables}
\fi
\ifdefined\figurename
  \renewcommand*\figurename{Figure}
\else
  \newcommand\figurename{Figure}
\fi
\ifdefined\tablename
  \renewcommand*\tablename{Table}
\else
  \newcommand\tablename{Table}
\fi
}
\@ifpackageloaded{float}{}{\usepackage{float}}
\floatstyle{ruled}
\@ifundefined{c@chapter}{\newfloat{codelisting}{h}{lop}}{\newfloat{codelisting}{h}{lop}[chapter]}
\floatname{codelisting}{Listing}

\makeatother
\makeatletter
\@ifpackageloaded{caption}{}{\usepackage{caption}}
\@ifpackageloaded{subcaption}{}{\usepackage{subcaption}}
\makeatother

\ifLuaTeX
  \usepackage{selnolig}  
\fi
\usepackage[]{natbib}
\usepackage{bookmark}

\IfFileExists{xurl.sty}{\usepackage{xurl}}{} 
\hypersetup{
  pdftitle={Title},
  pdfauthor={Author 1; Author 2},
  pdfkeywords={3 to 6 keywords, that do not appear in the title},
  colorlinks=true,
  linkcolor={blue},
  filecolor={Maroon},
  citecolor={Blue},
  urlcolor={Blue},
  pdfcreator={LaTeX via pandoc}}

\newcommand{\anon}{1}

\usepackage{setspace}

\usepackage[margin=1in]{geometry}
\usepackage{graphicx}
\usepackage{amsmath, amssymb, amsfonts}
\usepackage{xr-hyper} 
\usepackage{hyperref}
\usepackage{url}
\usepackage{soul}
\usepackage{color}
\usepackage{cancel}
\usepackage{natbib}
\usepackage{float}
\usepackage[dvipsnames]{xcolor}

\usepackage[normalem]{ulem}
\newcommand{\stkout}[1]{\ifmmode\text{\sout{\ensuremath{#1}}}\else\sout{#1}\fi}

\usepackage{tikz}
\usepackage{tikz-qtree}
\usetikzlibrary{trees}
\usetikzlibrary{shapes.multipart}
\usetikzlibrary{automata,positioning,decorations.pathmorphing}

\usepackage{multirow}      
\usepackage{graphicx}      
\usepackage{array}         
\usepackage{hhline}        
\usepackage{colortbl}      
\usepackage{caption}       
\usepackage{arydshln}         
\newcommand{\thickline}{\noalign{\hrule height 0.8pt} } 

\usepackage{array}
\usepackage{caption}
\usepackage{graphicx}
\usepackage{siunitx}
\usepackage[normalem]{ulem}
\usepackage{colortbl}
\usepackage{multirow}
\usepackage{hhline}
\usepackage{calc}
\usepackage{tabularx}
\usepackage{threeparttable}
\usepackage{wrapfig}
\usepackage{adjustbox}
\usepackage{hyperref}

\allowdisplaybreaks

\DeclareMathOperator{\doo}{do}

\DeclareMathOperator{\pa}{pa}

\DeclareMathOperator{\de}{de}

\DeclareMathOperator{\ch}{ch}

\DeclareMathOperator{\nd}{nd}

\newtheorem{theorem}{Theorem}

\newtheorem{corollary}[theorem]{Corollary}

\newtheorem{example}{Example}

\newtheorem{algorithm}{Algorithm}

\newenvironment{proof}[1][Proof]{\noindent\textbf{#1.} }{\ \rule{0.5em}{0.5em}}

\usepackage{xr}
\usepackage{mathtools}
\usepackage{algorithm}
\usepackage[noend]{algpseudocode}
\usepackage{enumerate}
\usepackage{parskip}
\def\ci{\perp\!\!\!\perp}

\newcommand{\blue}{\textcolor{blue}}

\newcommand{\E}{\mathbb{E}} 
\newcommand{\G}{{\mathcal G}}
\newcommand{\I}{\mathbb{I}}
\newcommand{\T}{\mathbb{T}}
\newcommand{\F}{\mathbb{F}}
\newcommand{\R}{\mathcal{R}}
\newcommand{\C}{\mathcal{C}}
\newcommand{\N}{\mathcal{N}}
\newcommand{\calS}{\mathcal{S}}
\newcommand{\expit}{\mathrm{expit}}  

\usepackage{authblk}

\newcommand{\eqtagtext}[1]{%
  \tag*{\text{\small #1}\ (\theequation)}%
}

\usetikzlibrary{arrows.meta}
\tikzset{
    circled arrow/.style={
        {Circle[open]}-{Circle[open]},
    },
    circled arrow directed/.style={
        -{Circle[open]}-{Triangle[length=3mm, width=3mm]},
    },
     circled arrow onesided/.style={
        {Circle[open]}->,
    }
}

\begin{document}

\def\spacingset#1{\renewcommand{\baselinestretch}%
{#1}\small\normalsize} \spacingset{1}


\if1\anon
{
  \title{\bf Weighting-Based Identification and Estimation in Graphical Models of Missing Data}
  \author{Anna Guo and Razieh Nabi\hspace{.2cm}\\
    Department of Biostatistics and Bioinformatics, Emory University}
  \maketitle
} \fi

\if0\anon
{
  \bigskip
  \bigskip
  \bigskip
  \begin{center}
    {\LARGE\bf Weighting-Based Identification and Estimation in Graphical Models of Missing Data}
\end{center}
  \medskip
} \fi

\bigskip
\begin{abstract}
We propose a constructive algorithm for identifying complete data distributions in graphical models of missing data. The complete data distribution is unrestricted, while the missingness mechanism is assumed to factorize according to a conditional directed acyclic graph. Our approach follows an interventionist perspective in which missingness indicators are treated as variables that can be intervened on.
A central challenge in this setting is that sequences of interventions on missingness indicators may induce and propagate selection bias, so that identification can fail even when a propensity score is invariant to available interventions.  To address this challenge, we introduce a tree-based identification algorithm that explicitly tracks the creation and propagation of selection bias and determines whether it can be avoided through admissible intervention strategies. The resulting tree provides both a diagnostic and a constructive characterization of identifiability under a given missingness mechanism. 
Building on these results, we develop recursive inverse probability weighting procedures that mirror the intervention logic of the identification algorithm, yielding valid estimating equations for both the missingness mechanism and functionals of the complete data distribution. Simulation studies and a real-data application illustrate the practical performance of the proposed methods. 
An accompanying \texttt{R} package, \href{https://github.com/annaguo-bios/flexMissing}{\texttt{flexMissing}}, implements all proposed procedures.
\end{abstract}

\noindent%
{\it Keywords:} Missing not at random, selection bias, missing data DAGs, causal inference, inverse probability weighting, estimating equations. 
\vfill

\newpage
\spacingset{1.8} 

\section{Introduction}
\label{sec:intro}

Data analysis across scientific disciplines is frequently complicated by systematically missing observations. While missingness is often assumed to be missing-completely-at-random (MCAR) or missing-at-random (MAR), these assumptions are frequently violated in practice, as the probability of missingness may depend on partially observed or unobserved variables. Such missing-not-at-random (MNAR) mechanisms can lead to biased inference if not properly accounted for \citep{rubin76inference, little2002statistical}.

Despite their prevalence, MNAR mechanisms are difficult to handle because the underlying complete data distribution cannot, in general, be expressed as a function of the observed data distribution. A canonical example is self-censoring or self-masking, where a variable directly influences its own probability of being missing, rendering the target distribution non-identifiable without further assumptions. In non-identified settings, common approaches include imposing parametric or semiparametric structures \citep{wu1988estimation, little2002statistical, wang2014instrumental, sun2018semiparametric, sportisse2020imputation, guo2023sufficient}, conducting sensitivity analyses \citep{rotnitzky1998semiparametric, scharfstein2003generalized, nabi2024semiparametric}, or deriving partial identification bounds \citep{horowitz2000nonparametric, manski2005partial}. 

At the same time, a growing body of work has identified MNAR mechanisms by imposing restrictions on the missingness selection mechanism. Examples include the permutation model \citep{robins97non-a}, the block-conditional MAR model \citep{zhou10block}, itemwise conditionally independent nonresponse and no self-censoring models \citep{sadinle2017itemwise,shpitser16consistent}, and discrete choice models \citep{tchetgen2018discrete}. These models differ in the specific restrictions imposed on the missingness mechanism, while the corresponding identification arguments are agnostic to assumptions on the complete data distribution.

A recent line of work studies missing data models by drawing parallels with causal graphical models with hidden variables \citep{tian02general, shpitser06id, bhattacharya2022semiparametric, richardson2023nested}. In this framework, a directed acyclic graph encodes assumptions about the full law, including both the complete data law, referred to as the target law, and the missingness mechanism \citep{glymour2006using, daniel2012using, martel2013definition, mohan13missing, thoemmes2014cautious, shpitser15missing, bhattacharya19mid, nabi20completeness, mohan2021graphical, nabi2023causal}. Identification is then formulated from an interventionist perspective, in which missingness indicators are treated as intervention nodes. Target law identification then is expressed as a sequence of reweighting operations that adjust for selection bias induced by conditioning on observed cases. 

\cite{bhattacharya19mid} showed that causal identification strategies are insufficient for missing data models, because they do not adequately account for selection bias induced by interventions. They instead characterized identification via admissible sequences of indicator interventions defined by a partial order, rather than the total order common in causal settings. While this framework clarifies key identification requirements and structural barriers due to selection bias, it offers no constructive method for assessing identifiability in a given model or for developing estimation and inference procedures.

In this paper, we develop a general and tractable framework for identification, estimation, and inference in a broad class of missing data models. We impose no restrictions on the target law and instead model the missingness mechanism using a conditional directed acyclic graph. Our approach yields a constructive identification strategy that explicitly accounts for selection bias induced by conditioning on observed cases and clarifies when such bias can be avoided through admissible intervention strategies on missingness indicators.

Our central contribution is a tree-based identification algorithm that tracks the creation and propagation of selection bias across sequences of interventions on missingness indicators. Rather than returning a binary identifiability verdict, the algorithm produces explicit identification rules for propensity scores, potentially on restricted evaluation sets, and determines whether these rules jointly suffice to identify the target distribution. When target law identification is not achievable, the same tree structure can be queried to specify which subsets of propensity scores, together with their required evaluations, are sufficient to identify and estimate particular functionals of the target law.  

Building on these identification results, we develop general estimation and inference procedures that mirror the intervention logic of the identification trees. For functionals identified by the identifiable components of the missingness mechanism, we propose recursive inverse probability weighting estimators that respect the same admissibility constraints  required for identification and yield valid estimating equations for both the missingness mechanism and target functionals. We establish large-sample properties and demonstrate practical performance through a real-data application and simulation studies, including empirical comparisons with classical approaches such as the EM algorithm and multiple imputation. 

The paper is organized as follows. Section~\ref{sec:prelim} introduces notation and the graphical framework for missing data models. Section~\ref{sec:hoops} provides intuition for identification through examples. The identification algorithm is presented in Section~\ref{sec:identification}, followed by weighting based estimation and inference in Section~\ref{sec:estimation}. Simulation studies appear in Section~\ref{sec:sims}, and Section~\ref{sec:application} presents an application to survey data. Section~\ref{sec:conc} concludes, with proofs in the appendix.

\section{Problem Setup and Notation}
\label{sec:prelim}

We begin by introducing notation and  missing data assumptions. We omit variables that are always observed; all results extend directly to settings with fully observed covariates. 

Let $X = (X_1,\ldots,X_K)^T$ be a random vector with joint distribution $p(X)$, termed the \textit{target law}, belonging to a model $\mathcal{M}_X$. Let $R = (R_1,\ldots,R_K)^T$ denote binary missingness indicators, where $R_k = 1$ if $X_k$ is observed and $R_k = 0$ if $X_k$ is missing, with conditional distribution $p(R \,|\, X)$, termed the \textit{missingness mechanism}, in model $\mathcal{M}_{R \,|\, X}$. The joint distribution $p(X,R)$, termed the \textit{full law}, lies in the product model $\mathcal{M} = \mathcal{M}_X \otimes \mathcal{M}_{R \,|\, X}$. The observed data are a coarsened version of $X$: define $X^* = (X_1^*,\ldots,X_K^*)^T$ by $X_k^* = X_k$ if $R_k = 1$ and $X_k^* = \text{``?''}$ otherwise. The distribution $p(X^*,R)$ is called the \textit{observed data law}. Following \citep{nabi2023causal}, we view each $X_k$ as the counterfactual value of $X_k^*$ had it been fully observed, or equivalently under the intervention $R_k = 1$. 

We consider missing data models $\mathcal{M} = \mathcal{M}_X \otimes \mathcal{M}_{R \,|\, X}$, where the target model $\mathcal{M}_X$ is unrestricted and the missingness mechanism $\mathcal{M}_{R \,|\, X}$ is constrained by graphical assumptions. In particular, we assume $p(R \,|\, X)$ factorizes according to a conditional directed acyclic graph (DAG) $\mathcal{G}$, that is, $p(R \,|\, X) = \prod_{R_k \in R} p(R_k \,|\, \pa_{\mathcal{G}}(R_k))$, where $\pa_{\mathcal{G}}(R_k)$ denotes the parents of $R_k$ in $\mathcal{G}$. No graphical assumptions are imposed on the target law $p(X)$; $\mathcal{M}_X$ may correspond to a hidden-variable DAG, a Markov random field, or any other model that need not admit a graphical representation. Thus, we use $\G$ solely to encode restrictions on $p(R \,|\, X)$.

The DAG $\mathcal{G}$ follows standard missing data DAG (mDAG) conventions: it is acyclic and contains no edges from missingness indicators $R$ to variables in $X$ \citep{mohan13missing}. For each $R_k \in R$, let $\de_\G(R_k)$ denote its descendants in $\mathcal{G}$, including $R_k$, and define the non-descendants as $\nd_\G(R_k) = X \cup R \setminus \de_\G(R_k)$. Standard graphical d-separation rules apply \citep{pearl09causality}. All mDAGs in the remainder of the paper are assumed to satisfy these conventions unless stated otherwise.

This paper focuses on identification and inference in graphical models of missing data, without committing to a specific choice of the functional of the target law. We analyze identifiability by characterizing which components of the missingness mechanism are identifiable from the observed data law $p(X^*, R)$ under $\mathcal M$. In some settings this yields identification of the target law $p(X)$, while in others it suffices to identify particular functionals, denoted by $\theta(p(X))$. We also outline how such identification can be used to conduct inference on $\theta$.

The central role of the missingness mechanism follows from the identity $p(X) = p(X, R = 1) \, / \, p(R = 1 \,|\, X)$, which shows that recovery of the target law, when possible, depends on identification of $p(R = 1 \,|\, X)$. This representation makes explicit the connection between identification in missing data models and inverse probability weighting, in which the target law is expressed as a reweighted version of the complete-case distribution. 

Under the graphical restriction on the missingness mechanism, the conditional distribution $p(R = 1 \,|\, X)$ factorizes into a product of conditional probabilities 

\vspace{-1.35cm}
\begin{align}
   p(X) = \displaystyle p(X, R=1) \, \big/ \, \big\{ \displaystyle \prod_{R_k \in R}p(R_k \mid \pa_\G(R_k)) \Big{\rvert}_{R=1} \big\}, 
   \label{eq:Bayes}
\end{align}%
\vspace{-1.35cm}

where $p(.)\vert_{R_j=1}$ denotes the evaluation of $p(.)$ at $R_j=1$. We define the propensity score of $R_k$ as $\pi_k(\pa_\G(R_k)) = p(R_k=1 \,|\, \pa_\G(R_k))$. Our identification strategy centers on when these scores are identifiable and how they enable recovery of the target law or functionals thereof. 

We assume $\pi_k(\pa_\G(R_k)) > \sigma > 0$ a.s. for a fixed positive constant $\sigma$, and for all $R_k \in R$. This condition ensures nonparametric identification of the target law and its smooth functionals, and finite asymptotic variance of the proposed weighting estimators \citep{robins2000sensitivity}.

\section{Building Blocks for Propensity Score Identification}
\label{sec:hoops}

Identification of the propensity score $\pi_k$ requires expressing this conditional probability as a unique functional of the observed data law $p(X^*, R)$. The main difficulty arises when $\pa_\G(R_k)$ contains variables in $X$ that are subject to missingness.

When a parent $X_j \in \pa_\G(R_k)$ is subject to missingness, any representation of $\pi_k$ in terms of observed data must replace $X_j$ by its observed counterpart on rows where $R_j = 1$, thereby introducing selection through the corresponding missingness indicators. To keep track of this bookkeeping requirement, for each $R_k \in R$ we define the \emph{counterfactual-induced selection set}

\vspace{-1.5cm}
\refstepcounter{equation}\label{eq:counterfactual_selection_set}
\begin{align}
   \mathcal S_k^x = \{ R_j \in R : X_j \in \pa_\G(R_k) \}. \eqtagtext{(counterfactual-induced selection set for $R_k$)}
\end{align}%
\vspace{-1.5cm}

The set $\mathcal S_k^x$ records indicators that become relevant solely because counterfactual variables appear among the parents of $R_k$. Its role is descriptive, not prescriptive: the presence of $R_j \in \mathcal S_k^x$ indicates that substituting $X_j$ with its observed counterpart is unavoidable when forming candidate representations of $\pi_k$, but does not determine whether such substitutions are valid for identification. Whether this substitution can be justified depends on the assumptions encoded in the missingness mechanism.

For identification of the target law, it is not necessary that $\pi_k$ be identifiable as a full conditional distribution. It suffices that $\pi_k$ be identifiable at the evaluation under which all missingness indicators equal one, as required by \eqref{eq:Bayes}. We therefore define the \emph{indicator-induced selection set} $S_k^r$ as the set of indicators whose evaluation at one is required in order for $\pi_k$ to be identified at this evaluation. By construction, $\mathcal S_k^r \subseteq R \cap \pa_\G(R_k)$.

The sets $\mathcal S_k^x$ and $\mathcal S_k^r$ capture conceptually distinct, though potentially overlapping, sources of selection in identifying $\pi_k$. The former records indicators implicated by counterfactual substitution due to missing parents of $R_k$, while the latter records indicators whose evaluation at one is required for $\pi_k$ to be usable in recovering the target law. These sets play a common role in identification; both determine which missingness indicators may induce selection in candidate representations of $\pi_k$. It is therefore convenient to collect them into a single object that summarizes all selection relevant for identifying the propensity score of $R_k$. For each $R_k \in R$, we thus define the overall \emph{selection set}

\vspace{-1.75cm}
\refstepcounter{equation}\label{eq:selection_set}
\begin{align}
\mathcal S_k = \mathcal S_k^x \cup \mathcal S_k^r.  \eqtagtext{(selection set for $R_k$)}
\end{align}%
\vspace{-1.75cm}

A central obstacle in identifying $\pi_k$ arises when some indicators in $\mathcal S_k^x$ are descendants of $R_k$. Presence of such indicators obstructs identification through purely associational arguments. For each $R_k \in R$, we collect these into the \emph{problematic set}

\vspace{-1.75cm}
\refstepcounter{equation}\label{eq:R_problem}
\begin{align}
\mathcal R_k^p = \mathcal S_k^x \cap \de_\G(R_k). \eqtagtext{(problematic set for $R_k$)}
\end{align}%
\vspace{-1.75cm}

\subsection{Associational and causal irrelevancy}
\label{subsec:assoc}

Our identification arguments rely on two distinct notions of irrelevancy. The first is \emph{associational irrelevancy}, which follows from the local Markov property of mDAGs. In particular, for each $R_k \in R$, $R_k \ci \nd_\G(R_k) \setminus \pa_\G(R_k) \,|\, \pa_\G(R_k)$.

\begin{example}(Associational irrelevancy)
Consider an mDAG in which $X_2 \rightarrow R_1$ and $X_1 \rightarrow R_2$, with no edges between $R_1$ and $R_2$. Here, $\pi_1(\pa_\G(R_1)) = p(R_1 = 1 \,|\, X_2)$ and the counterfactual-induced selection set is $S_1^x = \{R_2\}$. Since $R_2$ is a non-descendant non-parent of $R_1$, the local Markov property implies $R_1 \ci R_2 \,|\, X_2$. Consequently, $\pi_1$ is fully identified via $p(R_1 = 1 \,|\, X_2, R_2 = 1)$, and $S_1^r = \emptyset$. An analogous argument applies to identification of $\pi_2$.
\end{example}

More generally, associational irrelevancy allows us to append $R_j = 1$ to the conditioning set of $\pi_k$ whenever $R_j \in \mathcal S_k^x \cap \{\nd_\G(R_k) \setminus \pa_\G(R_k)\}$, yielding

\vspace{-1.5cm}
\begin{align}
	\pi_k(\pa_\G(R_k)) \coloneqq p(R_k = 1 \mid \pa_\G(R_k))
	= p(R_k = 1 \mid \pa_\G(R_k), R_j = 1).
	\label{eq:id_association}
\end{align}%
\vspace{-1.5cm}

In contrast, when $R_j \in \mathcal S_k^x \cap \pa_\G(R_k)$, the propensity score $\pi_k$ is identifiable only when $R_j = 1$, and thus $R_j$ necessarily belongs to the indicator-induced selection set $\mathcal S_k^r$. Such parent indicators form a collider structure $X_j \rightarrow R_k \leftarrow R_j$, termed a \emph{colluder} \citep{bhattacharya19mid}. Although $\pi_k$ is not fully identifiable in this case, the target law may still be identifiable, since only $\pi_k$ evaluated at $R = 1$ is required for target law identification.  

Associational irrelevancy implies that if the problematic set $\mathcal R_k^p$ is empty, all selection relevant for identifying $\pi_k$ arises from parent indicators, so $\pi_k$ evaluated at $\mathcal S_k^r = 1$ is identifiable. If instead $\mathcal R_k^p \neq \emptyset$, associational irrelevancy fails for indicators in $\mathcal S_k^x \cap \de_\G(R_k)$, since $R_k \not\ci \de_\G(R_k) \,|\, \pa_\G(R_k)$. However, identification may still be achieved via \emph{causal irrelevancy}. 

The key observation underlying causal irrelevancy is that the propensity score $\pi_k$ is invariant to interventions on missingness indicators other than $R_k$ itself. This invariance property, also known as autonomy, modularity, and stability \citep{haavelmo1944probability, spirtes01causation, dawid2010identifying, pearl09causality}, allows us to search for post-intervention distributions in which problematic descendant relationships with $R_k$ are broken. 

We define an intervention on $R_j \in R$ as an operation that fixes $R_j=1$. At the graphical level, this removes all incoming edges into $R_j$ and replaces $X_j$ by its observed counterpart, since $R_j = 1$. At the probabilistic level, this corresponds to a truncated factorization of the full law in which $p(X, R)$ is first evaluated at $R_j = 1$ and then renormalized by the propensity score of $R_j$, yielding a post-intervention distribution denoted by $p(X, R \setminus R_j \,|\, \doo(R_j=1))$. Section~\ref{sec:identification} gives the explicit form of this operator and its use in constructing identifying functionals.

According to causal irrelevancy, for any $R^* \subseteq R \setminus \{R_k\}$,

\vspace{-1.65cm}
\begin{align}
\pi_k(\pa_\G(R_k)) \coloneqq p(R_k = 1 \mid \pa_\G(R_k))
= p(R_k = 1 \mid \pa_\G(R_k), \doo(R^* = 1)),
\label{eq:id_causal}
\end{align}%
\vspace{-1.5cm}

where $\doo(R^* = 1)$ denotes an intervention that sets all indicators in $R^*$ to one. Identification of $\pi_k$ via causal irrelevancy proceeds by intervening on indicators in the problematic set $\mathcal R_k^p$ in order to sever descendant relationships with $R_k$.

\begin{example}(Causal irrelevancy)
Suppose $X_1 \rightarrow R_2$ and $R_2 \rightarrow R_1$, so that $S_2^x = \{R_1\}$ and $R_1$ is a descendant of $R_2$. Associational irrelevancy fails, but since $\pi_2$ is invariant to interventions on $R_1$ we can write $\pi_2$ as $p(R_2 = 1 \,|\, X_1) = p(R_2 = 1 \,|\, X_1, \doo(R_1 = 1))$. The post-intervention distribution induced by $\doo(R_1=1)$ is $p(X_1, X_2, R_2 \,|\, \doo(R_1=1)) \coloneqq p(X_1, X_2, R_1=1, R_2)/p(R_1 = 1 \,|\, R_2)$. Using simple probability rules, $\pi_{2}$ is thus identified from the identified margin $p(X_1, R_2 \,|\, \doo(R_1=1)) = p(X_1, R_1=1, R_2)/p(R_1 = 1 \,|\, R_2)$.  
\end{example}

\subsection{Selection bias and propagation}
\label{subsec:selection_bias}

Identification via causal irrelevancy hinges on the ability to intervene on indicators in the problematic set $\mathcal R_k^p$ without inducing selection that obstructs identification of $\pi_k$. While such interventions break descendant relationships with $R_k$, each intervention may itself induce selection through the selection set of the intervened indicator. Selection on non-descendants of $R_k$ is typically benign: if such indicators lie in $\pa_\G(R_k)$, they can be incorporated into the indicator-induced selection set $\mathcal S_k^r$, while selection on non-parent non-descendants is independent of $R_k$ by the local Markov property. Difficulties arise when an intervention induces selection on $R_k$ or its descendants, as such selection may render $\pi_k$ unidentifiable, and in some settings this induced selection can be accommodated by subsequent interventions, while in others it is unavoidable and obstructs identification altogether.

When multiple interventions are performed sequentially, selection induced by earlier interventions may propagate to later ones. Specifically, if $R_j$ is intervened on prior to an intervention on $R_k$, the portion of the selection induced by $\doo(R_j=1)$ that propagates through $R_k$ is 

\vspace{-1.65cm}
\refstepcounter{equation}\label{eq:propa}
\begin{align}
    \mathcal S_{j \, \downarrow \, k} \coloneqq \mathcal S_j \cap \pa_\G(R_k), \eqtagtext{(selection propagation rule)}
\end{align}
\vspace{-1.75cm}

which we refer to as the \emph{selection propagation rule}. Propagation of selection through parent relationships can obstruct identification, while absence of such propagation permits sequential identification, as we illustrate via examples below. 

\begin{example}(Admissible and inadmissible descendant interventions)
In Figure~\ref{fig:seq_select}(a), identification of $\pi_3(X_1) \coloneqq p(R_3=1 \,|\, X_1)$ requires addressing the problematic set $\mathcal R_3^p=\{R_1\}$. Although $\pi_3$ is invariant to interventions on either $R_1$ or $R_2$, intervening on $R_1$ induces selection on $R_3$, since $\mathcal S_1=\{R_3\}$, yielding a post-intervention distribution available only at $R_3=1$ and hence insufficient for identification. In contrast, intervening on $R_2$ induces no selection on $R_3$, as $\mathcal S_2=\emptyset$, and the resulting post-intervention distribution admits a margin that identifies $\pi_3$. This example illustrates that although multiple descendant interventions may leave $\pi_k$ invariant, only a subset yield admissible post-intervention distributions for identification. 
If $X_3 \to R_1$ is replaced by $X_2 \to R_1$, then intervening on $R_1$ induces selection on $R_2$, which cannot be resolved by conditional independence alone. Identification of $\pi_{3}$ therefore requires intervening on $R_2$, either alone or jointly with $R_1$.
\end{example}

\begin{figure*}[t]
	\begin{center}
		\scalebox{0.7}{
            \begin{tikzpicture}[>=stealth, node distance=1.2cm, 
				decoration = {snake, pre length=3pt,post length=7pt,},
				every text node part/.style={align=center}]
				\tikzstyle{format} = [thick, circle, minimum size=1.0mm, inner sep=0pt]
                \begin{scope}[xshift=0cm]
					\path[->, thick]
					node[format] (x11) {$X_1$}
					node[format, right of=x11, xshift=0.5cm] (x21) {$X_2$}
					node[format, right of=x21, xshift=0.5cm] (x31) {$X_3$}
					node[format, below of=x11] (r1) {$R_1$}
					node[format, below of=x21] (r2) {$R_2$}
					node[format, below of=x31] (r3) {$R_3$}
					
					(x11) edge[blue, dashed, -, bend left=0] (x21)
					(x21) edge[blue, dashed, -, bend left=0] (x31)
					(x11) edge[blue, dashed, -, bend left=20] (x31)

					(x11) edge[blue] (r3)
                    (x31) edge[blue] (r1)
					
					(r2) edge[blue] (r1)
					(r3) edge[blue] (r2)

                    node[format, below of=r2, xshift=.cm, yshift=0.2cm] (a) {{\small $(a)$}}  ;
				\end{scope}
                
				\begin{scope}[xshift=4.75cm]
					\path[->, thick]
					node[format] (x11) {$X_1$}
					node[format, right of=x11, xshift=0.5cm] (x21) {$X_2$}
					node[format, right of=x21, xshift=0.5cm] (x31) {$X_3$}
					node[format, right of=x31, xshift=0.5cm] (x41) {$X_4$}
					node[format, below of=x11] (r1) {$R_1$}
					node[format, below of=x21] (r2) {$R_2$}
					node[format, below of=x31] (r3) {$R_3$}
					node[format, below of=x41] (r4) {$R_4$}
					
					(x11) edge[blue, dashed, -, bend left=0] (x21)
					(x21) edge[blue, dashed, -, bend left=0] (x31)
					(x31) edge[blue, dashed, -, bend left=0] (x41)
					(x11) edge[blue, dashed, -, bend left=20] (x31)
					(x11) edge[blue, dashed, -, bend left=30] (x41)
					(x21) edge[blue, dashed, -, bend left=20] (x41)
					
					(x11) edge[blue] (r2)
					(x11) edge[blue] (r4)
					(x21) edge[blue] (r4)
                    (x41) edge[blue] (r1)
                    (x41) edge[blue] (r3)
					
					(r2) edge[blue] (r1)
					(r3) edge[blue] (r2)
					(r4) edge[blue] (r3)
					
					
					
                    node[format, below of=r2, xshift=1.cm, yshift=0.2cm] (b) {{\small $(b)$}}  ;
				\end{scope}
				
				\begin{scope}[xshift=11cm]
					\path[->, thick]
					node[format] (x11) {$X_1$}
					node[format, right of=x11, xshift=0.5cm] (x21) {$X_2$}
					node[format, right of=x21, xshift=0.5cm] (x31) {$X_3$}
					node[format, right of=x31, xshift=0.5cm] (x41) {$X_4$}
					node[format, below of=x11] (r1) {$R_1$}
					node[format, below of=x21] (r2) {$R_2$}
					node[format, below of=x31] (r3) {$R_3$}
					node[format, below of=x41] (r4) {$R_4$}
					
					(x11) edge[blue, dashed, -, bend left=0] (x21)
					(x21) edge[blue, dashed, -, bend left=0] (x31)
					(x31) edge[blue, dashed, -, bend left=0] (x41)
					(x11) edge[blue, dashed, -, bend left=20] (x31)
					(x11) edge[blue, dashed, -, bend left=30] (x41)
					(x21) edge[blue, dashed, -, bend left=20] (x41)
					
					(x11) edge[blue] (r2)
					(x11) edge[blue] (r4)
					(x21) edge[blue] (r4)
                    (x41) edge[blue] (r1)
                    (x41) edge[blue] (r3)
					
					(r2) edge[blue] (r1)
					(r3) edge[blue] (r2)
					(r4) edge[blue] (r3)
                    (r4) edge[blue, bend left=20] (r2)
					
					
					
                    node[format, below of=r2, xshift=1.cm, yshift=0.2cm] (c) {{\small $(c)$}}  ;
				\end{scope}

                \begin{scope}[xshift=17.5cm, yshift=0cm]
					\path[->, thick]
					node[format] (x11) {$X_1$}
					node[format, right of=x11, xshift=0.5cm] (x21) {$X_2$}
					node[format, right of=x21, xshift=0.5cm] (x31) {$X_3$}
					node[format, right of=x31, xshift=0.5cm] (x41) {$X_4$}
					node[format, below of=x11] (r1) {$R_1$}
					node[format, below of=x21] (r2) {$R_2$}
					node[format, below of=x31] (r3) {$R_3$}
					node[format, below of=x41] (r4) {$R_4$}
					
					(r4) edge[blue] (r3)
					(r3) edge[blue] (r2)
					(r2) edge[blue] (r1)
					(r4) edge[blue, bend left=20] (r1)
					(x21) edge[blue] (r4)
					(x11) edge[blue] (r3)
					(x41) edge[blue] (r2)
					
					(x11) edge[blue, dashed, -] (x21)
					(x21) edge[blue, dashed, -] (x31)
					(x31) edge[blue, dashed, -] (x41)
					(x11) edge[blue, dashed, -, bend left=20] (x31)
					(x11) edge[blue, dashed, -, bend left=30] (x41)
					(x21) edge[blue, dashed, -, bend left=20] (x41)
					
					
					node[format, below of=r2, xshift=1.cm, yshift=0.2cm] (d) {{\small $(d)$}}  ;
				\end{scope}

			\end{tikzpicture}
		}
	\end{center}
	\caption{mDAGs illustrating selection behavior under interventions: (a) (in)admissible  interventions; (b) non-propagating selection; (c) propagating selection; (d) identification may require interventions on descendants outside the causal path between $R_k$ and $R_j \in \mathcal R_k^p$. }
	\label{fig:seq_select}
\end{figure*}
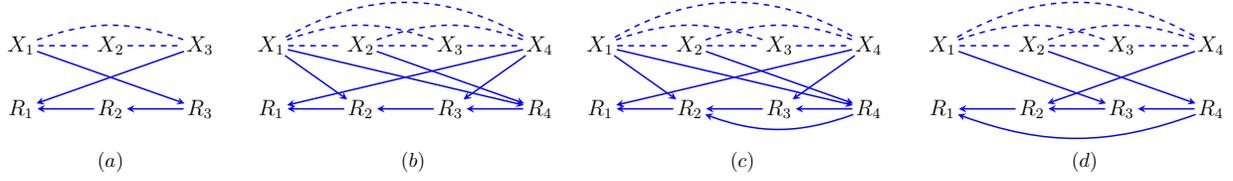

The above example shows that selection from an inadmissible intervention can sometimes be repaired by alternative or additional interventions, though not always. 

\begin{example}(Unavoidable selection)
Consider an mDAG with $X_1 \rightarrow R_2$, $X_2 \rightarrow R_1$, and $R_2 \rightarrow R_1$. Identifying $\pi_2(X_1) \coloneqq p(R_2 = 1 \,|\, X_1)$ involves the problematic set $\mathcal R_2^p = \{R_1\}$. Although $\pi_2$ is invariant to intervening on $R_1$, such an intervention induces selection on $R_2$ ($\mathcal S_1 = \{R_2\}$), leaving no post-intervention distribution that identifies $\pi_2$. This unavoidable selection bias blocks identification, as shown in \citep{nabi2022testability, guo2023sufficient}.
\end{example} 

Selection induced by one intervention may also constrain subsequent intervention choices through propagation. Whether such propagation occurs depends on the graphical structure.

\begin{example}(Sequential interventions and propagation)
In Figure~\ref{fig:seq_select}(b), $\pi_2$ is identified by intervening on $R_1$, which induces selection on $R_4$. This selection does not obstruct identification of $\pi_2$, since $R_4$ is a non-descendant of $R_2$. Crucially, this induced selection does not propagate through the subsequent intervention on $R_2$, as $\mathcal S_{1 \, \downarrow \, 2} = \emptyset$. As a result, intervening on $R_2$ yields an admissible post-intervention distribution from which $\pi_4$ is identified. In contrast, in Figure~\ref{fig:seq_select}(c), an intervention on $R_1$ again induces selection on $R_4$, but now $R_4 \in \pa_\G(R_2)$. Consequently, selection propagates according to $\mathcal S_{1 \, \downarrow \, 2} = \{R_4\}$. The subsequent intervention on $R_2$ then induces selection on the target indicator itself, and identification of $\pi_4$ fails.
\end{example} 

Identification may require intervening on descendants of $R_k$ that lie outside causal paths from $R_k$ to $\mathcal R_k^p$. Moreover, interventions used to identify one propensity score may constrain others due to induced and propagating selection.

\begin{example}(Interventions beyond causal pathways)
\label{ex:beyond_paths}
In Figure~\ref{fig:seq_select}(d), $\pi_3$ admits multiple admissible strategies. Intervening on $R_1$ alone separates $R_3$ from its problematic set and identifies $\pi_3$ without inducing selection on $R_4$. Intervening on $R_2$ alone or together with $R_2$ also identifies $\pi_3$, but only at $\mathcal S_3^r=\{R_4\}$, which remains
sufficient for target law identification but introduces additional selection. 
For $\pi_4$, $\mathcal R_4^p=\{R_2\}$ and $\pi_4$ is invariant to interventions on $R_3$. However, intervening on $R_3$ alone induces selection on $R_1$, which remains a descendant of $R_4$ in $p(\cdot \,|\, \doo(R_3=1))$. Thus identifying $\pi_4$ requires intervention on both $R_1$ and $R_3$. Consequently, $\pi_3$ must be identified via intervention only on $R_1$; intervening on $R_2$ instead induces selection on $R_4$ that propagates through $R_3$ and obstructs identification of $\pi_4$.
\end{example}

Selection bias is central to identifiability under causal irrelevancy. Identification requires a post-intervention distribution with no selection on $R_k$ or its non-intervened descendants and with $R_k \ci R_j$ for all $R_j \in \mathcal R_k^p$. This calls for a constructive strategy that jointly accounts for induced and propagating selection and intervention order.

\section{Identification Algorithm}
\label{sec:identification}

In this section, we formalize the identification logic developed in Section~\ref{sec:hoops} into a constructive procedure, summarized in Algorithm~\ref{alg:ID}. The input to the algorithm is an mDAG $\mathcal G$ encoding restrictions on the missingness mechanism $p(R\,|\,X)$. The algorithm outputs a forest $\mathbb F$ collecting trees $\mathbb T_k$ associated with each missingness indicator $R_k\in R$, together with a set $\mathcal D$ of indicators whose propensity scores cannot be identified.

For each $R_k\in R$, the algorithm attempts to identify the propensity score $\pi_k(\pa_\G(R_k))$ evaluated at an indicator-induced selection set $\mathcal S_k^r=1$. If identification succeeds, the tree $\mathbb T_k$ encodes the sequence of interventions used to obtain an identified representation of $\pi_k$ evaluated at $\mathcal S_k^r=1$. If identification fails, $R_k$ is added to $\mathcal D$, and the corresponding tree records where and why identification breaks down. The target law is identified if and only if $\mathcal D=\emptyset$, since \eqref{eq:Bayes} expresses $p(X)$ in terms of the collection $\{\pi_k\vert_{\mathcal S_k^r=1}\}_{k=1}^K$.

The algorithm processes missingness indicators sequentially according to a \textbf{valid reversed topological order} $\tau$ on the mDAG $\mathcal G$. For each indicator $R_k$, the algorithm attempts to identify  $\pi_k$ by exploiting either associational or causal irrelevancy. Trees constructed in identifying propensity scores of indicators earlier in the order are reused to guide intervention choices for identifying propensity scores of indicators later in the order. 

For a fixed $R_k$, identification is immediate when the problematic set $\mathcal R_k^p$, defined in \eqref{eq:R_problem}, is empty. In this case, associational irrelevancy arguments from Section~\ref{subsec:assoc} suffice, and the tree $\mathbb T_k$ consists only of the root node $R_k$. We therefore focus on the nontrivial case $\mathcal R_k^p\neq\emptyset$.

When $\mathcal R_k^p\neq\emptyset$, identification relies on causal irrelevancy. Conceptually, Algorithm~\ref{alg:ID} searches for a maximal admissible intervention strategy on descendants of $R_k$, that is, a collection of interventions that separates $R_k$ from its problematic set $\mathcal R_k^p$ while avoiding selection that propagates to $R_k$ or obstructs subsequent identification steps. This is achieved by first considering a maximal set of admissible interventions and then pruning those that induce selection obstructing identification.

The remainder of this section describes the algorithm in detail. We first introduce the tree construction procedure used to encode candidate intervention strategies, then formalize the identification status check, and finally describe the pruning operation that removes unnecessary interventions. We conclude by characterizing the identified functionals for propensity scores whenever identification succeeds.

\vspace{-0.25cm}
\paragraph*{Trees and candidate interventions [\blue{tree-construction}].}\hspace{-0.3cm} 
We now describe the tree construction for a fixed indicator $R_k$. The first step identifies which descendants of $R_k$ can be intervened on without immediately obstructing identification of $\pi_k$. 

Recall that indicators whose propensity scores are not identifiable, collected in $\mathcal D$, cannot be intervened on in any admissible strategy. In addition, some descendants of $R_k$ cannot be intervened on because doing so would induce selection on $R_k$ itself. Specifically, these are descendants of $R_k$ that have $X_k$ as a parent, referred to as \emph{colluder descendants} of $R_k$ 

\vspace{-1.5cm}
\begin{align}
\mathcal C^{\text{dir}}_{k,k}
\coloneqq
\{R_j \in \de_\G(R_k) \, : \, X_k \in \pa_\G(R_j)\}.
\eqtagtext{(colluder descendants of $R_k$)}
\end{align}%
\vspace{-1.5cm}

Thus, the initial set of indicators eligible for intervention when identifying $\pi_k$ is 

\vspace{-1.5cm}
\refstepcounter{equation}\label{eq:candidate_intervention}
\begin{align}
R^* \coloneqq \de_\G(R_k)\setminus\{\mathcal C^{\text{dir}}_{k,k},\mathcal D\}.
\eqtagtext{(candidate intervention set for $R_k$)}
\end{align}%
\vspace{-1.65cm}
 
Given a candidate intervention set $R^*$ for $R_k$ defined in \eqref{eq:candidate_intervention}, the algorithm constructs a provisional intervention tree $\mathbb T_k$ by attaching each $R_i\in R^*$ as a child of $R_k$, and augmenting this child with its previously constructed tree $\mathbb T_i$ from $\mathbb F$. The interventions encoded by $\mathbb T_k$ induce selection through two mechanisms. First, counterfactual substitution arising from missing parents of $R_k$, i.e.,  $\mathcal S_k^x$ defined in \eqref{eq:counterfactual_selection_set}. Second, intervening on children of $R_k$ induces additional selection captured by the selection sets  associated with each $R_j\in \ch_{\mathbb T_k}(R_k)$, i.e., $\mathcal S_j$ defined in \eqref{eq:selection_set}. Let $\mathcal T_k \coloneqq \ch_{\mathbb T_k}(R_k)$. We collect all such selection into the \emph{pre-selection set} 

\vspace{-1.5cm}
\refstepcounter{equation}\label{eq:pre-select}
\begin{align}
\widetilde{\mathcal S}_k
\coloneqq
\mathcal S_k^x
\cup
\bigcup_{R_j\in \mathcal T_k} \mathcal S_j. 
\eqtagtext{(pre-selection set for $R_k$)} 
\end{align}%
\vspace{-1.5cm} 

The pre-selection set $\widetilde{\mathcal S}_k$ aggregates all selection imposed while attempting to identify $\pi_k$ and is used exclusively to assess identifiability at the current stage. In contrast, the selection set $\mathcal S_k$ for $R_k$, defined in \eqref{eq:selection_set}, characterizes only the portion of this selection that propagates through $R_k$ when $R_k$ itself is intervened on to identify subsequent propensity scores, according to the propagation rule in \eqref{eq:propa}. By construction, $\mathcal S_k$ is always a subset of $\widetilde{\mathcal S}_k$, reflecting the fact that not all selection induced during identification propagates forward in the algorithm.

\vspace{-0.25cm}
\paragraph*{Identification check and failure modes [\blue{id-status}].}\hspace{-0.3cm} 
 Given $\mathbb T_k$ and $\widetilde{\mathcal S}_k$, the propensity score $\pi_k$ is identifiable under evaluation $\mathcal S_k^r=1$ if 

\vspace{-1.5cm}
\refstepcounter{equation}\label{eq:id_criteria}
\begin{align}
R_k \ci \widetilde{\mathcal S}_k\setminus \pa_\G(R_k)\mid \pa_\G(R_k) \ \text{ in } \ p(\cdot \,|\, \doo(\mathcal T_k=1)),
\eqtagtext{(identification criterion)}
\end{align}%
\vspace{-1.5cm}

where  $p(\cdot \,|\, \doo(\mathcal T_k=1))$ corresponds to the post-intervention distribution induced by intervening on the current children of $R_k$ in $\mathbb T_k$. The set $\widetilde{\mathcal S}_k\setminus \pa_\G(R_k)$ may include $R_k$ itself, and by convention $R_k$ is not independent of itself.

If \eqref{eq:id_criteria} fails, collect indicators in $\widetilde{\mathcal S}_k\setminus \pa_\G(R_k)$ for which \eqref{eq:id_criteria} fails to hold into the set 

\vspace{-1.6cm}
\begin{align}
\mathcal R_k^d
\coloneqq
\big\{
R_j\in \widetilde{\mathcal S}_k\setminus \pa_\G(R_k)
\, : \,
R_k\not\ci R_j\mid \pa_\G(R_k)
\text{ in } p(\cdot\mid \doo(\mathcal T_k=1))
\big\}.
\label{eq:R_nonsep}
\end{align}%
\vspace{-1.5cm}

If $\mathcal R_k^d\cap \mathcal R_k^p\neq\emptyset$, then no admissible intervention can separate $R_k$ from its problematic set $R_k^p$, and $\pi_k$ cannot be identifiable (even at evaluation $R=1$). In such cases, $R_k$ is added to $\mathcal D$, and the algorithm returns non-identifiability of both $\pi_k$ and the target law.  Intuitively, identification fails because all possible interventions have already been applied in an attempt to establish independence between $R_k$ and the problematic indicators in $\R^p_k$ given $\pa_\G(R_k)$.  

In contrast if $\mathcal R_k^d\cap \mathcal R_k^p = \emptyset$, the dependence between $R_k$ and $\mathcal R_k^p$ has already been resolved, but identification fails due to selection induced by unnecessary interventions. In this case, identification may still be achievable by pruning parts of the intervention trees, as described next. An example of pruning an unnecessary intervention was described in Example~\ref{ex:beyond_paths}, where identifying $\pi_4$ required identifying $\pi_3$ without intervening on $R_2$. 

\vspace{-0.25cm}
\paragraph*{Pruning unnecessary interventions [\blue{tree-prune}].}\hspace{-0.3cm}
If $\mathcal R_k^d\cap \mathcal R_k^p=\emptyset$, identification may still be possible by pruning interventions that introduce selection which propagates to $R_k$. 

Let $\mathcal C_k$ denote the collection of descendants of $R_k$ such that intervening on them induces selection on at least one indicator in $\mathcal R_k^d$. Specifically, for each $R_j \in \mathcal R_k^d$, define
%
$
\mathcal C^{\text{dir}}_{k,j}
\coloneqq
\{R_i\in \de_\G(R_k)\mid X_j\in \pa_\G(R_i)\}, 
$
%
the set of descendants of $R_k$ whose intervention induces selection on $R_j$. We then define

\vspace{-2.2cm}
\begin{align}
\mathcal C_k
\coloneqq
\bigcup_{R_j\in \mathcal R_k^d} \mathcal C^{\text{dir}}_{k,j}. 
\end{align}%
\vspace{-1.5cm}

Any indicator in $\mathcal C_k$ cannot be intervened on without obstructing identification of $\pi_k$. 

The pruning procedure proceeds as follows. For each candidate intervention $R_i\in R^*$:
If $R_i\in \mathcal C_k$, then intervening on $R_i$ necessarily induces selection on an indicator in $\mathcal R_k^d$, and $R_i$ is removed from the candidate set $R^*$. 
Otherwise if $R_i \not\in \mathcal C_k$, consider the subtree $\mathbb T_i$ retrieved from $\mathbb F$. If $\mathbb T_i$ contains children whose interventions induce selection on some $R_j\in \mathcal R_k^d$ that propagates through $R_i$ according to the propagation rule in \eqref{eq:propa}, then the corresponding branches are pruned from $\mathbb T_i$. After this pruning, the algorithm checks whether the propensity score $\pi_i$ remains identifiable under its updated evaluation set. If so, the pruned tree $\mathbb T_i$ is appended to $\mathbb T_k$; otherwise, $R_i$ is removed from $R^*$. Note that if any child of $R_i$ in $\mathbb{T}_i$, say $R_m$, is also a child of $R_k$ in $\mathbb{T}_k$, and the subtree $\mathbb{T}_m$ has already been pruned, then the corresponding subtree of $R_m$ in $\mathbb{T}_i$ must be updated to match the pruned version. This operation enables us to restrict attention to pruning the children of $R_k$, or their children, without needing to consider further descendants. See Appendix Figure~\ref{fig:supp_figs}(c) for an example, with detailed identification discussion given in Appendix Subsection~\ref{subsubsec:update_prune_tree}. All pruned branches are recorded locally in a set $\mathcal B$, and the final $\mathbb T_k$ is added to $\mathbb F$. 

Once all candidates in $R^*$ have been examined, the identification criterion \eqref{eq:id_criteria} is evaluated again using the updated tree $\mathbb T_k$. If it is satisfied, identification of $\pi_k$ succeeds. Otherwise, $\mathcal R_k^d$ is updated and the pruning procedure is repeated. Since the set of indicators is finite, this iterative process terminates.

When identification of $\pi_k$ succeeds, the algorithm constructs the indicator-induced selection set for $R_k$ as follows and proceed to identify the next propensity score according to $\tau$:

\vspace{-1.5cm}
\refstepcounter{equation}\label{eq:R-select}
\begin{align}
    \calS^r_k &\coloneqq \widetilde{\calS}_k \cap \pa_\G(R_k) 
    = \{\calS^x_k\cap \pa_\G(R_k)\} \cup \{ \cup_{R_j \in \ch_{\mathbb{T}_{k}}(R_k)  } \calS_{j \, \downarrow \, k}\}. 
    \eqtagtext{(indicator-induced selection set for $R_k$)}
\end{align}
\vspace{-1.cm}

This set has two components. First, any indicator in $\mathcal S_k^x \cap \pa_\G(R_k)$ must be set to one to render the corresponding counterfactual parents observed. Second, interventions used to identify $\pi_k$ may induce selection that propagates through $R_k$ via \eqref{eq:propa}. Under associational irrelevancy, discussed in Section~\ref{subsec:assoc}, this reduces to $\mathcal S_k^x \cap \pa_\G(R_k)$, since $\mathcal T_k = \emptyset$.

The selection set for $R_k$ in \eqref{eq:selection_set} can therefore be explicitly defined as 

\vspace{-1.5cm}
\begin{align}
\mathcal S_k
&\coloneqq
\mathcal S_k^x\cup \mathcal S_k^r
\ =
\mathcal S_k^x
\cup
\bigcup_{R_j\in \mathcal T_k} \mathcal S_{j\,\downarrow\, k}.
\label{eq:full-select}
\end{align}%
\vspace{-1.5cm}

\vspace{-0.3cm}
\paragraph*{Identification functional induced by the intervention tree $\mathbb T_k$.}\hspace{-0.3cm}
The preceding construction shows that whenever Algorithm~\ref{alg:ID} succeeds, the intervention tree $\mathbb T_k$ encodes a concrete strategy for intervening on missingness indicators so as to obtain a post-intervention distribution in which $R_k$ is independent of its problematic set, without inducing selection on $R_k$ or on its non-intervened descendants. 

To construct an identification functional for $\pi_k$ given $\mathbb T_k$, we define graphical and probabilistic fixing operators. For $R_i \in R$, let $\phi^{\mathcal G}_{R_i}$ denote the graphical fixing operation on an mDAG $\mathcal G$ that removes all incoming edges into $R_i$ and replaces the variable $X_i$ by its observed counterpart $X_i^*$, corresponding to the intervention $\doo(R_i=1)$. Further, we define the probabilistic fixing operator $\phi^p_{R_i}\{p\}$ that maps $p(X,R)$ to a law on $(X,R\setminus\{R_i\})$ defined by

\vspace{-1.5cm}
\begin{align}
\phi^p_{R_i}\{p\}(X,R\setminus R_i)
\;\coloneqq\;
\frac{p(X,R\setminus R_i,R_i=1)}{p(R_i=1 \mid \pa_\G(R_i))},
\label{eq:fixing_single}
\end{align}%
\vspace{-1.5cm} 

For a sequence of indicators $\sigma=(s_1,\ldots,s_m)$, we define the composed fixing operator $\phi^p_{\sigma} \;\coloneqq\; \phi^p_{s_m}\circ \cdots \circ \phi^p_{s_1}$, 
and interpret $\phi^p_{\sigma}$ as the post-intervention distribution induced by intervening sequentially on the indicators in $\sigma$.

Given $R_k$ and $\mathbb T_k$, let $\sigma_k=(s_1,\ldots,s_m)$ be any ordering of $\mathcal T_k$ consistent with the reverse topological order used by the algorithm. The post-intervention distribution induced by $\mathbb T_k$ is $\phi^p_{\sigma_k}\{p\}$. The following theorem shows that whenever the identification criterion holds in $\phi^p_{\sigma_k}\{p\}$ (i.e., condition \eqref{eq:id_criteria}), the tree $\mathbb T_k$ yields an explicit identifying functional for $\pi_k$.

\begin{theorem}[Identification functional induced by $\mathbb T_k$]
\label{thm:pi_id_functional}
Assume identification criterion \eqref{eq:id_criteria} holds for $R_k$ using the tree $\mathbb T_k$ in the post-intervention law $p_{\mathbb T_k}=\phi^p_{\sigma_k}\{p\}$. Then $\pi_k(\pa_\G(R_k))$, evaluated at $\mathcal S_k^r=1$, is identified from the observed data law via 

\vspace{-1.5cm}
\begin{align}
\pi_k(\pa_\G(R_k))\big\vert_{\mathcal S_k^r=1}
\;=\;
p_{\mathbb T_k}(R_k=1 \mid \pa_\G(R_k))
\big\vert_{\mathcal S_k^r=1}.
\label{eq:pi_id_functional}
\end{align}
\end{theorem}

\vspace{-0.5cm}
Theorem~\ref{thm:pi_id_functional} makes explicit the role of the intervention tree in identification. The tree $\mathbb T_k$ determines a post-intervention distribution in which the conditional distribution of $R_k$ given its parents is well-defined and identifiable, while the evaluation set $\mathcal S_k^r$ records the minimal indicator restrictions under which this functional can be recovered from observed data. Together, $(\mathbb T_k,\mathcal S_k^r)$ characterize how $\pi_k$ is identified. The identifying functional in Theorem~\ref{thm:pi_id_functional} admits a representation in terms of observed data and inverse-probability weights.

\begin{corollary}[Observed-data representation]
\label{cor:pi_id_ratio}
Under the conditions of Theorem~\ref{thm:pi_id_functional}, define $W_k \;\coloneqq\; \prod_{R_i \in \mathcal T_k} \pi_i(\pa_\G(R_i))^{-1}$, where each $\pi_i$ is evaluated at its own identified evaluation set as returned by Algorithm~\ref{alg:ID}. Then

\vspace{-1.2cm}
\begin{align}
\pi_k(\pa_\G(R_k))\Big\vert_{\mathcal S_k^r=1}
=
\frac{
\mathbb E\!\left[
\mathbb I(R_k=1)\, W_k
\;\middle|\;
\pa_\G(R_k),\;
\mathcal T_k=1,\;
\widetilde{\mathcal S}_k=1
\right]
}{
\mathbb E\!\left[
W_k
\;\middle|\;
\pa_\G(R_k),\;
\mathcal T_k=1,\;
\widetilde{\mathcal S}_k=1
\right]
}\Big\vert_{\mathcal S^r_k=1}.
\label{eq:pi_id_ratio}
\end{align}%
\vspace{-1.2cm}

If $\mathcal T_k=\emptyset$, then $W_k\equiv 1$ and the expression reduces to the associational identification formula.
\end{corollary}

We illustrate the tree-construction of Algorithm~\ref{alg:ID} through the following examples. 

\begin{example}(Identification without pruning) 
The mDAG $\mathcal G_1$ in Figure~\ref{fig:ex_tree}(a) illustrates a setting in which no pruning is required. The algorithm begins with $R_1$, whose propensity score $\pi_1=p(R_1=1 \,|\, R_2,R_3)$ is directly observed. The algorithm then proceeds to $R_2$ and $R_3$ in either order, since neither is a descendant of the other. For concreteness, consider $R_2$, whose propensity score is
$\pi_2=p(R_2=1 \,|\, X_1,X_3,R_4)$. Here, $\mathcal R_2^p=\{R_1\}$, so an intervention on $R_1$ is required to separate $R_2$ from its problematic set. Since $\mathcal S_1=\emptyset$, this intervention induces no additional selection, and $\pi_2$ is fully identified. An analogous argument applies to $R_3$. Finally, for $R_4$, we have $\pi_4=p(R_4=1 \,|\, X_1,X_2,X_3)$ with $\mathcal R_4^p=\{R_1,R_2,R_3\}$. Intervening on all three indicators yields separation without inducing selection, so $\pi_4$ is identified. As all propensity scores are identified without evaluation restrictions, the target law is identified. Key quantities for this example are summarized in Appendix Table~\ref{apptab:ex-tree1}, with $\mathbb{F}_1$ shown in Figure~\ref{fig:ex_tree}(d). 
\end{example}

\begin{figure}[!t]
	\begin{center}
		\scalebox{0.72}{
			\begin{tikzpicture}[>=stealth, node distance=1.3cm, 
				decoration = {snake, pre length=3pt,post length=7pt,},
				every text node part/.style={align=center}]
				\tikzstyle{format} = [thick, circle, minimum size=1.0mm, inner sep=0pt]
				\begin{scope}[xshift=0cm]
					\path[->, thick]
					node[format] (x11) {$X_1$}
					node[format, right of=x11, xshift=0.5cm] (x21) {$X_2$}
					node[format, right of=x21, xshift=0.5cm] (x31) {$X_3$}
					node[format, right of=x31, xshift=0.5cm] (x41) {$X_4$}
					node[format, below of=x11] (r1) {$R_1$}
					node[format, below of=x21] (r2) {$R_2$}
					node[format, below of=x31] (r3) {$R_3$}
					node[format, below of=x41] (r4) {$R_4$}
					
					(x11) edge[blue, dashed, -,, bend left=0] (x21)
					(x21) edge[blue, dashed, -,, bend left=0] (x31)
					(x31) edge[blue, dashed, -,, bend left=0] (x41)
					(x11) edge[blue, dashed, -,, bend left=20] (x31)
					(x11) edge[blue, dashed, -,, bend left=30] (x41)
					(x21) edge[blue, dashed, -,, bend left=20] (x41)
					
					(x11) edge[blue] (r2)
					(x11) edge[blue] (r3)
					(x11) edge[blue] (r4)
					(x21) edge[blue] (r3)
					(x21) edge[blue] (r4)
					(x31) edge[blue] (r4)
					(x31) edge[blue] (r2)
					
					(r2) edge[blue] (r1)
					(r3) edge[blue, bend left=30] (r1)
					(r4) edge[blue, bend left=30] (r2)
					(r4) edge[blue] (r3)
					
					
					node[format, below of=r2, yshift=0.cm, xshift=1.2cm] (a) {(a) $\G_1$} ;
				\end{scope}

				\begin{scope}[xshift=7.5cm, yshift=0.cm]
					\path[->, thick]
					node[format] (x11) {$X_1$}
					node[format, right of=x11] (x21) {$X_2$}
					node[format, right of=x21] (x31) {$X_3$}
					node[format, right of=x31] (x41) {$X_4$}
					node[format, right of=x41] (x51) {$X_5$}
					node[format, right of=x51] (x61) {$X_6$}
					node[format, below of=x11] (r1) {$R_1$}
					node[format, below of=x21] (r2) {$R_2$}
					node[format, below of=x31] (r3) {$R_3$}
					node[format, below of=x41] (r4) {$R_4$}
					node[format, below of=x51] (r5) {$R_5$}
					node[format, below of=x61] (r6) {$R_6$}
					
					(x11) edge[blue, dashed, -,] (x21)
					(x21) edge[blue, dashed, -,] (x31)
					(x31) edge[blue, dashed, -,] (x41)
					(x41) edge[blue, dashed, -,] (x51)
					(x51) edge[blue, dashed, -,] (x61)
					(x11) edge[blue, dashed, bend left, -] (x31)
					(x11) edge[blue, dashed, bend left, -] (x41)
					(x11) edge[blue, dashed, bend left, -] (x51)
					(x11) edge[blue, dashed, bend left, -] (x61)
					(x21) edge[blue, dashed, bend left, -] (x41)
					(x21) edge[blue, dashed, bend left, -] (x51)
					(x21) edge[blue, dashed, bend left, -] (x61)
					(x31) edge[blue, dashed, bend left, -] (x51)
					(x31) edge[blue, dashed, bend left, -] (x61)
					(x41) edge[blue, dashed, bend left, -] (x61)
					
					(x11) edge[blue] (r3)
					(x31) edge[blue] (r4)
					(x31) edge[blue] (r5)
					(x31) edge[blue] (r6)
					(x41) edge[blue] (r1)
					(x41) edge[blue] (r6)
					(x51) edge[blue] (r2)
					(x61) edge[blue] (r5)
					(r6) edge[blue] (r5)
					(r6) edge[blue] (r5)
					(r6) edge[blue, bend left] (r3)
					(r5) edge[blue, bend left=20] (r3)
					(r4) edge[blue] (r3)
					(r3) edge[blue] (r2)
					(r2) edge[blue] (r1)
					
					
					
					node[format, below of=r3, yshift=0.cm, xshift=1cm] (b) {(b) $\G_2$} ;
					;
				\end{scope}

                \begin{scope}[xshift=16cm, yshift=0cm]
					\path[->, thick]
					node[format] (x11) {$X_1$}
					node[format, right of=x11] (x21) {$X_2$}
					node[format, right of=x21] (x31) {$X_3$}
					node[format, right of=x31] (x41) {$X_4$}
					node[format, right of=x41] (x51) {$X_5$}
					
					node[format, below of=x11] (r1) {$R_1$}
					node[format, below of=x21] (r2) {$R_2$}
					node[format, below of=x31] (r3) {$R_3$}
					node[format, below of=x41] (r4) {$R_4$}
					node[format, below of=x51] (r5) {$R_5$}
					
					
					(x11) edge[blue, dashed, -,] (x21)
					(x21) edge[blue, dashed, -,] (x31)
					(x31) edge[blue, dashed, -,] (x41)
					(x41) edge[blue, dashed, -,] (x51)
					(x11) edge[blue, dashed, bend left, -] (x31)
					(x11) edge[blue, dashed, bend left, -] (x41)
					(x11) edge[blue, dashed, bend left, -] (x51)
					(x21) edge[blue, dashed, bend left, -] (x41)
					(x21) edge[blue, dashed, bend left, -] (x51)
					
					(x21) edge[blue] (r5)
					(x11) edge[blue] (r4)
					(x51) edge[blue] (r2)
					(x41) edge[blue] (r1)
					
					(r5) edge[blue] (r4)
					(r5) edge[blue, bend left] (r3)
					(r4) edge[blue, bend left] (r2)
					(r3) edge[blue] (r2)
					(r2) edge[blue] (r1)
					
					
					
					node[format, below of=r3, yshift=0.cm, xshift=0.0cm] (c) {(c) $\G_3$} ;
					;
				\end{scope}


                \begin{scope}[xshift=3cm, yshift=-4cm]
					\path[->, thick]
					node[format] (id) {$R$}
					node[format, below of=id, yshift=0.5cm] (r3) {$R_3$ }
					node[format, left of=r3, xshift=0.cm] (r2) {$R_2$}
					node[format, left of=r2, xshift=0.cm] (r1) {$R_1$}
					node[format, right of=r3, xshift=0.cm] (r4) {$R_4$}
					
					node[format, below of=r2, xshift=0.cm, yshift=0.cm] (r2-r1) {$R_1$}
					
					node[format, below of=r3, xshift=-0.cm, yshift=0.cm] (r3-r1) {$R_1$}
					
					node[format, below of=r4, yshift=0.cm] (r4-r2) {$R_2$}
					node[format, left of=r4-r2, xshift=0.75cm, yshift=0.cm] (r4-r1) {$R_1$}
					node[format, right of=r4-r2, xshift=-0.75cm, yshift=0.cm] (r4-r3) {$R_3$}
					
					node[format, below of=r4-r2, yshift=0.cm] (r4-r2-r1) {$R_1$}
					
					node[format, below of=r4-r3, yshift=0.cm] (r4-r3-r1) {$R_1$}
					
					(id) edge[black, -] (r1)
					(id) edge[black, -] (r2)
					(id) edge[black, -] (r3)
					(id) edge[black, -] (r4)
					
					(r2) edge[blue] (r2-r1)
					(r3) edge[blue] (r3-r1)
					
					(r4) edge[blue] (r4-r1)
					(r4) edge[blue] (r4-r2)
					(r4) edge[blue] (r4-r3)
					(r4-r2) edge[blue] (r4-r2-r1)
					(r4-r3) edge[blue] (r4-r3-r1)
					
					node[format, below of=r3-r1, yshift=-1cm, xshift=-0.2cm] (d) {(d) ${\cal F}_1$ corresponding to $\G_1$} ;
					
					;
				\end{scope}
                
				\begin{scope}[xshift=9cm, yshift=-4.cm]
					\path[->, thick]
					node[format] (id) {$R$}
					node[format, below of=id, yshift=0.5cm] (r3) {$R_3$ }
					node[format, left of=r3, xshift=0.2cm] (r2) {$R_2$}
					node[format, left of=r2, xshift=0.5cm] (r1) {$R_1$}
					node[format, right of=r3, xshift=0.25cm] (r4) {$R_4$}
					node[format, right of=r4, xshift=0.35cm] (r5) {$R_5$}
					node[format, right of=r5, xshift=0.6cm] (r6) {$R_6$} 
					
					node[format, below of=r3, xshift=-0.35cm] (r3-r1) {$R_1$}
					node[format, below of=r3, xshift=0.35cm] (r3-r2) {$R_2$}
					
					node[format, below of=r4, xshift=-0.35cm] (r4-r2) {$R_2$}
					node[format, below of=r4, xshift=0.35cm] (r4-r3) {$R_3$}
					node[format, below of=r4-r3, xshift=-0.35cm] (r4-r3-r1) {$R_1$}
					node[format, below of=r4-r3, xshift=0.35cm] (r4-r3-r2) {$R_2$}
					
					node[format, below of=r5, xshift=-0.35cm] (r5-r1) {$R_1$}
					node[format, below of=r5, xshift=0.35cm] (r5-r3) {$R_3$}
					node[format, below of=r5-r3, xshift=-0.35cm] (r5-r3-r1) {$R_1$}
					node[format, below of=r5-r3, xshift=0.35cm] (r5-r3-r2) {$R_2$}
					
                    node[format, below of=r6, yshift=0.cm] (r6-r2) {$R_2$}
					node[format, below of=r6, xshift=-0.6cm] (r6-r1) {$R_1$}
					node[format, below of=r6, xshift=0.6cm] (r6-r3) {$R_3$}
					node[format, below of=r6-r3, xshift=-0.35cm] (r6-r3-r1) {$R_1$}
					node[format, below of=r6-r3, xshift=0.35cm] (r6-r3-r2) {$R_2$} 
					
					
					(id) edge[black, -] (r1)
					(id) edge[black, -] (r2)
					(id) edge[black, -] (r3)
					(id) edge[black, -] (r4)
					(id) edge[black, -] (r5)
					(id) edge[black, -] (r6)
					
					(r3) edge[blue] (r3-r1)
					(r3) edge[blue] (r3-r2)
					
					(r4) edge[blue] (r4-r3)
					(r4) edge[blue] (r4-r2)
					(r4-r3) edge[red, decorate] (r4-r3-r1)
					(r4-r3) edge[blue] (r4-r3-r2)
					
					(r5) edge[blue] (r5-r3)
					(r5) edge[blue] (r5-r1)
					(r5-r3) edge[blue] (r5-r3-r1)
					(r5-r3) edge[red, decorate] (r5-r3-r2)
					
					(r6) edge[blue] (r6-r3)
                    (r6) edge[red, decorate] (r6-r2)
					(r6) edge[blue] (r6-r1)
					(r6-r3) edge[blue] (r6-r3-r1)
					(r6-r3) edge[red, decorate] (r6-r3-r2)
					
					
					node[format, below of=r4-r2, yshift=-1cm, xshift=.5cm] (e) {(e)  ${\cal F}_2$ corresponding to $\G_2$} ;
					
					;
				\end{scope}
                
				\begin{scope}[xshift=18cm, yshift=-4cm]
					\path[->, thick]
					node[format] (id) {$R$}
					node[format, below of=id, yshift=0.5cm] (r3) {$R_3$ }
					node[format, left of=r3, xshift=0.5cm] (r2) {$R_2$}
					node[format, left of=r2, xshift=0.5cm] (r1) {$R_1$}
					node[format, right of=r3, xshift=-0.5cm] (r4) {$R_4$}
					node[format, right of=r4, xshift=0.5cm] (r5) {$R_5$}
					
					node[format, below of=r4, xshift=0cm] (r4-r2) {$R_2$}
					
					node[format, below of=r5, xshift=-0.6cm] (r5-r1) {$R_1$}
					node[format, below of=r5, xshift=0cm] (r5-r3) {$R_3$}
					node[format, below of=r5, xshift=0.6cm] (r5-r4) {$R_4$}
					
					node[format, below of=r5-r4, xshift=0cm] (r5-r4-r2) {$R_2$}

					(id) edge[black, -] (r1)
					(id) edge[black, -] (r2)
					(id) edge[black, -] (r3)
					(id) edge[black, -] (r4)
					(id) edge[black, -] (r5)
					
					(r4) edge[blue] (r4-r2)
					
					(r5) edge[red, decorate] (r5-r4)
					(r5) edge[blue] (r5-r3)
					(r5) edge[red, decorate] (r5-r1)
					
					(r5-r4) edge[red, decorate] (r5-r4-r2)
					
					node[format, below of=r5-r4, yshift=-1cm, xshift=-2.25cm] (f) {(f)  ${\cal F}_3$ corresponding to $\G_3$} ; 
					
					;
				\end{scope}
			\end{tikzpicture}
		}
	\end{center}
	\vspace{-1.25cm}
	\caption{(a, b, c) Examples used to illustrate the identification Algorithm~\ref{alg:ID}; (d, e, f) The corresponding constructed trees. }
	\label{fig:ex_tree}
\end{figure}
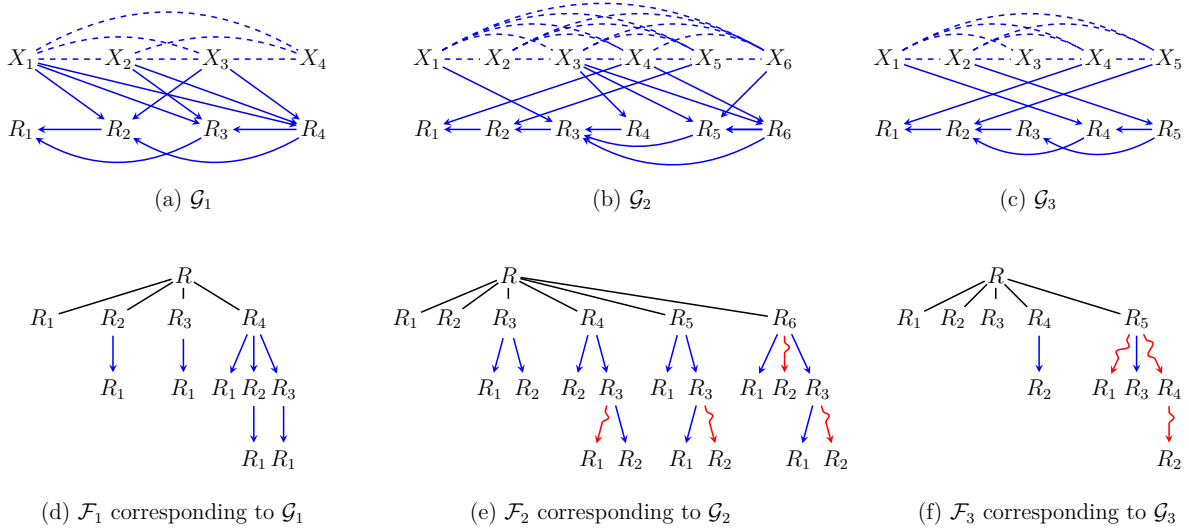

\begin{example}(Identification with pruning) 
The mDAG $\mathcal G_2$ in Figure~\ref{fig:ex_tree}(b) demonstrates how pruning removes unnecessary interventions that would otherwise obstruct identification. The propensity scores for $R_1$ and $R_2$ are identified via associational irrelevancy. Consider $R_3$, whose propensity score $\pi_3=p(R_3=1 \,|\, X_1,R_4,R_5,R_6)$ has $\mathcal R_3^p=\{R_1\}$. Intervening on both $R_1$ and $R_2$ yields identification of $\pi_3$ evaluated at $\mathcal S_3^r=\{R_4,R_5\}$. When identifying $\pi_4=p(R_4=1\,|\, X_3)$, however, intervening on $R_3$ using its full tree induces selection on $R_4$ itself. The algorithm detects this via $\mathcal R_4^d$ and prunes the branch corresponding to the unnecessary intervention on $R_2$ within $\mathbb T_3$. After pruning, $\pi_3$ remains identifiable under a reduced evaluation set, and $\pi_4$ becomes identifiable. A similar pruning step is required when identifying $\pi_6$, where selection induced indirectly through $\mathbb T_3$ must again be removed. In this graph, all problematic selections can be eliminated through pruning, and the target law is identified. Details are summarized in Appendix Table~\ref{apptab:ex-tree2}, with $\mathbb{F}_2$ shown in Figure~\ref{fig:ex_tree}(e). 
\end{example}

\begin{example}(Non-identification despite pruning) 
The mDAG $\mathcal G_3$ in Figure~\ref{fig:ex_tree}(c) illustrates a setting in which pruning cannot eliminate all selection bias. We focus on $R_5$, whose propensity score is
$\pi_5=p(R_5=1\,|\, X_2)$ with $\mathcal R_5^p=\{R_2\}$ and $\mathcal C^{\text{dir}}_{5,5}=\{R_2\}$. Interventions on $R_1$, $R_3$, and $R_4$ are initially admissible. Pruning removes the subtree for $R_4$, but this renders $\pi_4$ itself non-identifiable. Subsequent pruning eliminates $R_1$, leaving only $R_3$. An intervention on $R_3$ alone, however, cannot separate $R_5$ from $R_2$, yielding $\mathcal R_5^d=\mathcal R_5^p$ and triggering non-identification. Thus, even though pruning removes unnecessary interventions, all admissible intervention strategies induce unavoidable selection on the problematic set. The algorithm concludes that $\pi_5$, and hence the target law, is not identifiable. Appendix Table~\ref{apptab:ex-tree3} provides a summary of this example, with $\mathbb{F}_3$ shown in Figure~\ref{fig:ex_tree}(f). 
\end{example} 


\begin{algorithm}[tbp]
\setstretch{0.9}
	\caption{{\small  \textproc{Tree-based Identification Algorithm}$(\G(X, R, X^*))$}}
	\label{alg:ID}
	\begin{algorithmic}[1]		
		\vspace{0.25cm}		
		\State {\bf  Identification Procedure:} 
		\begin{itemize}
			\item Index $R$ according to a valid reversed topological order on $R$, denoted by $\tau$  
			\item Let $\mathbb{F}$ be the list of all constructed trees; initialize it to an empty list 
                \item Let $\mathcal{D}=\{R_j\in R\mid {\pi_j}_{\mid R=1}\text{ is not ID}\}$; initialize it to an empty list
			\item \textbf{For} each $R_k \in R$, let: 
			\begin{itemize}
				\item Let $\mathbb{T}_{k}$ be a tree data structure with $R_k$ as the root node

				\item $\calS^x_k \coloneqq \{R_j \in R \mid X_j \in \pa_\G(R_k)\}$ and ${\cal R}^p_k\coloneqq \calS^x_k \cap \de_\G(R_k)$ 
				
				
				
				\item \textbf{If} ${\cal R}^p_k= \emptyset$: \hspace{1em} $*$ Add $\mathbb{T}_{k}$ to $\mathbb{F}$  \hspace{2em}$*$ $\tilde{\calS}_k=\calS_k=\calS^x_k$ \hspace{2em}$*$ $\calS^r_k=\tilde{\calS}_k\cap \pa_\G(R_k)$
				\item[] \textbf{Else}: \hspace{1em} $*$ $(\mathbb{T}_{k}, \tilde{\calS}_k, \calS_k, \calS^r_k,  \mathbb{F},\mathcal{D}) \leftarrow$  \blue{\bf tree-construction}($R_k, \mathbb{F}, \mathcal{D}$) 
					
			\end{itemize}
		
			\item \textbf{If} ${\cal D}=\emptyset$: Print \textit{target law is identified}
            \item RETURN $\mathcal{D},\mathbb{F}$
		\end{itemize}
		
		\vspace{0.25cm}
		\State \blue{\bf tree-construction}($R_k, \mathbb{F}, \mathcal{D}$)
        {\setstretch{1.0}
		\begin{itemize}		
        \item Let $\C^{\text{dir}}_{k,k}\coloneqq\{R_i \in \de_\G(R_k) \mid X_k\in\pa_\G(R_i)\}$
		\item \textbf{For} ${R}_i \! \in \! \de_\G(R_k) \! \setminus \! \{\C^{\text{dir}}_{k,k},\mathcal{D}\}$: -- Add $R_k \rightarrow R_i$ in $\mathbb{T}_{k}$, augment $R_i$ in $\mathbb{T}_{k}$ with $\mathbb{T}_{i}$ in $\mathbb{F}$
        \item Let $\tilde{\mathcal{S}}_k = 
        \calS^x_k
        \cup_{R_i \in \ch_{\mathbb{T}_{k}}(R_k)  } {\cal S}_{i}  $ 
        \item (id\_status, $\tilde{\calS}_k,\T_k$) $\leftarrow$ \blue{\bf id-status}($R_k, \tilde{\calS}_k, \mathbb{T}_{k}$)
        \item $\calS^r_k = \tilde{\calS}_k \cap \pa_\G(R_k)$ and  
        $\calS_k=\calS^x_k \cup \calS^r_k$
        \item \textbf{If} id\_status=T: \hspace{1em} $*$ Add $\mathbb{T}_{k}$ to $\mathbb{F}$
        \hspace{0.5cm}
        \textbf{Else:} \hspace{1em} $*$ Add $R_k$ to $\mathcal{D}$
		\item RETURN $\mathbb{T}_{k}, \tilde{\calS}_k, \calS_k, \calS^r_k, \mathbb{F},\mathcal{D}$ 
	\end{itemize}
    }
		
	\vspace{0.25cm}
	\State \blue{\bf tree-prune}($R_i, \mathbb{T}_{i}, \mathbb{T}_{k}, {\cal C}_k, {\cal B}$)
	\begin{itemize}
		\item \textbf{For} all $R_m \in \ch_{\mathbb{T}_{i}}(R_i)$:
		\begin{itemize}
			\item \textbf{If} $R_m \in \mathcal{B}$: \hspace{1em} $*$ Update $\mathbb{T}_{m}$ with the one in $\mathcal{B}$
			\item \textbf{If} $R_m \in {\cal C}_k$: Prune the node $R_m$ and its corresponding tree from the subtree $\mathbb{T}_{i}$ 
		\end{itemize}
		\item Update: 
        {$*$ $\tilde{\mathcal{S}}_i = 
		\calS^x_i \ \cup_{R_m \in \ch_{\mathbb{T}_{i}}(R_i)  } \ {\cal S}_m$ \hspace{1em} $*$ $\calS^r_i=\tilde{\calS}_i\cap\pa_\G(R_i)$ \hspace{1em} $*$ $\calS_i=\calS^x_i\cup\calS^r_i$}
        \item (id\_status, $\tilde{\calS}_i,\T_i$) $\leftarrow$ \blue{\textbf{id-status}}($R_i,\tilde{\calS}_i,\T_i$)
        \item \textbf{If} id\_status=T: -- Update $\mathbb{T}_i$ in $\cal B$. Replace $\T_i$ augmented to $R_i$ in $\T_k$ by $\T_i$ from $\mathcal{B}$
        \item[] \textbf{Else}: Prune $R_i$ and augmented $\T_i$ from $\T_k$
		\item RETURN $(\mathbb{T}_{i}, \tilde{\calS}_i, \mathbb{T}_{k}, {\cal B})$
	\end{itemize}

	\vspace{0.25cm}
	\State \blue{\bf id-status}($R_k, \tilde{\mathcal{S}}_k, \mathbb{T}_{k}$)
	\begin{itemize}
        \item Let $\R^d_k\coloneqq \big\{R_j\in \tilde{\calS}_k \backslash \pa_\G(R_k) \big| R_k \not\ci R_j \mid \pa_\G(R_k) \text{ in } p(. \ | \ \doo\{\ch_{\mathbb{T}_{k}}(R_k)\})\big\}$
		\item \textbf{If} $\R^d_k=\emptyset$:  -- id\_status = T \quad -- RETURN id\_status, $\tilde{\calS}_k,\T_k$
		\item \textbf{Else If}: $\R^p_k\cap \R^d_k\neq \emptyset$: -- id\_status = F \quad -- RETURN id\_status, $\tilde{\calS}_k,\T_k$
        \item \textbf{Else}: 
        \begin{itemize}

            \item Let $\mathcal{B}$ collect updated branches for $\T_k$ and the corresponding indicators; initialized to an empty list. 
            \item Let $\C_k\coloneqq \cup_{R_j\in \R^d_k} \ \C^{\text{dir}}_{k,j}$ where $\C^{\text{dir}}_{k,j}\coloneqq\{R_i \in \de_\G(R_k) \mid X_j\in\pa_\G(R_i)\}$

            \item \textbf{For} $R_i\in \ch_{\T_k}(R_k)$: \hspace{0.3cm} {\scriptsize (consistent with the order $\tau$)} 
                    \begin{itemize}
                        \item \textbf{If} $R_i\in \C_k$: Prune the node $R_i$ and its corresponding tree from $\T_k$
                        \item \textbf{If}  $R_i \not\in \C_k$, and either (i) $\ch_{\mathbb{T}_{i}}(R_i)\cap \mathcal{B}\neq \emptyset$, or (ii) $\exists R_j\in\R^d_k \text{ s.t. }\ch_{\mathbb{T}_{i}}(R_i)\cap \C^{\text{dir}}_{k,j}\neq\emptyset$ and $R_j\in\pa_\G(R_k)$: 
                        
                        \hspace{0.25cm} -- $(\mathbb{T}_{i}, {\calS}_i, \mathbb{T}_{k}, {\cal B}) \leftarrow $ \blue{\bf tree-prune}($R_i, \mathbb{T}_{i}, \mathbb{T}_{k}, {\cal C}_k, {\cal B}$) 
                    \end{itemize}
                \item Update: 
                    { $*$ $\tilde{\mathcal{S}}_k = 
		              \calS^x_k \ \cup_{R_i \in \ch_{\mathbb{T}_{k}}(R_k)  } \ {\cal S}_i$}
                \item RETURN \blue{\bf id-status}($R_k, \tilde{\calS}_k, \mathbb{T}_{k}$)
            \end{itemize}
        \end{itemize}
	\end{algorithmic}
\end{algorithm} 


\section{Estimation Procedure}
\label{sec:estimation}

The identification algorithm in Section~\ref{sec:identification} returns, for each $R_k\in R$, an intervention tree $\mathbb T_k$ and an evaluation set $\mathcal S_k^r$ such that $\pi_k(\pa_\G(R_k))\vert_{\mathcal S_k^r=1}$ is identified as a functional of the observed data law. Given $n$ i.i.d. copies of $(X^*,R)$ drawn from the observed data law induced by the full law $p(X,R) \in \mathcal M$, this section develops estimation and inference procedures that mirror that construction. We first describe \textit{recursive inverse probability weighted}  estimating euqations for the identified propensity scores $\{\pi_k\vert_{\mathcal S_k^r=1}\}_{k=1}^K$. We then describe estimation of generic parameters defined through moment conditions under the target law. These procedures are summarized in Appendices Algorithms~\ref{alg:propensity} and \ref{alg:est}.

\subsection{Estimation of propensity scores}
\label{subsec:est-ps}

Estimation proceeds sequentially in the reversed topological order $\tau$ used by the identification algorithm. Let $\mathcal P$ and $\mathcal E$ collect the fitted propensity score models and the estimating equations, respectively, initialized as $\emptyset$. For each $R_k$, fix a parametric model for $\pi_k(\pa_\G(R_k);\theta_k)\vert_{\mathcal S_k^r=1}$ indexed by a finite dimensional parameter $\theta_k$.

Let 

\vspace{-2cm}
\begin{align}
W_k(\theta_{\mathcal T_k})
\;\coloneqq\;
\prod_{R_i\in \mathcal T_k}
\frac{\I(R_i=1)}{\pi_i(\pa_\G(R_i);\theta_i)\vert_{\mathcal S_i^r=1}},
\label{eq:Wk}
\end{align}%
\vspace{-1.25cm} 

be the inverse propensity weight with the convention $W_k(\theta_{\mathcal T_k})\equiv 1$ if $\mathcal T_k=\emptyset$, where $\mathcal T_k$ denotes the children of $R_k$ in $\mathbb T_k$. The product in \eqref{eq:Wk} implements the same intervention logic used for identification, namely, intervening on the indicators in $\mathcal T_k$, by setting them to one and reweighting by the product of their propensity scores.

Recall the pre-selection set $\widetilde{\mathcal S}_k$ in \eqref{eq:pre-select}, and let $f_k(\pa_\G(R_k))$ be a vector whose dimension matches that of $\theta_k$. We estimate $\theta_k$ by solving the empirical estimating equation

\vspace{-1.75cm}
\begin{align}
P_n \Psi_k(X^*,R; \theta_k,\hat{\theta}_{\mathcal T_k}) \;=\; 0,
\label{eq:ps-est-eq}
\end{align}%
\vspace{-2cm}
 
where 

\vspace{-1.8cm}
\begin{align}
\Psi_k(X^*,R;\theta_k,\hat{\theta}_{\mathcal T_k})
\;\coloneqq\;
\I(\widetilde{\mathcal S}_k=1)\,
W_k(\hat{\theta}_{\mathcal T_k})\,
f_k(\pa_\G(R_k))\,
\big\{R_k-\pi_k(\pa_\G(R_k);\theta_k)\big\}.
\label{eq:Psi_k_def}
\end{align}
\vspace{-1.65cm}

The restriction $\I(\widetilde{\mathcal S}_k=1)$ ensures that the covariates in $\pa_\G(R_k)$ and $\pa_\G(R_i)$ for all $R_i\in\mathcal T_k$ are observed and that each propensity score appearing in the weight $W_k(\theta_{\mathcal T_k})$ is evaluated under its required evaluation $\mathcal S_i^r=1$. When $\mathcal T_k=\emptyset$, \eqref{eq:Psi_k_def} reduces to a standard estimating equation computed on the subset $\widetilde{\mathcal S}_k=1$. When $\mathcal T_k\neq\emptyset$, the weighting term implements the post-intervention distribution used to justify identification of $\pi_k$.

The pruning step in Algorithm~\ref{alg:ID} may modify a previously stored subtree $\mathbb T_i$ when appending it to $\mathbb T_k$. This changes the child set $\mathcal T_i=\ch_{\mathbb T_i}(R_i)$ and therefore changes the estimating equation \eqref{eq:Psi_k_def} used to estimate $\theta_i$. In this case, $\theta_i$ must be re-estimated using the pruned version of $\mathbb T_i$. We use the same notation $\Psi_i$ for the modified estimating functions, but to distinguish them, we collect them in $\mathcal{E}_k$, with $\mathcal P_k$  collecting the propensity score fits after all re-estimation steps needed for identifying $\pi_k$; the subscript $k$ indicates they are specifically tailored for estimating $\pi_k$. Re-estimation is also performed in the order $\tau$, ensuring that any parameters required to form weights in descendant equations are available when needed.

The estimator $\hat{\theta}_k$ is consistent for $\theta_k$ if the models for $\pi_k$ and all propensity scores for $R_i \in \mathcal{T}_k$ are correctly specified and standard regularity conditions hold \citep{liang1986longitudinal,robins1995analysis}. Its asymptotic variance has a sandwich form, with contributions from the estimating equation for $\theta_k$ and the associated propensity score equations, detailed below.

Let $\boldsymbol{\theta}_k$ denote the stacked parameter vector consisting of $\theta_k$ and the parameters associated with all indicators appearing in the (possibly pruned) intervention tree $\mathbb T_k$, ordered so that $\theta_k$ is last. Let $\boldsymbol{\Psi}_k(\boldsymbol{\theta}_k)$ be the corresponding stacked estimating function obtained by collecting the estimating equations \eqref{eq:Psi_k_def} for $R_k$, for its children in $\mathbb T_k$, and recursively for all descendants of those children in $\mathbb T_k$. By construction, $\boldsymbol{\Psi}_k$ includes exactly the estimating equations whose solutions are required to form the inverse propensity weights used in estimating $\pi_k$.

\begin{theorem}[Asymptotic normality of recursive propensity score estimators]
\label{thm:ps-asymp}
Assume that each propensity score model appearing in $\boldsymbol{\Psi}_k(\boldsymbol{\theta}_k)$ is correctly specified, and that standard regularity conditions for M-estimation hold, including differentiability of $\boldsymbol{\Psi}_k$, finiteness of second moments, and nonsingularity of
$A_k\coloneqq \E\{\partial \boldsymbol{\Psi}_k(\boldsymbol{\theta}_k)/\partial \boldsymbol{\theta}_k\}$.
Then $\hat{\boldsymbol{\theta}}_k$, the solution to $P_n\boldsymbol{\Psi}_k(\boldsymbol{\theta}_k)=0$, satisfies

\vspace{-1.5cm}
\begin{align}
\sqrt{n}\,(\hat{\boldsymbol{\theta}}_k-\boldsymbol{\theta}_k)
\;\rightsquigarrow\;
N(0,\; V_k),
\qquad
V_k \;=\; A_k^{-1} B_k (A_k^{-1})',
\end{align}%
\vspace{-1.5cm}

where $B_k\coloneqq \E\{\boldsymbol{\Psi}_k(\boldsymbol{\theta}_k)\boldsymbol{\Psi}_k(\boldsymbol{\theta}_k)'\}$.
A consistent estimator of $V_k$ is $\hat V_k=\hat A_k^{-1}\hat B_k(\hat A_k^{-1})'$ with
$\hat A_k=P_n\{\partial \boldsymbol{\Psi}_k(\boldsymbol{\theta}_k)/\partial \boldsymbol{\theta}_k\}\vert_{\hat{\boldsymbol{\theta}}_k}$
and
$\hat B_k=P_n\{\boldsymbol{\Psi}_k(\boldsymbol{\theta}_k)\boldsymbol{\Psi}_k(\boldsymbol{\theta}_k)'\}\vert_{\hat{\boldsymbol{\theta}}_k}$.
\end{theorem}

The asymptotic variance of $\hat\theta_k$ is the bottom-right block of $V_k$ (and similarly for $\hat V_k$).

In settings where parametric models for $\pi_k$ are difficult to specify, one may replace $\pi_k(\cdot;\theta_k)$ by a flexible regression estimator of $p(R_k=1\,|\, \pa_\G(R_k))$ fit on the subset $\widetilde{\mathcal S}_k=1$, using observation weights $W_k$ when $\mathcal T_k\neq\emptyset$. Many flexible machine learning and statistical models support weights, including tree-based methods and boosting, and can flexibly capture nonlinearities and high-order interactions among covariates. Inference in this setting typically relies on the bootstrap, since closed form variance expressions are generally unavailable; formal asymptotic normality requires additional conditions (e.g., sample splitting or complexity restrictions) and is not pursued here.

\subsection{Statistical analyses for functionals of the target law}
\label{subsec:est-target}

Let $\theta$ denote a target parameter defined as a functional of the target law $p(X)$. We assume $\theta$ is characterized by a moment condition

\vspace{-1.5cm}
\begin{align}
\E\{\, M(\tilde{X}; \, \theta) \, \}=0,
\label{eq:target-moment}
\end{align}%
\vspace{-1.5cm} 

where $\tilde{X}\subseteq X$ is the subset of variables required to evaluate $M$ and the dimension of $M$ matches that of $\theta$. Let $\tilde{R}\subseteq R$ denote the missingness indicators for $\tilde{X}$.

Because $M(\tilde{X};\theta)$ is only observed when $\tilde{R}=1$, estimation must adjust for selection. Moreover, each propensity score appearing in the required inverse probability weights may itself only be evaluable on a further restricted subset determined by its selection set. This motivates the following closure construction.
Define the operator $\operatorname{cl}(\cdot)$ on sets of indicators by 

\vspace{-1.5cm}
\begin{align}
\operatorname{cl}(A)\;\coloneqq\; A \cup \bigcup_{R_i\in A} \mathcal S_i,
\label{eq:closure-op}
\end{align}%
\vspace{-1.5cm} 

where $\mathcal S_i$ is the selection set for $R_i$ produced by the identification algorithm. Starting from $A_0\coloneqq \tilde{R}$, define the recursion $A_{\ell+1}\coloneqq \operatorname{cl}(A_\ell)$. Since $R$ is finite, there exists $\ell^\ast$ such that $A_{\ell^\ast+1}=A_{\ell^\ast}$. We define the smallest set of indicators containing $\tilde{R}$ that is closed under inclusion of the selection sets $\{\mathcal S_i\}$ 

\vspace{-2cm}
\begin{align}
\mathcal R \;\coloneqq\; A_{\ell^\ast},
\label{eq:R_t_def}
\end{align}%
\vspace{-1.7cm} 

Given estimates of $\pi_i(\pa_\G(R_i))\vert_{\mathcal S_i^r=1}$ for $R_i \in \mathcal R$ from Section~\ref{subsec:est-ps}, define

\vspace{-1.3cm}
\begin{align}
\Psi(X^*,R;\theta, \theta_{\mathcal R})
\;\coloneqq\;
\I(\mathcal R=1)\,
\Big\{
\prod_{R_i\in \mathcal R}
\pi_i(\pa_\G(R_i);\theta_i)\vert_{\mathcal S_i^r=1}
\Big\}^{-1}
M(\tilde{X};\theta),
\label{eq:Psi_t_def}
\end{align}%
\vspace{-1.3cm} 

where $\theta_{\mathcal R}\coloneqq (\theta_i)_{R_i\in \mathcal R}$. The estimator $\hat\theta$ is obtained by solving

\vspace{-1.65cm}
\begin{align}
P_n\Psi(X^*, R; \, \theta,\hat\theta_{\mathcal R})=0, 
\label{eq:theta_t_est}
\end{align}
\vspace{-1.5cm}

and is consistent for $\theta$ if all the propensity scores for indicators in $\R$ are consistently estimated. 

\begin{theorem}[Asymptotic normality for target parameters]
\label{thm:target-asymp}
Assume that all propensity score models associated with indicators in $\mathcal R$ are correctly specified, and that standard regularity conditions for the stacked estimating equations hold for the system consisting of the propensity score estimating equations used to fit $\{\theta_i:R_i\in\mathcal R\}$ together with \eqref{eq:theta_t_est}. Then $\hat\theta$ is consistent and $\sqrt{n}\,(\hat\theta-\theta)\;\rightsquigarrow\;N(0,\;V)$, where $V$ is obtained as the appropriate block of the sandwich covariance matrix for the stacked system $\boldsymbol{\Psi}(\boldsymbol{\theta})$, i.e., the estimating equations required to evaluate the inverse probability weights appearing in $\Psi$
\end{theorem}

We briefly illustrate the estimation procedure for a simple functional of the target law, namely the mean of a partially observed variable. 

\begin{example}
Consider estimation of $\theta=\E(X_3)$ in the missing data model of Figure~\ref{fig:ex_tree}(c). Evaluating the estimating function $M(\tilde{X};\theta)=X_3-\theta$ requires observing $X_3$, so the initial indicator set is $\tilde{R}=\{R_3\}$. In this graph, the selection set associated with $R_3$ is empty, implying that the closure construction stabilizes immediately and yields $\mathcal R=\{R_3\}$. As a result, $\theta$ is identified even though the full target law is not, and can be estimated using a single inverse probability weighted estimating equation based on the propensity score for $R_3$. By contrast, attempting to estimate the mean of other variables in the same graph leads to a closure $\mathcal R$ that includes indicators whose propensity scores are not identified, revealing non-identifiability of the corresponding means. 
\end{example}

These examples highlight how the proposed framework distinguishes identifiable functionals from non-identifiable ones, even when full recovery of the target law is impossible. Detailed derivations for these mean estimators, together with additional examples illustrating recursive closure, parametric regression, and causal effect estimation, are provided in Appendix~\ref{app:est_examples}.

\section{Simulations}
\label{sec:sims}

We evaluated the finite-sample performance of the proposed estimation procedures across three statistical tasks of increasing complexity: (1) mean estimation, (2) parametric regression, and (3) causal effect estimation. We compared our approach with complete-case analysis and two widely used imputation-based methods implemented in \textsf{R}: Amelia, which relies on a multivariate normal model estimated via EM with bootstrapping, and multiple imputation by chained equations (MICE). The accompanying \textsf{R} package \href{https://github.com/annaguo-bios/flexMissing}{\texttt{flexMissing}} implements the proposed methods. The simulation code is provided separately at \texttt{annaguo-bios/missing-tree-paper}.

Across all tasks, data were generated from four mDAGs of increasing complexity. These include one three-variable MAR model (Appendix Figure~\ref{fig:sims_tree}(a)) and three MNAR models: a three-variable mDAG (Appendix Figure~\ref{fig:sims_tree}(c)), a five-variable mDAG (Appendix Figure~\ref{fig:sims_tree}(e)), and a ten-variable mDAG (Appendix Figure~\ref{fig:sims_tree}(f)). We denote these by $\mathcal G_1$ through $\mathcal G_4$. For each mDAG, we considered sample sizes of 500, 1000, 2000, 4000, and 8000, with 500 Monte Carlo replicates per setting. Data-generating processes are detailed in Appendix~\ref{appsub:dgp}. Default settings were used for Amelia and MICE, including five imputations.

\textbf{Task~1: mean estimation.} 
For all mDAGs, we generated continuous data and targeted the mean of $X_3$. Estimation bias across simulation replicates is summarized using boxplots in Figure~\ref{fig:sim1-mean}. Under the MAR model $\mathcal G_1$, all methods except complete-case analysis exhibited negligible bias as sample size increased. In contrast, complete-case analysis remained biased even at the largest sample size. Under the MNAR models $\mathcal G_2$ through $\mathcal G_4$, only the proposed tree-based method consistently recovered the true mean, while all competing methods exhibited substantial bias that did not attenuate with increasing sample size. See Appendix~\ref{appsubsub:dgp-sim1} for the DGP and Appendices~\ref{appsub:sim-id} and \ref{appsub:sim-estimation-mean} for estimation details. 

\begin{figure}[t]
    \centering
    \includegraphics[width=1\linewidth,clip, trim=0 5 0 5]{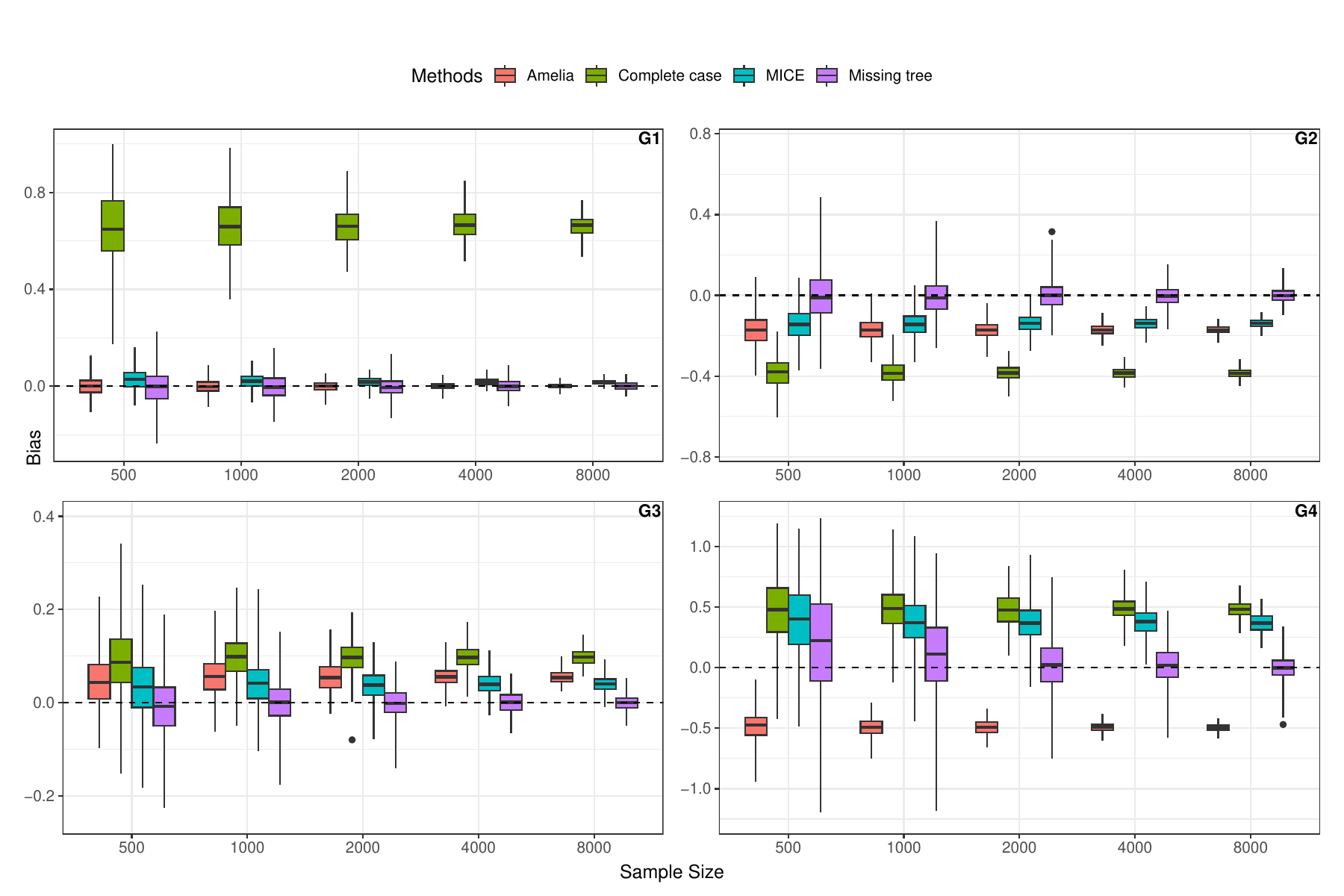}
    \caption{Simulation results for estimation of a mean using four missing data methods: Amelia, complete-case analysis, MICE, and the proposed tree-based method. Panels correspond to data generated under mDAGs $\mathcal{G}_1$ through $\mathcal{G}_4$.}
    \label{fig:sim1-mean}
\end{figure}

\textbf{Task~2: parametric regression.} 
We next assessed inference for regression parameters under missing data. The data-generating process from Task 1 was modified by removing the direct effect of $X_1$ on $X_3$, inducing the conditional independence $X_1 \ci X_3 \,|\, X_2$ in all mDAGs. We tested this null hypothesis using Wald tests in a correctly specified linear model for $X_3$ given $X_1$ and $X_2$. Performance was evaluated using bias, root mean squared error, type I error, and 95\% confidence interval coverage, with results summarized in Table~\ref{table:sim2}. For detailed DGP and estimation derivations see Appendices~\ref{appsubsub:dgp-sim2}, \ref{appsub:sim-id} and \ref{appsub:sim-estimation-regression}, respectively.

\newcommand{\huxb}[2]{}
\providecommand{\huxvb}[2]{\color[RGB]{#1}\vrule width #2pt}
  \providecommand{\huxtpad}[1]{\rule{0pt}{#1}}
  \providecommand{\huxbpad}[1]{\rule[-#1]{0pt}{#1}}

\begin{table}[t]
\centering
\captionsetup{justification=centering,singlelinecheck=off}
\caption{Results for conditional independence tests across methods and graphs.}
 \setlength{\tabcolsep}{0pt}
\resizebox{1\textwidth}{!}{
\renewcommand{\arraystretch}{1.1}
 \setlength{\extrarowheight}{0pt}%
 \setlength{\lineskip}{0pt}\setlength{\lineskiplimit}{0pt}%
 \setlength{\tabcolsep}{0pt}
 \setlength{\arrayrulewidth}{0.5pt}
\begin{tabular}{l l l l l l l l l l l l l l l l l l l l l l}

\thickline

\hhline{>{\huxb{0, 0, 0}{1}}->{\huxb{0, 0, 0}{1}}->{\huxb{0, 0, 0}{1}}->{\huxb{0, 0, 0}{1}}->{\huxb{0, 0, 0}{1}}->{\huxb{0, 0, 0}{1}}->{\huxb{0, 0, 0}{1}}->{\huxb{0, 0, 0}{1}}->{\huxb{0, 0, 0}{1}}->{\huxb{0, 0, 0}{1}}->{\huxb{0, 0, 0}{1}}->{\huxb{0, 0, 0}{1}}->{\huxb{0, 0, 0}{1}}->{\huxb{0, 0, 0}{1}}->{\huxb{0, 0, 0}{1}}->{\huxb{0, 0, 0}{1}}->{\huxb{0, 0, 0}{1}}->{\huxb{0, 0, 0}{1}}->{\huxb{0, 0, 0}{1}}->{\huxb{0, 0, 0}{1}}->{\huxb{0, 0, 0}{1}}->{\huxb{0, 0, 0}{1}}-}

\multicolumn{1}{!{\huxvb{0, 0, 0}{0}}l!{\huxvb{0, 0, 0}{0}}}{\huxtpad{-1pt + 1em}\raggedright \hspace{-1pt} \textbf{{\fontsize{7.5pt}{9pt}\selectfont }} \hspace{-1pt}\huxbpad{-1pt}} &
\multicolumn{1}{c!{\huxvb{0, 0, 0}{0}}}{\huxtpad{-1pt + 1em}\centering \hspace{-1pt} \textbf{{\fontsize{7.5pt}{9pt}\selectfont }} \hspace{-1pt}\huxbpad{-1pt}} &
\multicolumn{5}{c!{\huxvb{0, 0, 0}{0}}}{\huxtpad{-1pt + 1em}\centering \hspace{-1pt} \textbf{{\fontsize{7.5pt}{9pt}\selectfont G1}} \hspace{-1pt}\huxbpad{-1pt}} &
\multicolumn{5}{c!{\huxvb{0, 0, 0}{0}}}{\huxtpad{-1pt + 1em}\centering \hspace{-1pt} \textbf{{\fontsize{7.5pt}{9pt}\selectfont G2}} \hspace{-1pt}\huxbpad{-1pt}} &
\multicolumn{5}{c!{\huxvb{0, 0, 0}{0}}}{\huxtpad{-1pt + 1em}\centering \hspace{-1pt} \textbf{{\fontsize{7.5pt}{9pt}\selectfont G3}} \hspace{-1pt}\huxbpad{-1pt}} &
\multicolumn{5}{c!{\huxvb{0, 0, 0}{0}}}{\huxtpad{-1pt + 1em}\centering \hspace{-1pt} \textbf{{\fontsize{7.5pt}{9pt}\selectfont G4}} \hspace{-1pt}\huxbpad{-1pt}} \tabularnewline
\end{tabular}}\label{table:sim2}

\end{table}

Under MAR $\mathcal G_1$, all methods achieved nominal performance. Under MNAR $\mathcal G_2$, only the proposed method maintained negligible bias, near-nominal type I error, and correct coverage. Competing methods exhibited inflated type I error and poor coverage due to spurious associations induced by conditioning on observed data. For the more complex MNAR models $\mathcal G_3$ and $\mathcal G_4$, complete-case analysis and the proposed method performed well, while imputation-based methods failed in a model-dependent manner. The unbiasedness of complete-case analysis follows from the imposed graphical structure; see Subsection~\ref{appsub:sim-estimation-regression} for a detailed discussion. Overall, while the validity of competing approaches depended on the specific missingness structure, the proposed method consistently achieved correct inference across all settings considered.

\textbf{Task~3: causal effect estimation.} 
Finally, we evaluated estimation of the causal effect of a binary treatment $X_2$ on a continuous outcome $X_3$ using a g-formula (back-door adjustment) controlling for $X_1$. In this task, $X_2$ was generated as a binary variable, violating the multivariate normal assumption underlying Amelia. Results are summarized in Figure~\ref{fig:sim3-causal}. For detailed DGP and estimation derivations see Appendices~\ref{appsubsub:dgp-sim3}, \ref{appsub:sim-id} and \ref{appsub:sim-estimation-causal}, respectively.

\begin{figure}[t]
    \centering
    \includegraphics[width=1\linewidth,clip, trim=0 5 0 5]{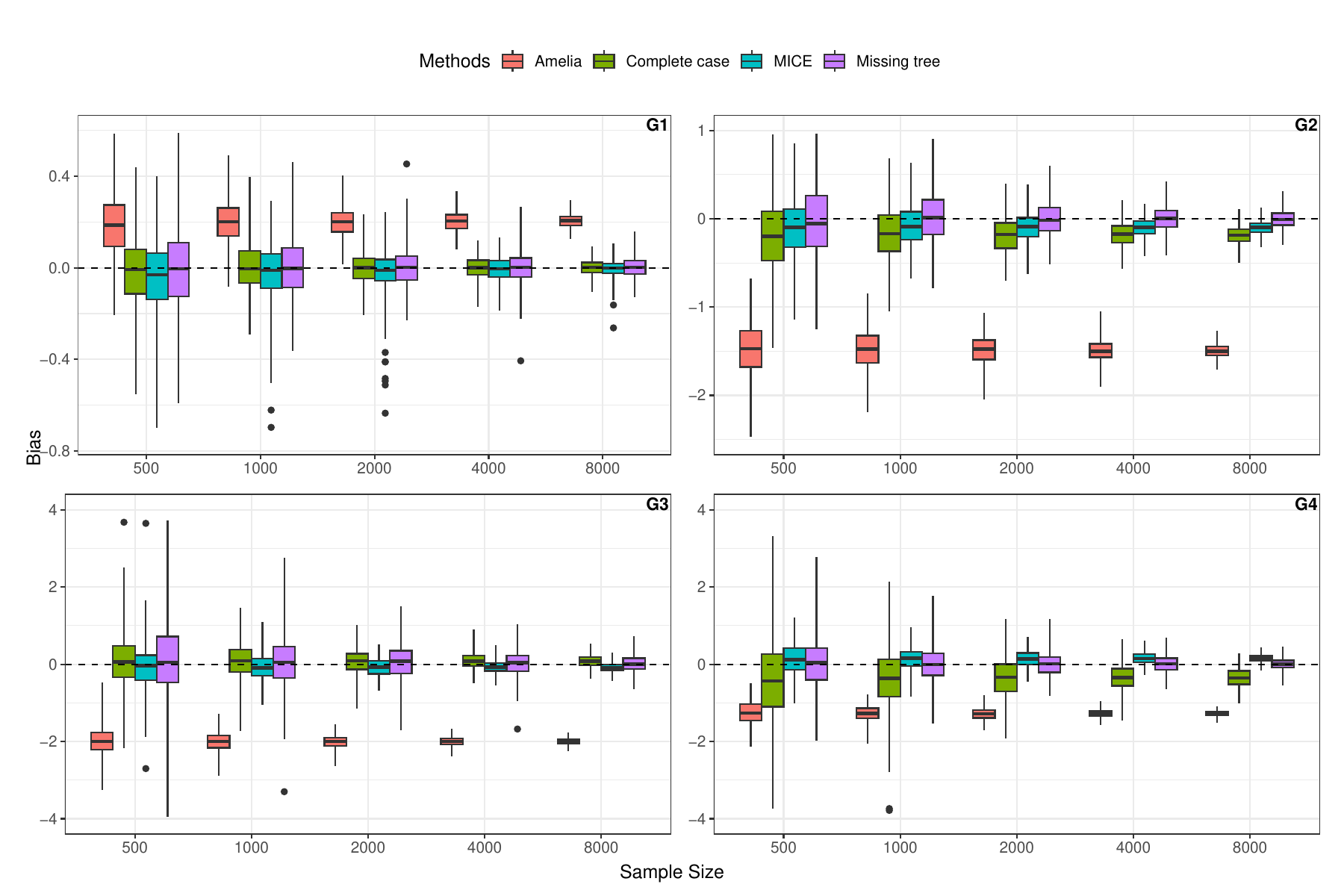}
    \caption{Simulation results for estimating an average causal effect. Missing data are handled using four methods: Amelia, complete-case analysis, MICE, and the proposed tree-based method. Panels corresponds to data generated under mDAGs $\mathcal{G}_1$ through $\mathcal{G}_4$.}
    \label{fig:sim3-causal}
\end{figure}

Across all mDAGs, the proposed method yielded estimates with negligible bias. Amelia produced biased estimates under all models, with bias persisting as sample size increased. Complete-case analysis and MICE were unbiased under the MAR model $\mathcal G_1$, but yielded biased estimates under the MNAR models $\mathcal G_2$ through $\mathcal G_4$.

\section{Data application}
\label{sec:application}

We illustrate the proposed framework using data from the \href{https://services.fsd.tuni.fi/catalogue/FSD3185}{Student Feedback Survey for Bachelor Graduates 2016}, a national survey of 11{,}708 Finnish university students who completed a Bachelor’s degree or studied for three years in programs without one. The survey collects information on students' academic experiences, study financing, and well-being. 

Our analysis focuses on two questions concerning study financing. Students were asked whether they funded their studies through student loans and through personal income from work, with response options indicating complete, partial, or no funding. These variables, denoted by $X_1$ and $X_2$, exhibit missingness rates of $14.7\%$ and $8.2\%$, respectively. The relationship between $(X_1, X_2)$ and their missingness indicators is represented by the mDAG shown in Appendix Figure~\ref{fig:supp_figs}(b), which was used and justified by \cite{tikka2024full}. Under this graphical model, the target law is identified. Our objective is to estimate the parameters of the target law, as well as the missingness mechanism. Since both variables are discrete, all relevant components are correctly specified using saturated regressions.

\begin{table}[t]
\centering
\caption{Parameter estimates with 95\% confidence intervals for the real data application. The estimate of $p(R_1 = 1)$ is $0.853\,(0.847, 0.860)$ and is omitted from the table for brevity.}
\label{table:realdata}
\setlength{\tabcolsep}{8pt}
\resizebox{0.95\linewidth}{!}{%
\begin{tabular}{p{4cm}ccc}
\toprule
 & $x_1 = 1$ & $x_1 = 2$ & $x_1 = 3$ \\
\midrule
$p(X_1=x_1)$ & 0.157 (0.149, 0.164) & 0.443 (0.433, 0.453) & 0.401 (0.391, 0.410) \\
\midrule
$p(X_2=1 \,|\, x_1)$ & 0.370 (0.344, 0.395) & 0.432 (0.417, 0.446) & 0.473 (0.458, 0.489) \\
$p(X_2=2 \,|\, x_1)$ & 0.538 (0.512, 0.564) & 0.515 (0.515, 0.545) & 0.439 (0.423, 0.454) \\
$p(X_2=3 \,|\, x_1)$& 0.093 (0.078, 0.108) & 0.039 (0.033, 0.044) & 0.088 (0.079, 0.097) \\
\midrule
$p(R_2 = 1 \,|\, R_1 = 1, x_1)$ & 0.902 (0.887, 0.916) & 0.969 (0.963, 0.974) & 0.995 (0.993, 0.997) \\
\bottomrule
\end{tabular}
}
\end{table}

Table~\ref{table:realdata} reports point estimates and 95\% confidence intervals for the target law and missingness parameters, closely matching the maximum likelihood results of \cite{tikka2024full}.

\section{Discussion}
\label{sec:conc}

Commonly used complete-case, imputation-based, and EM-based methods can perform poorly under missing-not-at-random mechanisms because they are typically agnostic to the structure of the missingness process. In contrast, the approach developed in this paper explicitly tailors identification and estimation to a specified missingness mechanism. By encoding assumptions about missingness through a graphical model and using this structure to guide both identification and weighting-based estimation, our framework makes transparent when nonparametric identification and valid inference is possible and why it may fail. This perspective highlights that robustness to MNAR mechanisms does not come from ignoring missingness structure, but rather from leveraging it directly in a principled way.

From an identification standpoint, several important directions for future work remain. First, while this paper focuses on identification of the target law, full law identification (via our intervention/weighted based arguments) remains an open problem in general graphical missing data models. Second, our framework treats the target law as unrestricted; when additional assumptions such as independence constraints on the target law are available, they can facilitate identification and should be incorporated directly into the algorithm. Finally, the proposed identification algorithm is sound but not complete: there exist settings in which the algorithm concludes non-identification even though identification is possible via alternative arguments. For example, in the mDAG shown in Appendix Figure~\ref{fig:supp_figs}(a), our algorithm fails to identify the target law due to non-identification of a propensity score, yet it can be identified using an odds ratio parameterization, as shown by \cite{nabi20completeness}. At the same time, odds ratio parameterizations are known to fail in the presence of colluder structures, which are naturally accommodated by the present framework. An important open direction is therefore to study whether different parameterizations and identification strategies can be systematically combined to obtain procedures that are both sound and complete.

A second set of open directions concerns estimation and model selection. While this paper develops recursive inverse probability weighted estimators based on parametric models for the missingness mechanism, future work should investigate more flexible semiparametric and machine learning–based estimators, together with influence function–based inference. In addition, the current framework assumes that the missingness mechanism is specified a priori that is Markov relative to a conditional mDAG. Developing methods that combine identification and estimation with model selection or discovery of the missingness structure from data is an important and challenging direction. Progress along these lines would further broaden the practical applicability of graphical approaches to missing data, especially in complex observational studies where the missingness mechanism is only partially understood.


\vspace{1cm}
\begingroup
\renewcommand{\baselinestretch}{0.97}
\selectfont  
\setlength{\bibsep}{10pt}    
\bibliography{references}
\endgroup

\newpage

\appendix

\setcounter{figure}{0}
\renewcommand{\thefigure}{E.\arabic{figure}}

\setcounter{algorithm}{0}
\renewcommand{\thealgorithm}{E\arabic{algorithm}}
\renewcommand{\theHalgorithm}{E.\arabic{algorithm}}

\section*{\Large Appendix}

The appendix is organized as follows. 
Appendix~\ref{app:notation} introduces a glossary of notation and terminology used throughout the paper.
Appendix~\ref{app:examples_details} presents additional examples and technical details that supplement the identification results in Sections~\ref{sec:hoops} and~\ref{sec:identification}, and the estimation results in Section~\ref{sec:estimation}. 
Appendix~\ref{app:id_examples} contains worked examples of the tree-based identification algorithm. 
Appendix~\ref{app:est_examples} contains worked examples with step-by-step derivations for example estimation of means, regression coefficients, and causal effects. 
Appendix~\ref{app:est_alg} provides pseudocode for the proposed estimation algorithms for propensity score estimations outlined in Section~\ref{subsec:est-ps} and estimation of functionals of the target law outlined in Section~\ref{subsec:est-target}. 
Appendix~\ref{app:proofs} contains proofs of the main theoretical results.
Appendix~\ref{app:sims} provides full details of the simulation study, including data-generating processes for all graphs and tasks, additional figures, and supplementary tables.

\section{Glossary of terms and notations} 
\label{app:notation} 

\begin{table}[H]
\begin{center}
\caption{\centering Glossary of terms and notations}
\label{tab:notations}
\addtolength{\tabcolsep}{8pt}
\resizebox{\textwidth}{!}{
{\small
\begin{tabular}{ ll} 
    \hline  
    \textbf{Symbol}     & \textbf{Definition}  
    \\ \hline 
    $X,X_k$   & Vector of variables, $k$-th element of $X$ 
    \\
    $p(X),\mathcal{M}_X$ & Target law, model of $p(X)$
    \\
    $R,R_k$  & Missingness indicators of $X,X_k$   
    \\
    $R_k=1/R_k=0$  &  $X_k$ missing/observed
    \\
    $p(R\mid X),\mathcal{M}_{R\mid X}$ & Missingness mechanism, model of $p(R\mid X)$ 
    \\
    $p(X,R)$ & Full law
    \\
    $\mathcal{M}=\mathcal{M}_X \otimes \mathcal{M}_{R\mid X}$ & Model of $p(X,R)$ 
    \\
    $X^*,X^*_k$ & Coarsened version of $X,X_k$
    \\
    $p(X^*,R)$ & Observed data law 
    \\
    $\pa_\G(R_k)$ & Parents of $R_k$ in $\G$
    \\
    $\de_\G(R_k)$ & Descendants of $R_k$ in $\G$, including $R_k$ 
    \\
    $\nd_\G(R_k)$ & Non-descendants of $R_k$, defined as $X\cup R\backslash \de_\G(R_k)$
    \\
    $\theta(p(X))$ & Particular functional of $p(X)$ 
    \\
    $p(\cdot)|_{R_j=1}$ & Evaluation of $p(\cdot)$ at $R_j=1$
    \\
    $\pi_k(\pa_\G(R_k))$ & Propensity score of $R_k$, defined as $p(R_k=1\mid \pa_\G(R_k))$ 
    \\
    $\calS^x_k$ & Counterfactual-induced selection set, defined as $\{R_j\in R:X_j\in\pa_\G(R_k)\}$
    \\
    $\calS^r_k$ & Indicator-induced selection set 
    \\
    $\calS_k=\calS^x_k\cup \calS^r_k$ & Selection set
    \\
    $\R^p_k=\calS^x_k\cap \de_\G(R_k)$ & Problematic set 
    \\
    $\doo(R_j=1)$ & An intervention on $R_j$ setting it to 1
    \\
     $\mathcal S_{j \, \downarrow \, k} \coloneqq \mathcal S_j \cap \pa_\G(R_k)$ & Selection propagated from $\doo(R_j)$ to $R_k$ 
    \\
    $\T_k$ & Identification tree associated with $R_k$
    \\
    $\F$ & Forest that collects all trees 
    \\
    $\mathcal{D}$ & Collection of indicators whose propensity score is not identified 
    \\
    $\tau$ & A valid reversed topological order on the mDAG $\G$ 
    \\
    $\mathcal C^{\text{dir}}_{k,k}$ & Colluder descendants of $R_k$, defined as $\{R_j \in \de_\G(R_k) \, : \, X_k \in \pa_\G(R_j)\}$ 
    \\
    $R^*$ & Candidate intervention set for $R_k$ 
    \\
    $\mathcal T_k \coloneqq \ch_{\mathbb T_k}(R_k)$ & Children of $R_k$ in $\T_k$
    \\
    $\widetilde{\mathcal S}_k
    \coloneqq
    \mathcal S_k^x
    \cup
    \bigcup_{R_j\in \mathcal T_k} \mathcal S_j$ & Pre-selection set for $R_k$ 
    \\
    $p(\cdot \,|\, \doo(\mathcal T_k=1))$ & Post-intervention distribution where indicators in $\mathcal T_k$ are intervened on 
    \\
   $\mathcal R_k^d$ & Indicator in $\widetilde{\mathcal S}_k \backslash \pa_\G(R_k)$ which is dependent on $R_k$ given $\pa_\G(R_k)$ in $p(\cdot \,|\, \doo(\mathcal T_k=1))$  
    \\
    $\mathcal C^{\text{dir}}_{k,j}$ &  Descendants of $R_k$ that selection on $R_j$, defined as $\{R_i\in \de_\G(R_k)\mid X_j\in \pa_\G(R_i)\}$
    \\
    $\mathcal C_k
    \coloneqq
    \bigcup_{R_j\in \mathcal R_k^d} \mathcal C^{\text{dir}}_{k,j}$ & Descendants of $R_k$ that select on $\mathcal R^d_k$ 
    \\
    $\mathcal B$ & Collection of pruned branches
    \\
    $\phi^{\mathcal G}_{R_i}$ & Graphical fixing operation applied to $R_i$ on an mDAG $\mathcal G$ 
    \\
    $\phi^p_{R_i}\{p\}$ & Probabilistic fixing operator applied to $R_i$ in $p(X,R)$
    \\
    $\sigma_k=(s_1,\ldots,s_m)$ & Any ordering of $\mathcal T_k$ consistent with $\tau$ 
    \\
    $p_{\mathbb T_k}=\phi^p_{\sigma_k}\{p\}$ & The post-intervention distribution induced by $\mathbb T_k$
    \\
    $\mathcal P$ & Collection of fitted propensity score models 
    \\
    $\mathcal E$ & Collection of estimating equations
    \\
    $W_k(\theta_{\mathcal T_k})$ & Inverse propensity weight for estimating $\pi_k$, defined in \eqref{eq:Wk} 
    \\
    $\mathcal P_k$  &  Collection of re-fitted propensity score models tailored for estimating $\pi_k$
    \\
    $\mathcal E_k$ & Collection of re-constructed estimating equations tailored for estimating $\pi_k$  
    \\
    $\boldsymbol{\theta}_k$ &  Stacked parameter vector relevant for estimating $\theta_k$
    \\
    $\boldsymbol{\Psi}_k(\boldsymbol{\theta}_k)$ & Stacked estimating function corresponding to $\boldsymbol{\theta}_k$ 
    \\
    $M(\tilde{X}; \, \theta)$ & Estimating function for $\theta$ 
    \\
    $\tilde{X}$ &  Collection of variables required to evaluate $M$
    \\
    $\tilde{R}$ & Missingness indicators for $\tilde{X}$ 
    \\
    $\operatorname{cl}(\cdot)$ & Closure operation defined as $\operatorname{cl}(A)\;\coloneqq\; A \cup \bigcup_{R_i\in A} \mathcal S_i$
    \\
    $\mathcal R$ & Smallest set containing $\tilde{R}$ and is closed under inclusion of $\calS_i$ 
    \\
    $\theta_{\mathcal R}\coloneqq (\theta_i)_{R_i\in \mathcal R}$ &  Collection of parameters indexing propensity score $R_i\in\mathcal R$
    \\
    \hline
\end{tabular}
}
}
\end{center}
\end{table}

\section{Additional examples and details} 
\label{app:examples_details} 

\subsection{Examples illustrating key identification concepts} 
\label{app:id_examples} 

\subsubsection{Identification example via Figures~\ref{fig:ex_tree}(a, d)}

The mDAG $\G_1$ in Figure~\ref{fig:ex_tree}(a) is an example where the tree-prune procedure never gets executed. The algorithm begins with $R_1$, for which $\pi_1\coloneqq p(R_1 = 1 \,|\, R_2, R_3)$ is directly observed and therefore identified.
It then proceeds to either $R_2$, followed by $R_3$, or vice versa. Since neither variable is a descendant of the other, both orders are valid. Their propensity scores are identified analogously, so we focus on $R_2$ for illustration. $\pi_2\coloneqq p(R_2 = 1 \,|\, X_1, X_3, R_4)$, with $\calS^x_2 = \{R_1, R_3\}$ and $\R^p_2 = \{R_1\}$. An intervention on $R_1$ is applied to invoke $R_2 \perp R_1 \,|\, \pa_\G(R_2)$ in the resulting post-intervention distribution. This intervention does not impose any additional selection ($\calS_1 = \emptyset$), thus $\widetilde{\calS}_2=\calS^x_2$, and consequently $\calS^r_2 = \widetilde{\calS}_2\cap\pa_\G(R_2)= \emptyset$. Thus, $\pi_2$ is identified as a full conditional distribution. 
Finally, the algorithm considers $R_4$, for which $\pi_4\coloneqq p(R_4\,|\, X_1,X_2,X_3)$, with $\calS^x_4=\R^p_4=\{R_1,R_2,R_3\}$. $\pi_4$ is identified by intervening on all indicators in $\R^p_4$. Since $\R^d_k = \emptyset$, no pruning is needed when appending constructed subtrees to the tree for $R_4$. As a result, all propensity scores, each evaluated at its corresponding indicator evaluation set which in this case is empty, are identified, and hence the target law is identified. Key definitions of the identification procedure, including $\calS^x_k,\, \R^p_k,\, \C^{\text{dir}}_{k,k},\, \ch_{\T_k}(R_k),\, \widetilde{\calS}_k,$ and $\calS_k$ for each indicator $R_k$, are summarized in Appendix Table~\ref{apptab:ex-tree1}.

\begin{table}[t]
\centering
\renewcommand{\arraystretch}{1.3}  
\caption{Key definitions used in the tree-based identification algorithm, illustrated using the mDAG in Figure~\ref{fig:ex_tree}(a).}
\label{apptab:ex-tree1}
\resizebox{\textwidth}{!}{\begin{tabular}{lcccccccc}
\toprule
$R_k$ & \textbf{Propensity scores} & $\mathcal{S}^x_k$ & $\R^p_k$ & $C^{\text{dir}}_{k,k}$ & $\ch_{\T_k}(R_k)$ & $\tilde{\mathcal{S}}_k$ & $\mathcal{S}^r_k$ & $\mathcal{S}_k$ \\
\midrule
$R_1$ & $p(R_1 | R_2, R_3)$ & $\emptyset$ & $\emptyset$ & $\emptyset$ & $\emptyset$ & $\emptyset$ & $\emptyset$ & $\emptyset$ 
\\
$R_2$ & $p(R_2 | X_1, X_3, R_4)$ & $\{R_1, R_3\}$ & $\{R_1\}$ & $\emptyset$ & $\{R_1\}$ & $\{R_1, R_3\}$ & $\emptyset$ & $\{R_1, R_3\}$ 
\\
$R_3$ & $p(R_3 | X_1, X_2, R_4)$ & $\{R_1, R_2\}$ & $\{R_1\}$ & $\emptyset$ & $\{R_1\}$ & $\{R_1, R_2\}$ & $\emptyset$ & $\{R_1, R_2\}$ 
\\
$R_4$ & $p(R_4 | X_1, X_2, X_3)$ & $\{R_1, R_2, R_3\}$ &  $\{R_1, R_2, R_3\}$ &$\emptyset$ & $\{R_1, R_2, R_3\}$ & $\{R_1, R_2, R_3\}$ & $\emptyset$ & $\{R_1, R_2, R_3\}$ 
\\
\bottomrule
\end{tabular}
}
\end{table}

\subsubsection{Identification example via Figures~\ref{fig:ex_tree}(b, e)}

The mDAG $\G_2$ in Figure~\ref{fig:ex_tree}(b) is an example where tree-prune gets executed. With the pruning, all selection biases can be avoided and all the propensity scores under proper evaluation are identified. 
The identification algorithm begins with $R_1$, followed by $R_2$, whose propensity scores are easily identified via the associational irrelevancy criterion. We then focus on $R_3$, for which $\pi_3 \coloneqq p(R_3=1 \,|\, X_1, R_4, R_5, R_6)$, with $\R^p_3 = \{R_1\}$. Intervening on both $R_1$ and $R_2$ makes $\pi_3$ evaluated at $\calS^r_3 = \{R_4, R_5\}$ identified. In this evaluation, $R_4$ results from the intervention on $R_1$, whereas $R_5$ results from the intervention on $R_2$. The algorithm can then proceed to either $R_4$ or $R_5$, and then to the other. For illustration, we first consider $R_4$, for which $\pi_4 \coloneqq p(R_4=1 \,|\, X_3)$ with $\R^{p}_4=\{R_3\}$ and $\C^{\text{dir}}_{4,4}=\{R_1\}$. Thus, only $R_2$ and $R_3$ are qualified candidates for intervention. The subtree for $R_2$ attaches directly to the tree of $R_4$, whereas the subtree for $R_3$ requires pruning. This is because, in $\T_3$, the branch $R_1 \in \C^{\text{dir}}_{4,4}$ induces selection on $R_4$, which propagates through $R_3$ since
$R_4\in\pa_\G(R_3)$. Consequently, $R_1$ must be pruned from $\T_3$ before $\T_3$ is appended to $\T_4$.
After pruning, the selection set of $R_3$, $\calS_3$, is updated from $\{R_1, R_4, R_5\}$ to $\{R_1, R_5\}$, and $\pi_4$ becomes identified.
The propensity score of $R_5$, evaluated at $\calS^r_5=\{R_6\}$, is identified analogously as that of $R_6$, and is therefore omitted.
Finally, the algorithm turns to $R_6$, for which $\pi_6 \coloneqq p(R_6=1 \,|\, X_3, X_4)$, with $\R^p_6=\{R_3\}$ and $\C^{\text{dir}}_{6,6}=\{R_5\}$. The qualified candidates for intervention are therefore $\{R_1, R_2, R_3\}$ (note that $R_4$ is not a descendant of $R_6$ and thus excluded). In the initial assessment, identification is hindered by selection on $\R^d_6 = \{R_5\}$, which arises from branches $R_2$ and $R_3$. Intervening on $R_2$ directly selects on $R_5$, and thus $R_2$ must be pruned from the tree $\T_6$. In contrast, the selection of $R_5$ induced by intervening on $R_3$ arises because identifying $\pi_3$ involves an unnecessary intervention on $R_2$. As a result, $R_2$ is pruned in the subtree $\T_3$. After pruning, $\pi_6$ becomes identified. Appendix Table~\ref{apptab:ex-tree2} summarizes key definitions for the identification procedure.

\begin{table}[t]
\centering
\renewcommand{\arraystretch}{1.3}  
\caption{Key definitions used in the tree-based identification algorithm, illustrated using the mDAG in Figure~\ref{fig:ex_tree}(b).}
\label{apptab:ex-tree2}
\resizebox{\textwidth}{!}{\begin{tabular}{lcccccccc}
\toprule
$R_k$ & \textbf{Propensity scores} & $\mathcal{S}^x_k$ & $\R^p_k$ & $C^{\text{dir}}_{k,k}$ & $\ch_{\T_k}(R_k)$ & $\tilde{\mathcal{S}}_k$ & $\mathcal{S}^r_k$ & $\mathcal{S}_k$ \\
\midrule
$R_1$ & $p( R_1| X_4,R_2)$ & $\{R_4 \}$ & $\emptyset$ & $\emptyset$ & $\emptyset$ & $\{R_4\}$ & $\emptyset$ & $\{ R_4\}$ 
\\
$R_2$ & $p( R_2|X_5,R_3 )$ & $\{ R_5\}$ & $\emptyset$ & $\emptyset$ & $\emptyset$ & $\{ R_5\}$ & $\emptyset$ & $\{ R_5\}$ 
\\
$R_3$ & $p(R_3 | X_1,R_4,R_5,R_6)$ & $\{ R_1\}$ & $\{ R_1\}$ & $\emptyset$ & $\{ R_1,R_2\}$ & $\{ R_1,R_4,R_5\}$ & $\{R_4,R_5 \}$ & $\{ R_1,R_4,R_5\}$ 
\\
\multicolumn{2}{l}{$R_3$ ($R_3\rightarrow R_1$ pruned)} & $\{ R_1\}$ & $\{ R_1\}$ & $\emptyset$ & $\{ R_2\}$ & $\{ R_1,R_5\}$ & $\{R_5 \}$ & $\{ R_1,R_5\}$ 
\\
\multicolumn{2}{l}{$R_3$ ($R_3\rightarrow R_2$ pruned)} & $\{ R_1\}$ & $\{ R_1\}$ & $\emptyset$ & $\{ R_1\}$ & $\{ R_1,R_4\}$ & $\{R_4\}$ & $\{ R_1,R_4\}$ 
\\
$R_4$ & $p(R_4| X_3)$ & $\{ R_3\}$ & $\{ R_3\}$ & $\{ R_1\}$ & $\{ R_2,R_3\}$ & $\{R_3, R_5 \}$ & $\emptyset$ & $\{R_3 \}$ 
\\
$R_5$ & $p(R_5 |X_3,X_6,R_6 )$ & $\{ R_3,R_6\}$ & $\{ R_3\}$ & $\{ R_2\}$ & $\{ R_1,R_3\}$ & $\{R_1,R_3,R_4,R_6 \}$ & $\{ R_6\}$ & $\{ R_3,R_6\}$ 
\\
$R_6$ & $p( R_6|X_3,X_4 )$ & $\{ R_3,R_4\}$ & $\{R_3\}$ & $\{ R_5\}$ & $\{R_1,R_3 \}$ & $\{ R_1,R_3,R_4\}$ & $\emptyset$ & $\{ R_3,R_4\}$ 
\\
\bottomrule
\end{tabular}
}
\end{table}

\subsubsection{Identification example via Figures~\ref{fig:ex_tree}(c, f)}

The mDAG $\G_3$ in Figure~\ref{fig:ex_tree}(c) is an example where tree-prune gets executed. However, some selections cannot be avoided. Thus, the algorithm concludes the target law is not identified. We omit the discussion of $R_1$ through $R_4$, as the technique for identifying their propensity scores is the same as described above. We now focus on $R_5$, whose propensity score is $\pi_5 \coloneqq p(R_5=1 \,|\, X_2)$, with $\R^p_5 = \{R_2\}$ and $\C^{\text{dir}}_{5,5} = \{R_2\}$. The candidates for intervention are therefore $\{R_1, R_3, R_4\}$. The subtrees for $R_1$ and $R_3$ can be attached directly, whereas the subtree for $R_4$ is pruned at $R_2$ during the initial assessment of identifiability, where $\R^d_5 = \{R_5\}$. This pruning is necessary as intervenetion on $R_2$ selects on $R_5$, and this selection progagates through $R_4$ given that $R_5\in\pa_\G(R_4)$. After pruning, the propensity score for $R_4$ is no longer identified, thus $R_4$ is pruned from $\T_5$. The second assessment returns $\R^d_5 = \{R_4\}$, which arises from the intervention on $R_1$. This leads to pruning $R_1$, leaving $R_3$ as the only indicator on which an intervention is applied. However, an intervention on $R_3$ alone is insufficient to identify $\pi_5$, because in the third assessment we obtain $\R^d_5 = \{R_2\} = \R^p_5$, which implies that identification fails. Details are provided in Appendix Table~\ref{apptab:ex-tree3}.

\begin{table}[t]
\centering
\renewcommand{\arraystretch}{1.3}  
\caption{Key definitions used in the tree-based identification algorithm, illustrated using the mDAG in Figure~\ref{fig:ex_tree}(c).}
\label{apptab:ex-tree3}
\begin{tabular}{lcccccccc}
\toprule
$R_k$ & \textbf{Propensity scores} & $\mathcal{S}^x_k$ & $\R^p_k$ & $\C^{\text{dir}}_{k,k}$ & $\ch_{\T_k}(R_k)$ & $\tilde{\mathcal{S}}_k$ & $\mathcal{S}^r_k$ & $\mathcal{S}_k$ \\
\midrule
$R_1$ & $p(R_1 |X_4,R_2 )$ & $\{R_4\}$ & $\emptyset$ & $\emptyset$ & $\emptyset$ & $\{R_4\}$ & $\emptyset$ & $\{R_4\}$
\\
$R_2$ & $p( R_2|X_5,R_3,R_4 )$ & $\{R_5\}$ & $\emptyset$ & $\emptyset$ & $\emptyset$ & $\{R_5\}$ & $\emptyset$ & $\{R_5\}$
\\
$R_3$ & $p(R_3 |R_5 )$ & $\emptyset$ & $\emptyset$ & $\emptyset$ & $\emptyset$ & $\emptyset$ & $\emptyset$ & $\emptyset$
\\
$R_4$ & $p(R_4 |X_1,R_5 )$ & $\{R_1\}$ & $\{R_1\}$ & $\{R_1\}$ & $\{R_2\}$ & $\{R_1,R_5\}$ & $\{R_5\}$ & $\{R_1,R_5\}$
\\
$R_5$ & $p( R_5|X_2)$ not ID & $\{R_2\}$ & $\{R_2\}$ & $\{R_2\}$ & - & - & - & - 
\\
\bottomrule
\end{tabular}
\end{table}

\subsubsection{Identification example via Figures~\ref{fig:supp_figs}(c, d)}\label{subsubsec:update_prune_tree}
The mDAG in Appendix Figure~\ref{fig:supp_figs}(c) provides an example in which pruning requires updating a grown subtree to match its pruned version. We omit the discussion of $R_1$ through $R_4$, as constructing their trees is straightforward and does not involve pruning. We now focus on $R_5$, whose propensity score is $\pi_5\coloneqq p(R_5=1\mid X_2)$, with $\mathcal R^p_5=\{R_3\}$ and $\mathcal C^{\text{dir}}_{5,5}=\{R_2\}$. The candidates for intervention are therefore $\{R_1,R_3,R_4\}$. The initial assessment of identifiability yields $\mathcal R^d_5={R_5}$, indicating that intervening on these candidates induce selection on $R_5$ itself. This selection arises from intervening on $R_2$. To remove this selection, the edge $R_3\rightarrow R_2$ in subtree $\T_3$ is pruned, and $\pi_3$ remains identifiable. The resulting pruned $\T_3$ is added to the pruned-tree collection $\mathcal B$. Similarly, the edge $R_4\rightarrow R_2$ in subtree $\T_4$ is pruned. In addition, the subtree $\T_3$ attached to $R_3$ in $\T_4$ is replaced to match its pruned version in $\mathcal B$. Otherwise, selection induced by $R_2$ would propagate to $R_3$ and further through $R_4$, thereby obstructing identification of $\pi_5$. Collecting pruned subtrees and reusing them when needed allows the algorithm to restrict attention to pruning children of $R_k$, or their children, without requiring consideration of deeper descendants. Selection arising from further descendants is handled automatically by replacing the relevant subtrees with their pruned versions when available in $\mathcal B$. This design substantially improves the feasibility of the identification algorithm, as exhaustively tracing selection through all descendants is computationally expensive, especially in large graphs. Details are provided in Appendix Table~\ref{apptab:ex-tree4}.
\begin{table}[t]
\centering
\renewcommand{\arraystretch}{1.3}  
\caption{Key definitions used in the tree-based identification algorithm, illustrated using the mDAG in Figure~\ref{fig:supp_figs}(c).}
\label{apptab:ex-tree4}
\resizebox{\textwidth}{!}{\begin{tabular}{lcccccccc}
\toprule
$R_k$ & \textbf{Propensity scores} & $\mathcal{S}^x_k$ & $\R^p_k$ & $\C_{k,k}$ & $\ch_{\T_k}(R_k)$ & $\tilde{\mathcal{S}}_k$ & $\mathcal{S}^r_k$ & $\mathcal{S}_k$ \\
\midrule
$R_1$ & $p(R_1 |R_2 )$ & $\emptyset$ & $\emptyset$ & $\emptyset$ & $\emptyset$ & $\emptyset$ & $\emptyset$ & $\emptyset$
\\
$R_2$ & $p( R_2|X_5,R_3)$ & $\{R_5\}$ & $\emptyset$ & $\emptyset$ & $\emptyset$ & $\{R_5\}$ & $\emptyset$ & $\{R_5\}$
\\
$R_3$ & $p(R_3 |X_1, R_4, R_5 )$ & $\{R_1\}$ & $\{R_1\}$ & $\emptyset$ & $\{R_1,R_2\}$ & $\{R_1,R_5\}$ & $\{R_5\}$ & $\{R_1,R_5\}$
\\
$R_3$ & ($R_3\rightarrow R_2$ pruned) & $\{R_1\}$ & $\{R_1\}$ & $\emptyset$ & $\{R_1\}$ & $\{R_1\}$ & $\emptyset$ & $\{R_1\}$
\\
$R_4$ & $p(R_4 |X_3,R_5 )$ & $\{R_3\}$ & $\{R_3\}$ & $\emptyset$ & $\{R_1,R_2,R_3\}$ & $\{R_1,R_3,R_5\}$ & $\{R_5\}$ & $\{R_1,R_3,R_5\}$
\\
$R_4$ & ($R_4\rightarrow R_2$ \& $\T_3$ pruned) & $\{R_3\}$ & $\{R_3\}$ & $\emptyset$ & $\{R_1,R_3\}$ & $\{R_1,R_3\}$ & $\{R_5\}$ & $\{R_1,R_3\}$
\\
$R_5$ & $p( R_5|X_3)$ & $\{R_3\}$ & $\{R_3\}$ & $\{R_2\}$ & $\{R_1,R_3,R_4\}$ & $\{R_1,R_3\}$ & $\emptyset$ & $\{R_3\}$ 
\\
\bottomrule
\end{tabular}
}
\end{table}

\subsection{Examples illustrating key estimation concepts} 
\label{app:est_examples} 

We provide detailed illustrations of the estimation procedures developed in Section~\ref{subsec:est-target}. These examples illustrate the closure construction for $\mathcal R$, the associated inverse probability weighted estimating equations, and the stacked variance calculations within concrete graphical models. All examples rely on the intervention trees and selection sets produced by Algorithm~\ref{alg:ID}; no additional identification arguments are introduced.

We focus on three inferential tasks: (i) \textbf{estimation of means}, (ii) \textbf{parametric regression coefficients}, and (iii) \textbf{causal effects}. The first two examples illustrate estimation of means in settings where the full target law is not identifiable, highlighting both cases in which a mean remains identifiable and cases in which identification fails due to unavoidable selection. The remaining examples demonstrate how the same framework extends to regression and causal parameters once the appropriate closure of missingness indicators is taken into account.

Throughout, $\tilde{R}$ denotes the set of indicators corresponding to variables required to evaluate a target estimating function, and $\mathcal R$ denotes its closure under the selection sets $\{\mathcal S_i\}$ as defined in \eqref{eq:R_t_def}. The estimating function $\Psi$ is constructed according to \eqref{eq:Psi_t_def}, and asymptotic variances are obtained from the stacked estimating equations described in Section~\ref{subsec:est-target}.

\subsubsection{Estimation task: Mean of a partially observed variable}
\label{app:est_ex_mean_v1}

We begin with estimation of the mean $\theta=\E(X_3)$ in the mDAG shown in Figure~\ref{fig:ex_tree}(c). The parameter $\theta$ is defined as the unique solution to the moment condition
\[
\E\{M(\tilde{X};\theta)\}=0,
\qquad
M(\tilde{X};\theta)=X_3-\theta.
\]

Although the target law $p(X)$ is not identified in this mDAG because the propensity score $\pi_5$ is not identifiable, the mean $\E(X_3)$ remains identifiable. Evaluating $M(\tilde{X};\theta)$ requires observing $X_3$, thus we initialize $\tilde{R}=\{R_3\}$. The selection set associated with $R_3$ is empty, implying that the closure construction in \eqref{eq:R_t_def} stabilizes immediately and yields $\mathcal R=\{R_3\}$.

The estimating function for $\theta$ therefore takes the inverse probability weighted form
\begin{align}
\Psi(X,R;\theta,\theta_3)
=
\frac{\I(R_3=1)}{\pi_3(\pa_\G(R_3);\theta_3)\vert_{\mathcal S_3^r=1}}
\,(X_3-\theta),
\label{eq:mean-X3-extree3}
\end{align}
where $\pi_3$ denotes the propensity score for $R_3$. The corresponding propensity score estimating equation is
\begin{align}
\Psi_3(X,R;\theta_3)
=
R_3-\pi_3(\pa_\G(R_3);\theta_3).
\label{eq:propensity3-extree3}
\end{align}

An estimator $\hat\theta$ is obtained by solving $P_n\Psi(X,R,\theta,\hat\theta_3)=0$, where $\hat\theta_3$ solves $P_n\Psi_3(X,R;\theta_3)=0$. The asymptotic variance of $\hat\theta$ follows from the stacked estimating equations $\boldsymbol{\Psi}=(\boldsymbol{\Psi}_3,\Psi)'$ as described in Section~\ref{subsec:est-target}, where $\boldsymbol{\Psi}_3=\{\Psi_3\}$.

By contrast, identification of the mean of other variables fails in this mDAG. For example, consider $\theta=\E(X_1)$. The initial indicator set is $A_0=\tilde{R}=\{R_1\}$. According to Appendix Table~\ref{apptab:ex-tree3}, $\pi_1$ can be evaluated only on $\mathcal S_1=\{R_4\}$, yielding $A_1=\operatorname{cl}(A_0)=\{R_1,R_4\}$. Next, $\pi_4$ can be evaluated only on $\mathcal S_4=\{R_1,R_5\}$, yielding $A_2=\operatorname{cl}(A_1)=\{R_1,R_4,R_5\}$. Since $\pi_5$ is not identified, the closure includes an indicator whose propensity score cannot be estimated, and no unbiased estimating equation for $\theta$ can be constructed. Thus $\E(X_1)$ is not identifiable in this graph.

\subsubsection{Estimation task: Mean requiring recursive closure}
\label{app:est_ex_mean_v2}

Consider the modification of the mDAG in Figure~\ref{fig:ex_tree}(c), shown in Appendix Figure~\ref{fig:sims_tree}(e), in which the edge $R_4\rightarrow R_2$ is replaced by $R_4\rightarrow R_3$. The corresponding intervention trees are shown in Appendix Figure~\ref{fig:sims_tree}(f).

Under this modification, additional propensity scores become identifiable. In particular, $\pi_4$, evaluated at $R_5=1$, is identified by intervening on $R_2$ and $R_3$, with the resulting selection set remaining $\mathcal S_4=\{R_1,R_5\}$. Moreover, $\pi_5$ becomes identifiable by intervening on $R_1$, $R_3$, and $R_4$, with pruning of $R_2$ from the subtree of $R_4$ to avoid inducing selection on $R_5$. The resulting selection set for $R_5$ is $\mathcal S_5=\{R_2\}$.

We again consider estimation of $\theta=\E(X_1)$. Starting from $A_0=\tilde{R}=\{R_1\}$, the closure construction proceeds as follows:
\[
A_1=\operatorname{cl}(A_0)=\{R_1,R_4\},\qquad
A_2=\operatorname{cl}(A_1)=\{R_1,R_4,R_5\},
\]
and incorporating $\mathcal S_5=\{R_2\}$ yields
\[
A_3=\operatorname{cl}(A_2)=\{R_1,R_2,R_4,R_5\}.
\]
Since $\mathcal S_2=\{R_5\}\subseteq A_3$, the recursion stabilizes and $\mathcal R=\{R_1,R_2,R_4,R_5\}$.

The estimating function for $\theta$ is therefore
\begin{align}
\Psi(X^*,R;\theta,\theta_{\mathcal R})
=
\frac{\I(\mathcal R=1)}{\prod_{R_i\in\mathcal R}\pi_i(\pa_\G(R_i);\theta_i)}
\,(X_1-\theta), 
\label{eq:mean-X1-extree3-ext}
\end{align}
with estimating functions for the propensity scores of $R_i\in\mathcal R$, defined by the intervention trees in Appendix Figure~\ref{fig:sims_tree}(f), given by
{\small\begin{align}
    &\Psi_1(X^*,R;\theta_1)  = \I(R_4=1)\ (R_1-\pi_1( X_4,R_2;\theta_1)).\label{eq:propensity1-extree3-ext}
    \\
    &\Psi_2(X^*,R;\theta_2)  = \I(R_5=1)\ (R_2-\pi_2( X_5,R_3;\theta_2)).\label{eq:propensity2-extree3-ext}
    \\
    &\Psi_4(X^*,R;\theta_4,\theta_2,\theta_3)  = \frac{\I(R_1=1,R_5=1) \ \I(R_2=1,R_3=1)}{\pi_2( X_5,R_3;\theta_2)\ \pi_3( R_4,R_5;\theta_3)}\ (R_4-\pi_4( X_1,R_5;\theta_4)).\label{eq:propensity4-extree3-ext}
    \\
    &\Psi_5(X^*,R;\theta_5,\theta_1,\theta_3,\theta_4)  = \frac{\I(R_2=1)\ \I(R_1=1,R_3=1,R_4=1)}{\pi_1( X_4,R_2;\theta_1)\ \pi_3( R_5;\theta_3)\ \pi_4( X_1,R_5;\theta_4)}\ (R_5-\pi_5( X_2;\theta_5)).\label{eq:propensity5-extree3-ext}
\end{align}}

The asymptotic variance of $\hat\theta$ is obtained from the stacked estimating equations
$\boldsymbol{\Psi} = (\boldsymbol{\Psi}_1, \boldsymbol{\Psi}_2, \boldsymbol{\Psi}_4, \boldsymbol{\Psi}_5, \Psi)'$,
where $\boldsymbol{\Psi}_k$ denotes the stacked estimating functions associated with $\Psi_k$.
The propensity scores $\pi_1$ and $\pi_2$ are identified via associational irrelevancy, so that $\boldsymbol{\Psi}_1=\{\Psi_1\}$ and $\boldsymbol{\Psi}_2=\{\Psi_2\}$.
In contrast, $\pi_4$ and $\pi_5$ are identified via causal irrelevancy, and their stacked systems additionally include estimating functions for intervened indicators and their descendants. Specifically, $\boldsymbol{\Psi}_4=\{\Psi_2,\Psi_3,\Psi_4\}$ and $\boldsymbol{\Psi}_5=\{\Psi_1,\Psi_3,\Psi_4,\Psi_5\}$, where $\Psi_3$ is given by \eqref{eq:propensity3-extree3}.
The $\Psi_4$ appearing in $\boldsymbol{\Psi}_4$ is defined in \eqref{eq:propensity4-extree3-ext}. In contrast, the $\Psi_4$ included in $\boldsymbol{\Psi}_5$, given by \eqref{eq:propensity4-extree3-ext-for-pi5}, is different and is specifically constructed for estimating $\pi_5$. In particular, $\pi_2$ is dropped from the inverse weight, as dictated by the intervention tree in Appendix Figure~\ref{fig:sims_tree}(f), where $R_2$ is pruned from the subtree of $R_4$ before it is appended to the tree of $R_5$.
{\small\begin{flalign}
    && \Psi_4(X,R;\theta_4,\theta_3)  &= \frac{\I(R_1=1) \ \I(R_3=1)}{\pi_3( R_4,R_5;\theta_3)}\ (R_4-\pi_4( X_1,R_5;\theta_4)). &\text{(tailored for estimating $\pi_5$)}
    \label{eq:propensity4-extree3-ext-for-pi5}
\end{flalign}}

\subsubsection{Estimation task: Parametric regression with missing covariates}
\label{app:est_ex_regpar}

Using the mDAG in Appendix Figure~\ref{fig:sims_tree}(e), consider estimation of the regression of $X_3$ on $(X_1,X_2)$ under the linear mean model
\[
\E(X_3\mid X_1,X_2)=\beta_0+\beta_1 X_1+\beta_2 X_2,
\]
with target parameter $\theta=(\beta_0,\beta_1,\beta_2)$. The corresponding moment condition is
\begin{align}
M(\tilde{X};\theta)
=
f(X_1,X_2)\,
\{X_3-(\beta_0+\beta_1 X_1+\beta_2 X_2)\},
\qquad
f(X_1,X_2)=(1,X_1,X_2)'.
\label{eq:M-regression-extree3-ext}
\end{align}

Evaluating $M$ requires observing $(X_1,X_2,X_3)$, so the initial indicator set is $\tilde{R}=\{R_1,R_2,R_3\}$. Applying the closure construction yields $\mathcal R=\{R_1,R_2,R_3,R_4,R_5\}$. The estimating function $\Psi$ follows from \eqref{eq:Psi_t_def}, and the asymptotic variance of $\hat\theta$ is obtained by stacking the estimating equations for all propensity scores in $\mathcal R$ together with $\Psi$. That is, $\boldsymbol{\Psi} = (\boldsymbol{\Psi}_1,\boldsymbol{\Psi}_2, \boldsymbol{\Psi}_3, \boldsymbol{\Psi}_4, \boldsymbol{\Psi}_5,\Psi)'$, where specifications of $\boldsymbol{\Psi}_1$ through $\boldsymbol{\Psi}_5$ are given in Subsections~\ref{app:est_ex_mean_v1} and ~\ref{app:est_ex_mean_v2}.

\subsubsection{Estimation task; Average causal effect}
\label{app:est_ex_causal}

Finally, we consider estimation of the average causal effect of a binary treatment $X_1$ on an outcome $X_3$. For clarity, we focus on the counterfactual mean $\theta=\E(X_3^{x_1})$ for a fixed treatment level $x_1\in\{0,1\}$. In the absence of missing data, $\theta$ is identified under standard causal assumptions using adjustment for the confounder $X_2$.

We assume that $\theta$, together with nuisance parameters $\theta_{trt}$ and $\theta_{or}$ indexing a treatment model and an outcome regression, respectively, is characterized by the moment condition
\begin{align}
M(\tilde{X};\theta,\theta_{trt},\theta_{or})
=
\begin{pmatrix}
f_{trt}(X_2)\,\{X_1-p(X_1\mid X_2;\theta_{trt})\} \\
f_{or}(X_1,X_2)\,\{X_3-\E(X_3\mid X_1,X_2;\theta_{or})\} \\
\frac{\I(X_1=x_1)}{p(X_1\mid X_2;\theta_{trt})}
\{X_3-\E(X_3\mid x_1,X_2;\theta_{or})\}
+\E(X_3\mid x_1,X_2;\theta_{or})
-\theta
\end{pmatrix}.
\label{eq:M-causal-extree3-ext}
\end{align}
In the function $M$, $f_{trt}(X_2)$ and $f_{or}(X_1,X_2)$ have dimensions matching $\theta_{trt}$ and $\theta_{or}$, respectively. We leave their specific forms unspecified, assuming only that they follow some parametric structure. For example, if the relationship between $X_3$ and the covariates is fully captured by a linear regression with main terms only, the second row of $M$ follows the form in \eqref{eq:M-regression-extree3-ext}, and $f_{or}$ can follow the form of $f$ in Subsection~\ref{app:est_ex_regpar} or any three-dimensional function of $X_1$ and $X_2$ that is feasible.

The missingness-adjusted estimating function $\Psi$ is obtained by substituting \eqref{eq:M-causal-extree3-ext} into \eqref{eq:Psi_t_def}, with $\mathcal R$ determined by the closure construction for the indicators required to evaluate $M$. Estimation and inference then follow from the general theory in Section~\ref{subsec:est-target}. 

\subsection{Estimation algorithm} 
\label{app:est_alg} 

See Algorithm~\ref{alg:propensity} for a summary description of our proposed \textit{recursive inverse probability weighted} estimators for the identified propensity scores $\{\pi_k\vert_{\mathcal S_k^r=1}\}_{k=1}^K$. 

See Algorithm~\ref{alg:est} for a summary description of estimation of generic parameters defined through moment conditions under the target law. 


\begin{algorithm}[H]
\setstretch{1.3}
	\caption{{\small  \textproc{Recursive IPW Estimation of Propensity Scores}$(\G(X, R, X^*), {\mathbb F},{\cal D})$}}
	\label{alg:propensity}
	\begin{algorithmic}[1]
		
		\vspace{0.25cm}		
		\State {\bf Estimation Procedure:} 
		\begin{itemize}
                \item Let $\mathcal{E}$ and $\mathcal{P}$ be the lists that collect the estimating equations and fitted models of propensity scores for all $R_k \notin \mathcal{D}$, respectively; initialize them to empty lists
                \item For $R_k\in R\setminus \mathcal{D}$ {\scriptsize (following a valid reverse topological order $\tau$)}:
                \begin{itemize}
                \item $\hat{\pi}_k\coloneqq \pi_k( \pa_\G(R_k);\hat{\theta}_k)|_{\calS^r_k=1}\leftarrow$ \blue{{\bf estimate-propensity}($R_k, \mathcal{P},\mathbb{F}$)}
                \item Add $\hat{\pi}_k$ to $\mathcal{P}$, and $\Psi_k$ to $\mathcal E$
                \end{itemize}
                \item RETURN $\mathcal{E},\mathcal{P}$
		\end{itemize}
        \par\vspace{0.1cm}
		\State \blue{{\bf estimate-propensity}($R_k,\mathcal{P},\mathbb{F}$)}
        \par\vspace{0.1cm}
            \begin{itemize}
                \item Assume the parameter $\theta_k$ indexing $\pi_k( \pa_\G(R_k);\theta_k)|_{\calS^r_k=1}$ is the unique solution to $\E(\Psi_k(X^*,R; \theta_k, \theta_{\mathcal T_k})=0$, where
                \par\vspace{-1cm}
                {\small\begin{align*}
                \Psi_k(X^*,R;\theta_k,\theta_{\mathcal T_k})
                \;\coloneqq\;
                \I(\widetilde{\mathcal S}_k=1)\,
                W_k(\theta_{\mathcal T_k})\,
                f_k(\pa_\G(R_k))\,
                \big\{R_k-\pi_k(\pa_\G(R_k);\theta_k)\big\}.
                \end{align*}}
                \par\vspace{-0.5cm}
                \begin{itemize}
                    \item $f_k(\pa_\G(R_k))$ is a function with the same dimension as $\theta_k$
                    \item {\small$W_k(\theta_{\mathcal T_k})=\prod_{R_i\in\mathcal T_k}(\I(R_i=1)/\pi_i(\pa_\G(R_i);\theta_i))$} is the inverse propensity weight. By convention, $W_k(\theta_{\mathcal T_k})=1$ when $\mathcal T_k=\emptyset$.
                \end{itemize}
                \item The estimator $\hat{\theta}_k$ is the solution to $P_n\ \Psi_k(X^*,R; \theta_k, \hat{\theta}_{\mathcal T_k})=0$
                
                {\small  \begin{itemize}
                    \item $\hat{\theta}_i\in\hat\theta_{\mathcal T_k}$ index the estimated propensity score for $R_i\in\ch_{\T_k}(R_k)$
                    \begin{itemize}
                        \item $\hat{\theta}_i$ is retrieved from $\mathcal{P}$ if $\T_i$ under $\T_k$ is not pruned, thus matches that in $\mathbb{F}$
                        \item Re-estimate $\hat{\theta}_i$ with $\Psi_i$ that respect the structure of $\T_i$ if it is pruned
                        \item $\mathcal{E}_k$ and $\mathcal{P}_k$ collect updated estimating function and the re-fitted models
                    \end{itemize}
                \end{itemize}}
                \item The asymptotic variance of {\small $\hat{\boldsymbol{\theta}}_{k}$ is $V_{k}=A^{-1}_{k}\ B_{k}\ (A^{-1}_{k})'$} and an estimator of it is {\small $\hat{V}_{k}=\hat{A}^{-1}_{k}\ \hat{B}_{k}\ (\hat{A}^{-1}_{k})'$}, where
                {\small \begin{itemize}
                    \item $\boldsymbol{\theta}_k$ is the stacked parameter vector consisting of $\theta_k$ and the parameters associated with all indicators appearing in tree $\mathbb T_k$, ordered so that $\theta_k$ is last.
                    \item $\boldsymbol{\Psi}_t$ is the corresponding stacked estimating function, collecting $\Psi_i$ for all indicator $R_i$ appear in $\mathbb T_k$. $\Psi_i$ is retrieved from $\mathcal{E}_k$ if $\Psi_i\in\mathcal{E}_k$, otherwise retrieve from $\mathcal{E}$.
                    \item {$A_{k}\coloneqq\E(\partial\boldsymbol{\Psi}_{k}/\partial \boldsymbol{\theta}_{k})$, $B_{k}\coloneqq \E(\boldsymbol{\Psi}_{k}\ \boldsymbol{\Psi}'_{k})$, $\hat{A}_{k}\coloneqq P_n(\partial\boldsymbol{\Psi}_{k}/\partial \boldsymbol{\theta}_{k}|_{\boldsymbol{\theta}_{k}=\hat{\boldsymbol{\theta}}_{k}})$, and $\hat{B}_{k}\coloneqq P_n(\boldsymbol{\Psi}_{k}\ \boldsymbol{\Psi}'_{k}|_{\boldsymbol{\theta}_{k}=\hat{\boldsymbol{\theta}}_{k}})$}
                \end{itemize}}
                \item The asymptotic variance of $\hat{\theta}_k$ is the bottom-right block of $V_k$, with its estimator given by the corresponding entry of $\hat{V}_k$
                \par\vspace{0.2cm}
                \item Alternatively, $\pi_k$ can be estimated by fitting $R_k$ on its parents $\pa_\G(R_k)$ with weight equals to $\I(\calS_k=1)\,W_k(\hat{\theta}_{\mathcal T_k})$, using flexible machine learning methods.
                \item RETURN $\pi_k(\pa_\G(R_k);\hat{\theta}_k),\Psi_k$.
            \end{itemize}
             \par\vspace{0.1cm}
	\end{algorithmic}
\end{algorithm} 

\begin{algorithm}[H]
\setstretch{1.2}
	\caption{{\small  \textproc{IPW Estimation of target law functionals}$(\G(X, R, X^*), {\mathbb F},{\cal D},M(\tilde{X};\theta))$}}
	\label{alg:est}
	\begin{algorithmic}[1]
		\vspace{0.25cm}		
		\State {\bf Estimation Procedure:} 
            \begin{itemize}
            \setlength{\itemsep}{0.3cm}
                \item $M(\tilde{X}; \theta)$ is a user-specified function of the target law with the same dimension as the target parameter $\theta$, satisfying $\theta$ as the unique solution of $\E(M(\tilde{X}; \theta)) = 0$, where $\tilde{X}\subseteq X$ are the variables involved in its construction
                \item Let $\mathcal{E},\mathcal{P}$ be the output of Algorithm~\ref{alg:propensity}
                \item Let $\tilde{R} = \{R_i \in R : X_i \in \tilde{X}\}$ collect the missingness indicators associated with $\tilde{X}$
                \item Let $\mathcal R$ denote the set of indicators containing $\tilde{R}$ that is closed under inclusion of the selection sets $\{\mathcal S_i\}$. Define $\mathcal R$ as follows:
                \begin{itemize}
                    \item Starting from $A_0=\tilde{R}$
                    \item Define the recursion $A_{\ell+1}\coloneqq \operatorname{cl}(A_\ell)$, where $\operatorname{cl}(A)\;\coloneqq\; A \cup \bigcup_{R_i\in A} \mathcal S_i$
                    \item For $\ell^\ast$ such that $A_{\ell^\ast+1}=A_{\ell^\ast}$, define $\mathcal R \;\coloneqq\; A_{\ell^\ast}$
                \end{itemize}
                \item \textbf{If} $\mathcal R \cap \mathcal{D}\neq \emptyset$: \textit{Fail to identify the target parameter}
                \item \textbf{Else:} define estimating function $\Psi(X^*,R;\theta, \theta_{\mathcal R})$ as in \eqref{eq:algo_Psi_M}, where $\theta_{\mathcal R}\coloneqq (\theta_i)_{R_i\in \mathcal R}$:
                {\small\begin{align}\label{eq:algo_Psi_M}
                    \Psi(X^*,R;\theta, \theta_{\mathcal R})
                    \;\coloneqq\;
                    \I(\mathcal R=1)\,
                    \Big\{
                    \prod_{R_i\in \mathcal R}
                    \pi_i(\pa_\G(R_i);\theta_i)\vert_{\mathcal S_i^r=1}
                    \Big\}^{-1}
                    M(\tilde{X};\theta).
                \end{align}}
                \noindent An estimator $\hat{\theta}$ is obtained by solving $P_n  \Psi(X^*,R;\theta, \hat{\theta}_{\mathcal R})=0$, where $\hat{\theta}_{\mathcal R}$ is retrieved from $\mathcal{P}$
                \item Define $\boldsymbol{\Psi}\coloneqq \{\boldsymbol{\Psi}_i : R_i\in\mathcal R\}$, retrieved from $\mathcal{E}$. Append $\Psi$ as its last element: $\boldsymbol{\Psi} =\{\boldsymbol{\Psi},\Psi\}$
                \item Let $\boldsymbol{\theta}$ collect the parameters in one-to-one correspondence with $\boldsymbol{\Psi}$, with its estimator denoted by $\hat{\boldsymbol{\theta}}_k$
                \item The asymptotic variance of $\hat{\boldsymbol{\theta}}$ is given by $V= A^{-1} B  (A^{-1})'$, with an estimator $\hat{V}=\hat{A}^{-1}\ \hat{B} \ (\hat{A}^{-1})'$, where
                \begin{itemize}
                    \item $A\coloneqq \E(\partial \boldsymbol{\Psi}/\partial\boldsymbol{\theta})$, $B\coloneqq \E(\boldsymbol{\Psi}\ \boldsymbol{\Psi}')$, $\hat{A}\coloneqq P_n(\partial \boldsymbol{\Psi}/\partial\boldsymbol{\theta} |_{\boldsymbol{\theta}=\hat{\boldsymbol{\theta}}})$, and $\hat{B}\coloneqq P_n(\boldsymbol{\Psi}\ \boldsymbol{\Psi}'|_{\boldsymbol{\theta}=\hat{\boldsymbol{\theta}}})$
                \end{itemize}
                \item The asymptotic variance of $\hat{\theta}$ is the bottom right element of $V$, with its estimator at the same entry of $\hat{V}$
                \item RETURN $\hat{\theta}, \hat{V}$
            \end{itemize}
	\end{algorithmic}
\end{algorithm}

\clearpage
\section{Proofs}
\label{app:proofs}

\begin{proof}[Proof of Theorem~\ref{thm:pi_id_functional}]
Let $\mathcal T_k$ be the set of indicators intervened on in $\mathbb T_k$, and let
$\sigma_k=(s_1,\ldots,s_m)$ be an ordering of $\mathcal T_k$ consistent with the reverse topological order used by Algorithm~\ref{alg:ID}. Write $\phi^p_{\sigma_k}\{p\}
\;=\; \phi^p_{s_m}\circ\cdots\circ \phi^p_{s_1}\{p\}$. We emphasize that $\phi^p_{\sigma_k}\{p\}$ need not be identifiable as a full law. The algorithm is only required to identify the conditional kernel 

\vspace{-1.5cm} 
\begin{align*}
\phi^p_{\sigma_k}\{p\}(R_k=1 \mid \pa_\G(R_k))
\quad\text{evaluated at}\quad \mathcal S_k^r=1.
\end{align*}%
\vspace{-1.5cm}

The tree induces a well-defined post-intervention kernel for $R_k$. Since $\mathcal T_k \subseteq R\setminus\{R_k\}$, the intervention sequence encoded by $\mathbb T_k$ does not intervene on $R_k$. By invariance of the missingness mechanism to interventions on indicators other than the target indicator, the propensity score of $R_k$ is unchanged by intervening on $\mathcal T_k$. In particular, as kernels,

\vspace{-1.5cm}
\begin{align}
\phi^p_{\sigma_k}\{p\}(R_k=1 \mid \pa_\G(R_k))
=
p(R_k=1 \mid \pa_\G(R_k))
=
\pi_k(\pa_\G(R_k)),
\label{eq:inv_kernel}
\end{align}%
\vspace{-1.5cm}

on the subset of the sample space where both sides are well-defined.

By assumption, the identification criterion \eqref{eq:id_criteria} holds \emph{for the post-intervention distribution induced by $\mathbb T_k$}, meaning that the conditional kernel
$\phi^p_{\sigma_k}\{p\}(R_k=1 \mid \pa_\G(R_k))$
is identifiable from the observed data law when evaluated at $\mathcal S_k^r=1$. That is,

\vspace{-1.5cm}
\begin{align}
\phi^p_{\sigma_k}\{p\}(R_k=1 \mid \pa_\G(R_k))\big\vert_{\mathcal S_k^r=1}
\end{align}%
\vspace{-1.5cm}

admits an observed-data functional (possibly involving previously identified propensity scores along $\mathbb T_k$) under the model $\mathcal M$. Combining the statements yields $\pi_k(\pa_\G(R_k))\big\vert_{\mathcal S_k^r=1}=\phi^p_{\sigma_k}\{p\}(R_k=1 \mid \pa_\G(R_k))\big\vert_{\mathcal S_k^r=1},$
which is precisely \eqref{eq:pi_id_functional}.
\end{proof}

\newpage
\begin{proof}[Proof of Corollary~\ref{cor:pi_id_ratio}]
Let $\mathcal T_k$ be the children of $R_k$ in $\mathbb T_k$, and let $p_{\mathbb T_k}=\phi^p_{\sigma_k}\{p\}$ be the induced post-intervention law after intervening on all indicators in $\mathcal T_k$. Write $\mathcal T_k=1$ for the event that all indicators in $\mathcal T_k$ equal one. The same convention applies to $\widetilde{\mathcal S}_k$. Repeated application of the fixing formula \eqref{eq:fixing_single} implies that

\vspace{-1.75cm}
\begin{align}
p_{\mathbb T_k}(x,r)
\;\propto\;
\mathbb I(\mathcal T_k=1, \widetilde{\mathcal S}_k=1)\, W_k \, p(x,r),
\label{eq:RN}
\end{align}%
\vspace{-1.5cm} 

where $W_k=\prod_{R_i\in \mathcal T_k}\pi_i(\pa_\G(R_i))^{-1}$ and each $\pi_i$ is evaluated at its admissible evaluation set returned by Algorithm~\ref{alg:ID}, which is ensured by $\mathbb{I}(\widetilde{\mathcal S}_k = 1)$. 
Expectations under $p_{\pi_k}$ can be written as weighted expectations under $p$. Consequently, for any measurable function $f$,

\vspace{-1.5cm}
\begin{align}
\mathbb E_{p_{\mathbb T_k}}[f(X,R)]
=
\frac{\mathbb E\!\left[f(X,R)\,\mathbb I(\mathcal T_k=1,\widetilde{\mathcal S}_k = 1)\,W_k\right]}
{\mathbb E\!\left[\mathbb I(\mathcal T_k=1,\widetilde{\mathcal S}_k = 1)\,W_k\right]},
\end{align}%
\vspace{-1.5cm} 

whenever the expectations exist, and conditioning preserves this weighted representation. Applying Theorem~\ref{thm:pi_id_functional}, we can  write the conditional probability under $p_{\mathbb T_k}$ as 

\vspace{-1.5cm}
\begin{align*}
\pi_k(\pa_\G(R_k))\big\vert_{\mathcal S_k^r=1}
&= p_{\mathbb T_k}(R_k=1 \mid \pa_\G(R_k))\big\vert_{\mathcal S_k^r=1} \\
&= p_{\mathbb T_k}(R_k=1 \mid \pa_\G(R_k), \mathcal T_k=1,\widetilde{\mathcal S}_k=1)\big\vert_{\mathcal S_k^r=1}  \\ 
&= \E_{\mathbb T_k}( \I(R_k=1) \mid \pa_\G(R_k), \mathcal T_k=1 ,\widetilde{\mathcal S}_k=1)\big\vert_{\mathcal S_k^r=1}. 
\end{align*}%
\vspace{-1.75cm}

Substituting the weighted representation \eqref{eq:RN} yields

\vspace{-1.75cm}
\begin{align}
\pi_k(\pa_\G(R_k))\big\vert_{\mathcal S_k^r=1}
=
\frac{
\mathbb E\!\left[
\mathbb I(R_k=1)\,W_k
\;\middle|\;
\pa_\G(R_k),\;\mathcal T_k=1,\; \widetilde{\mathcal S}_k=1
\right]
}{
\mathbb E\!\left[
W_k
\;\middle|\;
\pa_\G(R_k),\;\mathcal T_k=1,\; \widetilde{\mathcal S}_k=1
\right]
}\Big\vert_{\mathcal S^r_k=1},
\end{align}%
\vspace{-1.5cm}

which is \eqref{eq:pi_id_ratio}. If $\mathcal T_k=\emptyset$, then $W_k\equiv 1$ and the expression reduces to the associational identification formula, namely 

\vspace{-1.75cm}
\begin{align*}
\pi_k(\pa_\G(R_k))\big\vert_{\mathcal S_k^r=1}
= p(R_k=1 \mid \pa_\G(R_k)\backslash \mathcal S^x_k,\, \mathcal S^x_k=1). 
\end{align*}%
\end{proof}

\newpage
\begin{proof}[Proof of Theorem~\ref{thm:ps-asymp} (estimating equations for propensity scores)]
The estimating equation for unknown parameter $\theta_k$ indexing $\pi_k$ is formulated as $P_n \Psi_k(X^*,R; \theta_k,\hat{\theta}_{\mathcal T_k})$, with 
\begin{align*}
    \Psi_k(X^*,R;\theta_k,\theta_{\mathcal T_k})
    \;\coloneqq\;
    \I(\widetilde{\mathcal S}_k=1)\,
    W_k(\hat{\theta}_{\mathcal T_k})\,
    f_k(\pa_\G(R_k))\,
    \big\{R_k-\pi_k(\pa_\G(R_k);\theta_k)\big\}
\end{align*}
We establish the validity of this estimating equation by showing that $\E(\Psi_k(X^*,R; \theta_k,\theta_{\mathcal T_k})) = 0$. Let $\sigma_k=(s_1,\ldots,s_m)$ be an ordering of $\mathcal T_k$ consistent with $\tau$.
{\small\begin{align*}
    &\E(\Psi_k(X^*,R; \theta_k,\theta_{\mathcal T_k}))
    \\
    &=\E\Big(\I(\widetilde{\mathcal S}_k=1)\,
    \prod_{R_i\in \mathcal T_k}
    \frac{\I(R_i=1)}{\pi_i(\pa_\G(R_i);\theta_i)\vert_{\mathcal S_i^r=1}}\,
    f_k(\pa_\G(R_k))\,
    \big\{R_k-\pi_k(\pa_\G(R_k);\theta_k)\big\}\Big)
    \\
    &=\E\Big(\frac{\E(\I(s_1=1)\mid O,R\backslash s_1) }{\pi_{s_1}(\pa_\G(s_1);\theta_{s_1})} \I(\widetilde{\mathcal S}_k=1)\!\!\!\!\!\!
    \prod_{R_i\in \mathcal T_k\backslash s_1}
    \frac{\I(R_i=1)}{\pi_i(\pa_\G(R_i);\theta_i)\vert_{\mathcal S_i^r=1}}\,
    f_k(\pa_\G(R_k))\,
    \big\{R_k-\pi_k(\pa_\G(R_k);\theta_k)\big\}\Big)
    \\
    & = \E\Big(\cancelto{1}{\frac{\E(\I(s_1=1)\mid \pa_\G(s_1;\theta_{s_1})) }{\pi_{s_1}(\pa_\G(s_1);\theta_{s_1})}} \I(\widetilde{\mathcal S}_k=1)\!\!\!\!\!\!
    \prod_{R_i\in \mathcal T_k\backslash s_1}
    \frac{\I(R_i=1)}{\pi_i(\pa_\G(R_i);\theta_i)\vert_{\mathcal S_i^r=1}}\,
    f_k(\pa_\G(R_k))\,
    \big\{R_k-\pi_k(\pa_\G(R_k);\theta_k)\big\}\Big)
    \\
    & \vdots
    \\
    & = \E(f_k(\pa_\G(R_k))\,
    \big\{R_k-\pi_k(\pa_\G(R_k);\theta_k)\big\})=0,
\end{align*}}
where $\pi_{s_1}$ and $\theta_{s_1}$ denote the propensity score of $s_1$ and its indexing parameter.

The second equality follows by applying the tower rule, conditioning on $(X, R)$ with $s_1$ excluded in the inner expectation, and the third equality follows from the local Markov property. The omitted steps apply the same argument to the indicators ${s_2}$ through ${s_m}$.
\end{proof}

\bigskip
\begin{proof}[Proof of Theorem~\ref{thm:target-asymp} (estimating equations for functionals of the target law)]

The estimating equation for the target parameter $\theta$ is constructed as $P_n\Psi(X^*, R; \, \theta,\hat\theta_{\mathcal R})=0$, where
\begin{align*}
\Psi(X^*,R;\theta, \theta_{\mathcal R})
\;\coloneqq\;
\I(\mathcal R=1)\,
\Big\{
\prod_{R_i\in \mathcal R}
\pi_i(\pa_\G(R_i);\theta_i)\vert_{\mathcal S_i^r=1}
\Big\}^{-1}
M(\tilde{X};\theta),
\end{align*}
We establish the validity of this estimating equation by showing that $\E(\Psi(X^*, R; \, \theta,\theta_{\mathcal R}))=0$. Let $\sigma=(s_1,\cdots,s_l)$ be an ordering of $\mathcal R$ consistent with $\tau$.
{\small\begin{align*}
    &\E(\Psi(X^*, R; \, \theta,\theta_{\mathcal R}))
    \\
    &=\E\big(
    \I(\mathcal R=1)\,
    \Big\{
    \prod_{R_i\in \mathcal R}
    \pi_i(\pa_\G(R_i);\theta_i)\vert_{\mathcal S_i^r=1}
    \Big\}^{-1}
    M(\tilde{X};\theta)\big)
    \\
    &=\E\big(\frac{\E(\I({s_1}=1)\mid X,R\backslash {s_1})}{\pi_{s_1}(\pa_\G({s_1});\theta_{s_1})\vert_{\mathcal S^r_{s_1}=1}}
    \I(\mathcal R\backslash {s_1}=1)\,
    \Big\{
    \prod_{R_i\in \mathcal R\backslash {s_1}}
    \pi_i(\pa_\G(R_i);\theta_i)\vert_{\mathcal S_i^r=1}
    \Big\}^{-1}
    M(\tilde{X};\theta)\big)
    \\
    &=\E\big(\cancelto{1}{\frac{\E(\I({s_1}=1)\mid \pa_\G({s_1});\theta_{s_1})}{\pi_{s_1}(\pa_\G({s_1});\theta_{s_1})\vert_{\mathcal S^r_{s_1}=1}}}
    \I(\mathcal R\backslash {s_1}=1)\,
    \Big\{
    \prod_{R_i\in \mathcal R\backslash {s_1}}
    \pi_i(\pa_\G(R_i);\theta_i)\vert_{\mathcal S_i^r=1}
    \Big\}^{-1}
    M(\tilde{X};\theta)\big)
    \\
    & \vdots
    \\
    & = \E(M(\tilde{X};\theta))=0,
\end{align*}}
where $\mathcal R^r_{s_1}$ denotes the indicator-induced selection set of $s_1$.
The set $\mathcal R$ includes $\mathcal S^r_{s_1}$ by construction. Consequently, in the third equality, the numerator of the ratio functional, $\E(\I({s_1} = 1) \mid \pa_\mathcal{G}({s_1}); \theta_{s_1})$, is evaluated at $\mathcal{S}^r_{s_1}$ and therefore cancels with the denominator. The omitted steps apply the same argument to the indicators ${s_2}$ to ${s_l}$.
\end{proof}

\clearpage
\section{Simulation details} 
\label{app:sims} 

\subsection{Key identification concepts}\label{appsub:sim-id}
The mDAG in Figure~\ref{fig:sims_tree}(a) represents a MAR model, for which the target law is easily identified. For the remaining three mDAGs, we summarize the key definitions underlying the identification procedure in Appendix Tables~\ref{apptab:sim-tree2},~\ref{apptab:sim-tree3}, and~\ref{apptab:sim-tree4}.
\begin{table}[t]
\centering
\renewcommand{\arraystretch}{1.3}  
\caption{Key definitions used in the tree-based identification algorithm, illustrated using the mDAG in Fig.~\ref{fig:sims_tree}(b).}
\label{apptab:sim-tree2}
\begin{tabular}{lcccccccc}
\toprule
$R_k$ & \textbf{Propensity scores} & $\mathcal{S}^x_k$ & $\R^p_k$ & $\C^{\text{dir}}_{k,k}$ & $\ch_{\T_k}(R_k)$ & $\tilde{\mathcal{S}}_k$ & $\mathcal{S}^r_k$ & $\mathcal{S}_k$ \\
\midrule
$R_1$ & $p( R_1|R_2,R_3)$ & $\emptyset$ & $\emptyset$ & $\emptyset$ & $\emptyset$ & $\emptyset$ & $\emptyset$ & $\emptyset$
\\
$R_2$ & $p( R_2|X_1,X_3)$ & $\{R_1,R_3\}$ & $\{R_1\}$ & $\emptyset$ & $\{R_1\}$ & $\{R_1,R_3\}$ & $\emptyset$ & $\{R_1,R_3\}$
\\
$R_3$ & $p(R_3 |X_1,X_2 )$ & $\{R_1,R_2\}$ & $\{R_1\}$ & $\emptyset$ & $\{R_1\}$ & $\{R_1,R_2\}$ & $\emptyset$ & $\{R_1,R_2\}$
\\
\bottomrule
\end{tabular}
\end{table}
\begin{table}[t]
\centering
\renewcommand{\arraystretch}{1.3}  
\caption{Key definitions used in the tree-based identification algorithm, illustrated using the mDAG in Fig.~\ref{fig:sims_tree}(e).}
\label{apptab:sim-tree3}
\resizebox{\textwidth}{!}{\begin{tabular}{lcccccccc}
\toprule
$R_k$ & \textbf{Propensity scores} & $\mathcal{S}^x_k$ & $\R^p_k$ & $\C^{\text{dir}}_{k,k}$ & $\ch_{\T_k}(R_k)$ & $\tilde{\mathcal{S}}_k$ & $\mathcal{S}^r_k$ & $\mathcal{S}_k$ \\
\midrule
$R_1$ & $p(R_1 |X_4,R_2 )$ & $\{R_4\}$ & $\emptyset$ & $\emptyset$ & $\emptyset$ & $\{R_4\}$ & $\emptyset$ & $\{R_4\}$
\\
$R_2$ & $p( R_2|X_5,R_3)$ & $\{R_5\}$ & $\emptyset$ & $\emptyset$ & $\emptyset$ & $\{R_5\}$ & $\emptyset$ & $\{R_5\}$
\\
$R_3$ & $p(R_3 |R_4,R_5 )$ & $\emptyset$ & $\emptyset$ & $\emptyset$ & $\emptyset$ & $\emptyset$ & $\emptyset$ & $\emptyset$
\\
$R_4$ & $p(R_4 |X_1,R_5 )$ & $\{R_1\}$ & $\{R_1\}$ & $\{R_1\}$ & $\{R_2,R_3\}$ & $\{R_1,R_5\}$ & $\{R_5\}$ & $\{R_1,R_5\}$
\\
$R_4$ & ($R_4\rightarrow R_2$ pruned) & $\{R_1\}$ & $\{R_1\}$ & $\{R_1\}$ & $\{R_3\}$ & $\{R_1,R_5\}$ & $\emptyset$ & $\{R_1\}$
\\
$R_5$ & $p( R_5|X_2 )$ & $\{R_2\}$ & $\{R_2\}$ & $\{R_2\}$ & $\{R_1,R_3,R_4\}$ & $\{R_1,R_2,R_4\}$ & $\emptyset$ & $\{R_2\}$ 
\\
\bottomrule
\end{tabular}
}
\end{table}
\begin{table}[t]
\centering
\renewcommand{\arraystretch}{1.3}  
\caption{Key definitions used in the tree-based identification algorithm, illustrated using the mDAG in Fig.~\ref{fig:sims_tree}(g).}
\label{apptab:sim-tree4}
\resizebox{\textwidth}{!}{
\begin{tabular}{lcccccccc}
\toprule
$R_k$ & \textbf{Propensity scores} & $\mathcal{S}^x_k$ & $\R^p_k$ & $\C^{\text{dir}}_{k,k}$ & $\ch_{\T_k}(R_k)$ & $\tilde{\mathcal{S}}_k$ & $\mathcal{S}^r_k$ & $\mathcal{S}_k$ \\
\midrule
$R_1$ & $p(R_1 | R_2, R_3, R_4, R_5, R_6, R_7, R_8, R_9, R_{10})$ & $\emptyset$ & $\emptyset$ & $\emptyset$ & $\emptyset$ & $\emptyset$ & $\emptyset$ & $\emptyset$
\\
$R_2$ & $p(R_2 | X_1, R_3, R_4, R_5, R_6, R_7, R_8, R_9, R_{10})$ & $\{R_1\}$ & $\{R_1\}$ & $\emptyset$ & $\{R_1\}$ & $\{R_1\}$ & $\emptyset$ & $\{R_1\}$
\\
$R_3$ & $p(R_3 | X_1, X_2, R_4, R_5, R_6, R_7, R_8, R_9, R_{10})$ & $\{R_1,R_2\}$ & $\{R_1,R_2\}$ & $\emptyset$ & $\{R_1,R_2\}$ & $\{R_1,R_2\}$ & $\emptyset$ & $\{R_1,R_2\}$
\\
$R_4$ & $p(R_4 | X_2, X_3, R_5, R_6, R_7, R_8, R_9, R_{10})$ & $\{R_2,R_3\}$ & $\{R_2,R_3\}$ & $\emptyset$ & $\{R_1,R_2,R_3\}$ & $\{R_1,R_2,R_3\}$ & $\emptyset$ & $\{R_2,R_3\}$
\\
$\vdots$ &  &  &  &  &  &  &  & 
\\
$R_{10}$ & $p(R_{10} | X_8, X_9)$ & $\{R_8,R_9\}$ & $\{R_8,R_9\}$ & $\emptyset$ & $\{R_1,\cdots,R_9\}$ & $\{R_1,\cdots,R_9\}$ & $\emptyset$ & $\{R_8,R_9\}$ 
\\
\bottomrule
\end{tabular}
}
\end{table}

\subsection{Data generating processes for simulation study}\label{appsub:dgp}
\subsubsection{Task~1: mean estimation}\label{appsubsub:dgp-sim1}
\vspace{-1cm}

\begin{align*}
&(\G_1 \text{ shown in Appendix Figure~\ref{fig:sims_tree}(a)})  \\
& X_1 \sim \N(0,1), \ X_2 \sim \N(1 - X_1, 1), \ X_3  \sim \N(1 - 2X_2 + 3X_1, 1), \\
& R_2 \sim \mathrm{Binomial}(\expit(2 + X_1)), \ R_3  \sim \mathrm{Binomial}(\expit(1 + 0.5X_1)).
\end{align*}
\par\vspace{-1cm}
\noindent\rule{\linewidth}{0.6pt}
\par\vspace{-1cm}
\begin{align*}
&(\G_2 \text{ shown in Appendix Figure~\ref{fig:sims_tree}(c)})  \\
& X_1 \sim \N(1, 1), \ X_2 \sim \N(3 - 0.6 \lvert X_1 \rvert, 1), \\
& X_3 \sim \N(2 - X_2^2 + 4X_2 + 2X_1X_2, 1.5), \\
& R_1 \sim \mathrm{Binomial}(\expit(1 + R_2 + R_3)), \\
& R_2 \sim \mathrm{Binomial}(\expit(-0.5 X_1 + 0.15 X_3)), \\
& R_3 \sim \mathrm{Binomial}(\expit(3 + 0.5 X_1 - X_2)).
\end{align*}

\par\vspace{-1cm}
\noindent\rule{\linewidth}{0.6pt}
\par\vspace{-1cm}

\begin{align*}
&(\G_3 \text{ shown in Appendix Figure~\ref{fig:sims_tree}(e)})  \\
& X_1 \sim \N(1, 1), \\
& X_2 \sim \N(3 - 0.6 \lvert X_1 \rvert, 1), \\
& X_3 \sim \N(2 - X_2^2 + 4X_2 + 2X_1X_2, 1.5), \\
& X_5 \sim \N(2, 1), \\
& X_4 \sim \N(5X_5^3 - 5 \lvert X_3 \rvert X_5, 1), \\
& R_1 \sim \mathrm{Binomial}(\expit(1.2 + 0.01X_4 + 1.5R_2)), \\
& R_2 \sim \mathrm{Binomial}(\expit(-4 + X_5 + R_3)), \\
& R_3 \sim \mathrm{Binomial}(\expit(-0.8 + 2R_4 + 1.8R_5)), \\
& R_4 \sim \mathrm{Binomial}(\expit(0.3 + 1.5X_1 + 2R_5)), \\
& R_5 \sim \mathrm{Binomial}(\expit(0.8 + 1.5X_2)).
\end{align*}

\par\vspace{-1cm}
\noindent\rule{\linewidth}{0.6pt}
\par\vspace{-1cm}

{\small\begin{align*}
&(\G_4 \text{ shown in Appendix Figure~\ref{fig:sims_tree}(g)})  \\
& X_1 \sim \N(1, 1), \ X_2 \sim \N(3 - 0.6 \lvert X_1 \rvert, 1), \ X_3 \sim \N(2 - X_2^2 + 4X_2 + 2X_1X_2, 1.5), \\
& X_4 \sim \N(1 + X_3 - 0.5X_2, 1), \ X_5 \sim \N(1 + 0.9X_4 - 0.4X_3, 1), \ X_6 \sim \N(1 + 0.8X_5 - 0.3X_4, 1), \\
& X_7 \sim \N(1 + 0.7X_6 - 0.2X_5, 1), \ X_8 \sim \N(1 + 0.7X_7 - 0.2X_6, 1), \\
& X_9 \sim \N(1 + 0.7X_8 - 0.2X_7, 1), \ X_{10} \sim \N(1 + 0.7X_9 - 0.2X_8, 1), \\
& R_1 \sim \mathrm{Binomial}(\expit(-0.1 + 0.1R_2 - 0.1R_3 + 0.1R_4 - 0.1R_5 + 0.1R_6 - 0.1R_7 + 0.1R_8 - 0.1R_9 + 0.1R_{10})), \\
& R_2 \sim \mathrm{Binomial}(\expit(0.1 + 0.3X_1 - 0.1R_3 + 0.1R_4 - 0.1R_5 + 0.1R_6 - 0.1R_7 + 0.1R_8 - 0.1R_9 + 0.1R_{10})), \\
& R_3 \sim \mathrm{Binomial}(\expit(0.1 - 0.3X_2 + 0.3X_1 - 0.1R_4 + 0.1R_5 - 0.1R_6 + 0.1R_7 - 0.1R_8 + 0.1R_9 - 0.1R_{10})), \\
& R_4 \sim \mathrm{Binomial}(\expit(10 - X_3 + 0.2X_2 - 0.1R_5 + 0.1R_6 - 0.1R_7 + 0.1R_8 - 0.1R_9 + 0.1R_{10})), \\
& R_5 \sim \mathrm{Binomial}(\expit(10 - X_4 + 0.2X_3 - 0.1R_6 + 0.1R_7 - 0.1R_8 + 0.1R_9 - 0.1R_{10})), \\
& R_6 \sim \mathrm{Binomial}(\expit(2 - X_5 + X_4 - 0.1R_7 + 0.1R_8 - 0.1R_9 + 0.1R_{10})), \\
& R_7 \sim \mathrm{Binomial}(\expit(2 - X_6 + X_5 - 0.1R_8 + 0.1R_9 - 0.1R_{10})), \\
& R_8 \sim \mathrm{Binomial}(\expit(2 - 2X_7 + 2X_6 - 0.1R_9 + 0.1R_{10})), \\
& R_9 \sim \mathrm{Binomial}(\expit(2 - 2X_8 + 2X_7 - 0.1R_{10})), \\
& R_{10} \sim \mathrm{Binomial}(\expit(2 - X_9 + X_8)).
\end{align*}}

\subsubsection{Task~2: parametric regression}\label{appsubsub:dgp-sim2}
For the current task, the DGP from Subsection~\ref{appsubsub:dgp-sim1} is modified by setting the coefficients on $X_1$ to zero in the conditional distribution of $X_3$ given $X_1$ and $X_2$. Specifically, we have the following.
\begin{align*}
    & X_3  \sim \N(1 - 2X_2 + 0X_1, 1), \quad (\G_1 \text{ shown in Appendix Figure~\ref{fig:sims_tree}(a)}) \\
    & X_3 \sim \N(2 - X_2^2 + 0X_1X_2, 1.5), \quad (\G_2-\G_4 \text{ shown in Appendix Figure~\ref{fig:sims_tree}(c,e,g)})
\end{align*}

\subsubsection{Task~3: causal effect estimation}\label{appsubsub:dgp-sim3}
For the current task, the DGP from Subsection~\ref{appsubsub:dgp-sim1} is modified by generating $X_2$ as a binary variable, with details specified below.
\begin{align*}
& X_2 \sim \mathrm{Binomial}(\expit(1 - X_1)), \quad (\G_1 \text{ shown in Appendix Figure~\ref{fig:sims_tree}(a)}) \\
& X_2 \sim \mathrm{Binomial}(\expit(3 - 0.6 \lvert X_1 \rvert)), \quad (\G_2-\G_4 \text{ shown in Appendix Figure~\ref{fig:sims_tree}(c,e,g)}).
\end{align*}

\begin{figure}[!t]
	\begin{center}
		\scalebox{0.75}{
			\begin{tikzpicture}[>=stealth, node distance=1.7cm, 
				decoration = {snake, pre length=3pt,post length=7pt,},
				every text node part/.style={align=center}]
				\tikzstyle{format} = [thick, circle, minimum size=1.0mm, inner sep=0pt]
				\begin{scope}[xshift=-10cm]
					\path[->, thick]
					node[format] (x11) {$X_1$}
					node[format, right of=x11, xshift=0.5cm] (x21) {$X_2$}
					node[format, right of=x21, xshift=0.5cm] (x31) {$X_3$}
					
					node[format, below of=x21] (r2) {$R_2$}
					node[format, below of=x31] (r3) {$R_3$}

					node[format, below of=r2, yshift=0.25cm] (x2) {{$X^*_2$}}
					node[format, below of=r3, yshift=0.25cm] (x3) {{$X^*_3$}}
					
					(x11) edge[blue, bend left=0] (x21)
					(x21) edge[blue, bend left=0] (x31)
					(x11) edge[blue, bend left=20] (x31)

					(x11) edge[blue] (r2)
					(x11) edge[blue] (r3)
					
					(x21) edge[gray, gray, bend right=20] (x2)
					(x31) edge[gray, gray, bend left=20] (x3)
					
					(r2) edge[gray, gray] (x2)
					(r3) edge[gray, gray] (x3)
					node[format, below of=x2, yshift=0.8cm, xshift=1cm] (a) {(a) $\G_1$} ;
				\end{scope}
				\begin{scope}[xshift=3.5cm, yshift=-1cm]
					\path[->, thick]
					node[format] (id) {$R$}
					node[format, below right of=id, xshift=-0.75cm] (r3) {$R_3$ }
					node[format, below left of=id, xshift=0.75cm] (r2) {$R_2$}
					
					(id) edge[black, -] (r2)
					(id) edge[black, -] (r3)

					node[format, below of=id, yshift=-1.5cm, xshift=0.9cm] (b) {(b) ${\cal F}_1$ corresponding to $\G_1$} ;
					
					;
				\end{scope}

				\begin{scope}[xshift=-10cm, yshift=-6.5cm]
					\path[->, thick]
					node[format] (x11) {$X_1$}
					node[format, right of=x11, xshift=0.5cm] (x21) {$X_2$}
					node[format, right of=x21, xshift=0.5cm] (x31) {$X_3$}

                    node[format, below of=x11] (r1) {$R_1$}
					node[format, below of=x21] (r2) {$R_2$}
					node[format, below of=x31] (r3) {$R_3$}
					
					node[format, below of=r1, yshift=0.25cm] (x1) {{$X^*_1$}}
					node[format, below of=r2, yshift=0.25cm] (x2) {{$X^*_2$}}
					node[format, below of=r3, yshift=0.25cm] (x3) {{$X^*_3$}}
					
					(x11) edge[blue, bend left=0] (x21)
					(x21) edge[blue, bend left=0] (x31)
					(x11) edge[blue, bend left=20] (x31)

					(x11) edge[blue] (r2)
					(x31) edge[blue] (r2)
                    (x11) edge[blue] (r3)
					(x21) edge[blue] (r3)
                    (r2) edge[blue] (r1)
                    (r3) edge[blue, bend left=20] (r1)
					
					(x21) edge[gray, gray, bend right=20] (x2)
					(x31) edge[gray, gray, bend left=20] (x3)

                    (r1) edge[gray, gray] (x1)
					(r2) edge[gray, gray] (x2)
					(r3) edge[gray, gray] (x3)
					node[format, below of=x2, yshift=0.8cm, xshift=1cm] (a) {(c) $\G_2$} ;
				\end{scope}
				\begin{scope}[xshift=3.5cm, yshift=-6.5cm]
					\path[->, thick]
					node[format] (id) {$R$}
					node[format, below left of=id, xshift=0.75cm] (r1) {$R_1$}
					node[format, below right of=id, xshift=-0.75cm] (r2) {$R_2$}
					node[format, below right of=id, xshift=0.25cm] (r3) {$R_3$}

                    node[format, below left of=r2, xshift=1.25cm] (r2-r1) {$R_1$}
					node[format, below left of=r3, xshift=1.25cm] (r3-r1) {$R_1$}

                    (id) edge[black, -] (r1)
					(id) edge[black, -] (r2)
					(id) edge[black, -] (r3)

                    (r2) edge[blue] (r2-r1)
                    (r3) edge[blue] (r3-r1)

					node[format, below of=id, yshift=-2.3cm, xshift=0.9cm] (b) {(d) ${\cal F}_2$ corresponding to $\G_2$} ;
					
					;
				\end{scope}


					node[format, belo\begin{scope}[xshift=-10cm, yshift=-13cm]
					\path[->, thick]
					node[format] (x11) {$X_1$}
					node[format, right of=x11] (x21) {$X_2$}
					node[format, right of=x21] (x31) {$X_3$}
					node[format, right of=x31] (x41) {$X_4$}
					node[format, right of=x41] (x51) {$X_5$}
					
					node[format, below of=x11] (r1) {$R_1$}
					node[format, below of=x21] (r2) {$R_2$}
					node[format, below of=x31] (r3) {$R_3$}
					node[format, below of=x41] (r4) {$R_4$}
					node[format, below of=x51] (r5) {$R_5$}
					
					node[format, below of=r1] (x1s) {$X^*_1$}
					node[format, below of=r2] (x2s) {$X^*_2$}
					node[format, below of=r3] (x3s) {$X^*_3$}
					node[format, below of=r4] (x4s) {$X^*_4$}
					node[format, below of=r5] (x5s) {$X^*_5$}
					
					(x11) edge[blue] (x21)
					(x21) edge[blue] (x31)
					(x11) edge[blue, bend left] (x31)
                    (x31) edge[blue] (x41)
                    (x51) edge[blue] (x41)
					
					(x21) edge[blue] (r5)
					(x11) edge[blue] (r4)
					(x51) edge[blue] (r2)
					(x41) edge[blue] (r1)
					
					(r5) edge[blue] (r4)
					(r5) edge[blue, bend left] (r3)
                    (r4) edge[blue] (r3)
					(r3) edge[blue] (r2)
					(r2) edge[blue] (r1)
					
					(r1) edge[gray] (x1s)
					(x11) edge[gray, bend right=25] (x1s)
					(r2) edge[gray] (x2s)
					(x21) edge[gray, bend right=25] (x2s)
					(r3) edge[gray] (x3s)
					(x31) edge[gray, bend right=25] (x3s)
					
					(r4) edge[gray] (x4s)
					(x41) edge[gray, bend left=25] (x4s)
					(r5) edge[gray] (x5s)
					(x51) edge[gray, bend left=25] (x5s)
					
					node[format, below of=x3s, yshift=0.8cm, xshift=-0.1cm] (e) {(e) $\G_3$} ;
					;
				\end{scope}
				\begin{scope}[xshift=4.5cm, yshift=-12cm]
					\path[->, thick]
					node[format] (id) {$R$}
					node[format, below left of=id] (r3) {$R_3$ }
					node[format, left of=r3, xshift=0.5cm] (r2) {$R_2$}
					node[format, left of=r2, xshift=0.5cm] (r1) {$R_1$}
					node[format, right of=r3, xshift=-0.5cm] (r4) {$R_4$}
					node[format, right of=r4, xshift=0.cm] (r5) {$R_5$}
					
					node[format, below left of=r4, xshift=0.75cm] (r4-r2) {$R_2$}
                    node[format, below right of=r4, xshift=-0.75cm] (r4-r3) {$R_3$}
					
					node[format, below left of=r5, xshift=0.75cm] (r5-r1) {$R_1$}
					node[format, below right of=r5, xshift=-0.75cm] (r5-r3) {$R_3$}
					node[format, below right of=r5, xshift=0.25cm] (r5-r4) {$R_4$}
					
					node[format, below left of=r5-r4, xshift=0.75cm] (r5-r4-r2) {$R_2$}
                    node[format, below right of=r5-r4, xshift=-0.75cm] (r5-r4-r3) {$R_3$}

					(id) edge[black, -] (r1)
					(id) edge[black, -] (r2)
					(id) edge[black, -] (r3)
					(id) edge[black, -] (r4)
					(id) edge[black, -] (r5)
					
					(r4) edge[blue] (r4-r2)
                    (r4) edge[blue] (r4-r3)
					
					(r5) edge[blue] (r5-r4)
					(r5) edge[blue] (r5-r3)
					(r5) edge[blue] (r5-r1)
					
					(r5-r4) edge[red, decorate] (r5-r4-r2)
                    (r5-r4) edge[blue] (r5-r4-r3)
					
					node[format, below of=r5-r4, yshift=-1.2cm, xshift=-3cm] (f) {(f)  ${{\cal F}_3}$ corresponding to $\G_3$} ; 
					
					;
				\end{scope}
			

				\begin{scope}[xshift=-10cm, yshift=-20cm]
                    \path[->, thick, node distance=1.5cm]
                    node[format] (x11) {$X_1$}
                    node[format, right of=x11] (x21) {$X_2$}
                    node[format, right of=x21] (x31) {$X_3$}
                    node[format, right of=x31] (x41) {$X_4$}
                    node[red, right of=x41, xshift=-0.3cm] (xmid1) {$\cdots$}
                    node[format, right of=xmid1, xshift=-0.3cm] (x91) {$X_9$}
                    node[format, right of=x91] (x101) {$X_{10}$}
                    
                    node[format, below of=x11] (r1) {$R_1$}
                    node[format, below of=x21] (r2) {$R_2$}
                    node[format, below of=x31] (r3) {$R_3$}
                    node[format, below of=x41] (r4) {$R_4$}
                    node[red, below of=xmid1] (rmid) {$\cdots$}
                    node[format, below of=x91] (r9) {$R_9$}
                    node[format, below of=x101] (r10) {$R_{10}$}
                    
                    node[format, below of=r1] (x1s) {$X^*_1$}
                    node[format, below of=r2] (x2s) {$X^*_2$}
                    node[format, below of=r3] (x3s) {$X^*_3$}
                    node[format, below of=r4] (x4s) {$X^*_4$}
                    node[red, below of=rmid] (xmids) {$\cdots$}
                    node[format, below of=r9] (x9s) {$X^*_9$}
                    node[format, below of=r10] (x10s) {$X^*_{10}$}
                    
                    (x11) edge[blue] (x21)
                    (x11) edge[blue, bend left] (x31)
                    (x21) edge[blue] (x31)
                    (x21) edge[blue, bend left] (x41)
                    (x31) edge[blue] (x41)
                    (x31) edge[blue, bend left] (xmid1)
                    (x41) edge[blue] (xmid1)
                    (xmid1) edge[blue] (x91)
                    (x91) edge[blue] (x101)
                    
                    (x11) edge[blue] (r2)
                    (x11) edge[blue] (r3)
                    (x21) edge[blue] (r3)
                    (x21) edge[blue] (r4)
                    (x31) edge[blue] (r4)
                    (x31) edge[blue] (rmid)
                    (x91) edge[blue] (r10)
                    
                    (r2) edge[blue] (r1)
                    (r3) edge[blue, bend left=15] (r1)
                    (r4) edge[blue, bend left=30] (r1)
                    (rmid) edge[blue, bend left=30] (r1)

                    (r3) edge[blue] (r2)
                    (r4) edge[blue, bend left=15] (r2)
                    
                    (r4) edge[blue] (r3)
                    (rmid) edge[blue] (r4)
                    
                    (r10) edge[blue, bend left=15] (rmid)
                    (r10) edge[blue] (r9)
                    (r9) edge[blue] (rmid)
                    
                    (r1) edge[gray] (x1s)
                    (x11) edge[gray, bend right=25] (x1s)
                    (r2) edge[gray] (x2s)
                    (x21) edge[gray, bend right=25] (x2s)
                    (r3) edge[gray] (x3s)
                    (x31) edge[gray, bend right=25] (x3s)
                    (r4) edge[gray] (x4s)
                    (x41) edge[gray, bend right=25] (x4s)
                    (rmid) edge[gray] (xmids)
                    (xmid1) edge[gray, bend left=25] (xmids)
                    (r9) edge[gray] (x9s)
                    (x91) edge[gray, bend left=25] (x9s)
                    (r10) edge[gray] (x10s)
                    (x101) edge[gray, bend left=25] (x10s)
                    
                    node[format, below of=x2s, yshift=0.15cm, xshift=1.7cm] (g) {(g) $\G_4$} ;
                    ;
                \end{scope}
				\begin{scope}[xshift=4.5cm, yshift=-19cm]
					\path[->, thick]
					node[format] (id) {$R$}
					node[format, below left of=id] (r3) {$R_3$ }
					node[format, left of=r3, xshift=0.5cm] (r2) {$R_2$}
					node[format, left of=r2, xshift=0.5cm] (r1) {$R_1$}
					node[format, right of=r3, xshift=0.2cm] (r4) {$R_4$}
					node[red, right of=r4, xshift=0.2cm] (rmid) {$...$}
                    node[format, right of=rmid, xshift=0.2cm] (r10) {$R_{10}$}
					
					node[format, below left of=r2, xshift=1.25cm] (r2-r1) {$R_1$}
                    
                    node[format, below left of=r3, xshift=0.75cm] (r3-r1) {$R_1$}
                    node[format, below right of=r3, xshift=-0.75cm] (r3-r2) {$R_2$}
                    node[format, below left of=r3-r2, xshift=1.25cm] (r3-r2-r1) {$R_1$}

                    node[format, below left of=r4, xshift=0.5cm] (r4-r1) {$R_1$}
                    node[format, below left of=r4, xshift=1.25cm] (r4-r2) {$R_2$}
                    node[format, below right of=r4, xshift=-0.2cm] (r4-r3) {$R_3$}
                    node[format, below left of=r4-r2, xshift=1.25cm] (r4-r2-r1) {$R_1$}
                    node[format, below left of=r4-r3, xshift=0.9cm] (r4-r3-r1) {$R_1$}
                    node[format, below right of=r4-r3, xshift=-0.75cm] (r4-r3-r2) {$R_2$}
                    node[format, below left of=r4-r3-r2, xshift=1.25cm] (r4-r3-r2-r1) {$R_1$}

                    node[red, below left of=rmid, xshift=1.25cm] (rmid-rmid) {$\cdots$}

					node[format, below left of=r10, xshift=0.5cm] (r10-r1) {$R_1$}
                    node[red, below left of=r10, xshift=1.25cm] (r10-rmid) {$\cdots$}
                    node[format, below right of=r10, xshift=-0.5cm] (r10-r9) {$R_9$}
                    node[red, below left of=r10-r1, xshift=1.25cm] (r10-r1-rmid) {$\cdots$}
                    node[red, below left of=r10-r9, xshift=1.25cm] (r10-r9-rmid) {$\cdots$}
					
					(id) edge[black, -] (r1)
					(id) edge[black, -] (r2)
					(id) edge[black, -] (r3)
					(id) edge[black, -] (r4)
					(id) edge[black, -] (rmid)
                    (id) edge[black, -] (r10)
					
					(r2) edge[blue] (r2-r1)
                    (r3) edge[blue] (r3-r1)
                    (r3) edge[blue] (r3-r2)
                    (r3-r2) edge[blue] (r3-r2-r1)
                    (r4) edge[blue] (r4-r1)
                    (r4) edge[blue] (r4-r2)
                    (r4) edge[blue] (r4-r3)
                    (r4-r2) edge[blue] (r4-r2-r1)
                    (r4-r3) edge[blue] (r4-r3-r1)
                    (r4-r3) edge[blue] (r4-r3-r2)
                    (r4-r3-r2) edge[blue] (r4-r3-r2-r1)
                    (rmid) edge[blue] (rmid-rmid)
                    (r10) edge[blue] (r10-r1)
                    (r10) edge[blue] (r10-rmid)
                    (r10) edge[blue] (r10-r9)
                    (r10-r1) edge[blue] (r10-r1-rmid)
                    (r10-r9) edge[blue] (r10-r9-rmid)
					
					node[format, below of=r4-r2-r1, yshift=0cm, xshift=-0.45cm] (h) {(h)  ${\cal F}_4$ corresponding to $\G_4$} ; 
					
					;
				\end{scope}
			\end{tikzpicture}
		}
	\end{center}
	\vspace{-1.25cm}
	\caption{(a), (c), (d), and (e) show the mDAGs used in the simulation study in Section~\ref{sec:sims}. (a) corresponds to a MAR model, while (c), (e), and (g) correspond to MNAR models; the mDAG in (g) is a submodel of the permutation model in \citep{robins97non-a}. To save space, some variables and edges are omitted in (e). Specifically, each $X_i$ receives edges from $X_{i-1}$ and $X_{i-2}$, and each $R_i$ receives edges from $X_{i-1}$, $X_{i-2}$, and $R_j$ for all $j>i$. (b), (d), (f), and (h) are the corresponding visualizations of the constructed trees. In (h), the subtree corresponding to all $R_j$ with $j<i$ is attached under $R_i$, and no pruning is performed.}
	\label{fig:sims_tree}
\end{figure}

\begin{figure}[!t]
	\begin{center}
		\scalebox{0.72}{
			\begin{tikzpicture}[>=stealth, node distance=1.3cm, 
				decoration = {snake, pre length=3pt,post length=7pt,},
				every text node part/.style={align=center}]
				\tikzstyle{format} = [thick, circle, minimum size=1.0mm, inner sep=0pt]
				\begin{scope}[xshift=0cm]
					\path[->, thick]
					node[format] (x11) {$X_1$}
					node[format, right of=x11, xshift=0.5cm] (x21) {$X_2$}
					node[format, right of=x21, xshift=0.5cm] (x31) {$X_3$}
					node[format, below of=x11] (r1) {$R_1$}
					node[format, below of=x21] (r2) {$R_2$}
					node[format, below of=x31] (r3) {$R_3$}
					
					(x11) edge[blue] (r3)
					(x11) edge[blue] (r2)
                    (r3)  edge[blue] (r2)
                    (x31) edge[blue] (r1)
                    (r2)  edge[blue] (r1)
					(x11) edge[blue, dashed, -] (x21)
					(x31) edge[blue, dashed, -] (x21)
                    (x31) edge[blue, dashed, -, bend right=25] (x11)
					
					node[format, below of=r2, yshift=0.4cm] (a) {{\small $(a)$}}  ;
				\end{scope}

                \begin{scope}[xshift=6.cm, yshift=0cm]
					\path[->, thick]
					node[format] (x11) {$X_1$}
					node[format, right of=x11, xshift=0.5cm] (x21) {$X_2$}
					node[format, below of=x11] (r1) {$R_1$}
					node[format, below of=x21] (r2) {$R_2$}
					
					(r1) edge[blue] (r2)
					(x11) edge[blue] (r2)
					(x11) edge[blue, dashed, -] (x21)
					
					
					node[format, below of=r1, xshift=1.2cm, yshift=0.4cm] (b) {{\small $(b)$}}  ;
				\end{scope}

                \begin{scope}[xshift=10cm, yshift=0cm]
					\path[->, thick]
					node[format] (x11) {$X_1$}
					node[format, right of=x11, xshift=0.5cm] (x21) {$X_2$}
					node[format, right of=x21, xshift=0.5cm] (x31) {$X_3$}
                    node[format, right of=x31, xshift=0.5cm] (x41) {$X_4$}
                    node[format, right of=x41, xshift=0.5cm] (x51) {$X_5$}
					node[format, below of=x11] (r1) {$R_1$}
					node[format, below of=x21] (r2) {$R_2$}
					node[format, below of=x31] (r3) {$R_3$}
                    node[format, below of=x41] (r4) {$R_4$}
                    node[format, below of=x51] (r5) {$R_5$}
					
					(x11) edge[blue] (r3)
					(x31) edge[blue] (r4)
                    (x31) edge[blue] (r5)
                    (x51) edge[blue] (r2)
                    (r5)  edge[blue] (r4)
                    (r4)  edge[blue] (r3)
                    (r3)  edge[blue] (r2)
                    (r2)  edge[blue] (r1)
                    (r5)  edge[blue, bend left=25] (r3)
					(x11) edge[blue, dashed, -] (x21)
                    (x21) edge[blue, dashed, -] (x31)
                    (x31) edge[blue, dashed, -] (x41)
                    (x41) edge[blue, dashed, -] (x51)
                    (x31) edge[blue, dashed, -, bend right=25] (x11)
                    (x41) edge[blue, dashed, -, bend right=25] (x11)
                    (x51) edge[blue, dashed, -, bend right=25] (x11)
                    (x41) edge[blue, dashed, -, bend right=25] (x21)
                    (x51) edge[blue, dashed, -, bend right=25] (x21)
                    (x51) edge[blue, dashed, -, bend right=25] (x31)
					
					node[format, below of=r3, yshift=0.4cm] (c) {{\small $(c)$}}  ;
				\end{scope}

                \begin{scope}[xshift=5cm, yshift=-4cm, node distance=1.6cm]
					\path[->, thick]
					node[format] (id) {$R$}
					node[format, below of=id, yshift=0.5cm] (r3) {$R_3$ }
					node[format, left of=r3, xshift=0.5cm] (r2) {$R_2$}
					node[format, left of=r2, xshift=0.5cm] (r1) {$R_1$}
					node[format, right of=r3, xshift=0.5cm] (r4) {$R_4$}
					node[format, right of=r4, xshift=1cm] (r5) {$R_5$}

                    node[format, below left of=r3, xshift=0.6cm] (r3-r1) {$R_1$}
					node[format, below right of=r3, xshift=-0.6cm] (r3-r2) {$R_2$}
                    
					node[format, below left of=r4, xshift=0.5cm] (r4-r1) {$R_1$}
					node[format, below left of=r4, xshift=1.2cm] (r4-r2) {$R_2$}
					node[format, below right of=r4, xshift=-0.3cm] (r4-r3) {$R_3$}
                    node[format, below left of=r4-r3, xshift=0.6cm] (r4-r3-r1) {$R_1$}
                    node[format, below right of=r4-r3, xshift=-0.6cm] (r4-r3-r2) {$R_2$}
					
					node[format, below left of=r5, xshift=0.5cm] (r5-r1) {$R_1$}
					node[format, below right of=r5, xshift=-0.75cm] (r5-r3) {$R_3$}
					node[format, below right of=r5, xshift=0.85cm] (r5-r4) {$R_4$}
					
                    node[format, below left of=r5-r3, xshift=0.6cm] (r5-r3-r1) {$R_1$}
                    node[format, below right of=r5-r3, xshift=-0.6cm] (r5-r3-r2) {$R_2$}
                    
                    node[format, below left of=r5-r4, xshift=0.75cm] (r5-r4-r1) {$R_1$}
					node[format, below right of=r5-r4, xshift=-0.6cm] (r5-r4-r2) {$R_2$}
                    node[format, below right of=r5-r4, xshift=0.2cm] (r5-r4-r3) {$R_3$}
                    node[format, below left of=r5-r4-r3, xshift=0.6cm] (r5-r4-r3-r1) {$R_1$}
                    node[format, below right of=r5-r4-r3, xshift=-0.6cm] (r5-r4-r3-r2) {$R_2$}

					(id) edge[black, -] (r1)
					(id) edge[black, -] (r2)
					(id) edge[black, -] (r3)
					(id) edge[black, -] (r4)
					(id) edge[black, -] (r5)

                    (r3) edge[blue] (r3-r1)
                    (r3) edge[blue] (r3-r2)
					
					(r4) edge[blue] (r4-r1)
                    (r4) edge[blue] (r4-r2)
                    (r4) edge[blue] (r4-r3)
                    (r4-r3) edge[blue] (r4-r3-r1)
                    (r4-r3) edge[blue] (r4-r3-r2)

					(r5) edge[blue] (r5-r1)
                    (r5) edge[blue] (r5-r3)
                    (r5-r3) edge[blue] (r5-r3-r1)
                    (r5-r3) edge[red, decorate] (r5-r3-r2)
                    (r5) edge[blue] (r5-r4)
                    (r5-r4) edge[blue] (r5-r4-r1)
                    (r5-r4) edge[red, decorate] (r5-r4-r2)
                    (r5-r4) edge[blue] (r5-r4-r3)
                    (r5-r4-r3) edge[blue] (r5-r4-r3-r1)
                    (r5-r4-r3) edge[red, decorate] (r5-r4-r3-r2)
					
					node[format, below of=r4-r3-r1, yshift=-0.5cm, xshift=0.5cm] (d) {(d) Constructed forest corresponding to the mDAG in (c)} ; 
					
					;
				\end{scope}
			
			\end{tikzpicture}
		}
	\end{center}
	\vspace{0.5cm}
	\caption{Additional figures referenced in the main manuscript. }
	\label{fig:supp_figs}
\end{figure}
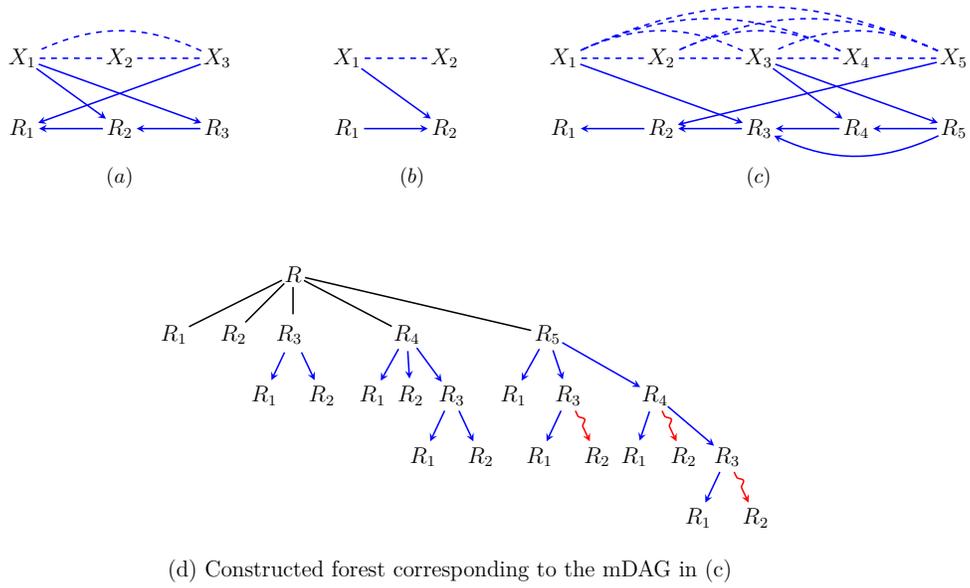

\subsection{Details of the estimating procedure}
\label{appsub:sim-estimation}

\subsubsection{Task 1: mean estimation.}
\label{appsub:sim-estimation-mean}
Estimation under the proposed method begins with estimating the propensity scores, with the fitted propensity score models collected in $\mathcal P$. The mean of $X_3$ is then computed as the empirical inverse-propensity-weighted mean of $X_3^*$, as shown in Equation~\eqref{appeq:sim-mean-est}, with $\theta_i$ retrieved from $\mathcal P$: 
\begin{align}\label{appeq:sim-mean-est}
    &P_n \big(\frac{\I(\mathcal R=1)}{\prod_{R_i\in\mathcal R}\pi_i(\pa_\G(R_i);\hat{\theta}_i)\vert_{\mathcal S_i^r=1}}\ X_3\big).
\end{align}
 For $\mathcal G_1$ through $\mathcal G_4$, the corresponding sets $\mathcal R$ are $\{R_3\}$, $\{R_1,R_2,R_3\}$, $\{R_3\}$, and $\{R_1,R_2,R_3\}$, respectively.

\subsubsection{Task 2: parametric regression.}
\label{appsub:sim-estimation-regression}
For each mDAG, regression coefficients are estimated under a correctly specified model. For $\G_1$, we consider $\E(X_3\mid X_1,X_2)=\beta_0+\beta_1 X_1+\beta_2 X_2$. For $\mathcal G_2$ through $\mathcal G_4$, we consider $\E(X_3\mid X_1,X_2)=\beta_0+\beta_1 X_1 X_2+\beta_2 X_2 + \beta_3 X_2^2$. The regression coefficients can be estimated either via estimating equations or via weighted regressions, with inverse propensity weights defined by the set $\mathcal R$. To be consistent with our implementation with the other three missing-data methods, where we performed regressions using complete or imputed data, we adopt the latter approach for the propsed method. The weighted regressions are fitted using observations with $\mathcal R=1$ and weight $\{
\prod_{R_i \in \mathcal R}
\pi_i\!\left(\pa_{\mathcal G}(R_i);\theta_i\right)
\big|_{\mathcal S_i^r = 1}
\}^{-1}$. For $\G_1$ through $\G_4$, the corresponding sets $\mathcal R$ are $\{R_2,R_3\},\,\{R_1,R_2,R_3\},\,\{R_1,R_2,R_3,R_4,R_5\}$, and $\{R_1,R_2,R_3\}$, respectively.

\textbf{On the unbiasedness of complete-case analysis for mDAGs $\G_3$ and $\G_4$.} 
For $\G_3$, the d-separation criterion implies that $X_3 \perp X_1 \mid X_2, R_1, R_2, R_3$. As a result, complete-case analysis yields unbiased estimation of the regression coefficients.

For $\G_4$, valid regression analysis is conducted under the post-intervention distribution $p_{\mathbb T_4}$, obtained by interventions $R_1=R_2=R_3=1$. Since the regression involves only $X_1$ through $X_3$, we further marginalize $p_{\mathbb T_4}$ over the irrelevant variables $X_4$ through $X_{10}$ and $R_4$ through $R_{10}$. We now show step by step how $p_{\mathbb T_4}$ is derived. First, intervening on $R_1$ yields $p_{\mathbb T_2}$.
\begin{align*}
    p_{\T_2}(X,R) \propto \I(R_1=1) \ \pi^{-1}_1 p(X_1^*,X_2,\cdots,X_{10},R).
\end{align*}
Under $p_{\mathbb T_2}$, $\pi_2$ is identified and can therefore be divided out from $p_{\mathbb T_2}$ to obtain $p_{\mathbb T_3}$:
\begin{align*}
    p_{\T_3}(X,R) \propto \I(R_1=1,R_2=1) \ \pi^{-1}_1 \ \pi^{-1}_2 \ p(X_1^*,X^*_2,\cdots,X_{10},R). 
\end{align*}
A similar argument applies to $R_3$, yielding $p_{\mathbb T_4}$, under which valid regression analysis is performed:
\begin{align*}
    p_{\T_4}(X,R) \propto \I(R_1=1,R_2=1,R_3=1) \ \pi^{-1}_1 \ \pi^{-1}_2 \ \pi^{-1}_3 \ p(X_1^*,X^*_2,X^*_3\cdots,X_{10},R). 
\end{align*}
The distribution $p_{\mathbb T_4}$, marginalized over irrelevant variables, is proportional to the complete-data distribution $\I(R_1=1,R_2=1,R_3=1)\,p(X_1^*,X_2^*,X_3^*)$. As a result, for this mDAG, complete-case analysis yields unbiased estimates of the regression coefficients.

The same argument fails for $\G_2$ because $\pi_3$ is not identified in the post-intervention distribution obtained by intervening on $R_1$ and $R_2$.

\subsubsection{Task 3: causal effect estimation.}
\label{appsub:sim-estimation-causal}
With $\mathcal G_1$, we illustrate the estimation of $\E(X_3^{x_2})$ for $x_2\in\{0,1\}$. The same procedure applies to the other mDAGs. We first fit a weighted regression as described in Subsection~\ref{appsub:sim-estimation-regression}. We then form predictions $\hat{\E}(X_3\mid X_1,x_2)=\hat{\beta}_0+\hat{\beta}_1 X_1+\hat{\beta}_2 x_2$ on observations with $\mathcal R=1$. The target parameter $\E(X_3^{x_2})$ is estimated as the empirical inverse-propensity-weighted mean of $\hat{\E}(X_3\mid X_1,x_2)$, using the same weights as those used to fit the regression.

\end{document}